\documentclass[11pt,a4paper]{article}
\usepackage[a4paper]{geometry}
\usepackage{amssymb,latexsym,amsmath,amsfonts,amsthm,enumerate}
\usepackage{graphicx}
\usepackage{epstopdf}
\usepackage{tikz}
\usepackage{pgflibraryshapes}
\usetikzlibrary{arrows,decorations.markings}
\usepackage{subfigure}
\usepackage{overpic}
\usepackage{fullpage}
\usepackage{comment,verbatim}
\usepackage{hyperref}
\usepackage{color}
\usepackage{mathrsfs}
\usepackage[title,titletoc]{appendix}
\usepackage{ulem}
\usepackage{enumitem}
\usepackage{authblk}

\def\Im {\mathop{\rm Im}\nolimits}
\def\arg {\mathop{\rm arg}\nolimits}
\def\Re {\mathop{\rm Re}\nolimits}

\def\Ai {{\rm Ai}}

\def\tr {{\rm tr}}
\def\PII {{\rm PII}}
\def\PV {{\rm PV}}

\newtheorem{pro}{{Proposition}}
\newtheorem{rem}{{Remark}}
\newtheorem{lem}{{Lemma}}

\newtheorem{thm}{{Theorem}}

\newtheorem{rhp}{RH problem} 

\numberwithin{equation}{section}

\begin{document}

\title{{Asymptotics of the finite-temperature sine kernel determinant
   }}

\author{Shuai-Xia Xu}
\affil[]{Institut Franco-Chinois de l'Energie Nucl\'{e}aire, Sun Yat-sen University, Guangzhou, 510275, China. E-mail: xushx3@mail.sysu.edu.cn}

\date{}

\maketitle

\noindent \hrule width 6.27in \vskip .3cm

\noindent {\bf{Abstract }} 
In the present paper, we study the asymptotics of the Fredholm determinant $D(x,s)$ of the finite-temperature deformation of the sine kernel, which represents the probability that there are no particles in the interval $(-x/\pi,x/\pi)$ in the bulk scaling limit of the  finite-temperature fermion system. The variable $s$ in $D(x,s)$ is related to  the temperature.  This determinant also corresponds to  the finite-temperature correlation function of  one-dimensional Bose gas. We derive the asymptotics of $D(x,s)$ in several different regimes in the $(x,s)$-plane. A third-order phase transition is observed in the asymptotic expansions as both $x$ and $s$ tend to positive infinity at certain related speed. The phase transition is then shown to be described by an integral involving the  Hastings-McLeod solution of the second Painlev\'e equation.


\vspace{.2in}


\vskip .3cm

\noindent \hrule width 6.27in\vskip 1.3cm

\tableofcontents

\section{Introduction }
In this paper, we consider the asymptotics of the Fredholm determinant of the finite-temperature deformation of the classical sine kernel.  Studies of this determinant have  attracted great interest due to their  applications in the finite-temperature fermion system, the Moshe-Neuberger-Shapiro (MNS) random matrix ensemble and  non-intersecting Brownian motions  \cite{DDMS, J, LW, MNS}. It is remarkable that the determinant has also been used to  characterize the finite-temperature correlation function of  one-dimensional Bose gas  in the  pioneering work of Its,  Izergin, Korepin and Slavnov \cite{IIKS}.

To see more explicitly how the determinant arises, we consider the system of $N$  noninteracting spinless  fermions in one dimension, confined in a harmonic  trap.  At zero temperature, it is known that the positions of the $N$ fermions behave statistically like the eigenvalues of  random matrices from the  Gaussian unitary ensemble (GUE); see the review article \cite{DDMS19} and the references therein. In the large $N$ limit, the distribution of the largest eigenvalue of a GUE matrix,  and consequently the rescaled position of the rightmost  fermion, is described by the Tracy-Widom distribution. This distribution  can be expressed in terms of the  Fredholm determinant of the Airy kernel
\begin{equation}\label{eq:KAi}
K_{\Ai}(\lambda,\mu)=\int_{0}^{+\infty}\Ai(\lambda+r)\Ai(\mu+r)dr,
\end{equation}
or,  equivalently,  through the  integral of the  Hastings-McLeod solution of the second Painlev\'e equation \cite{TW}.   
In the bulk  of the spectrum, the statistic of the eigenvalues, and therefore the positions of the fermions, are described by the determinant of the sine kernel
\begin{equation}\label{eq:Ksine}
K_{\sin}(\lambda,\mu)= \frac{\sin\left(\pi(\lambda-\mu)\right)}{\pi(\lambda-\mu)},
\end{equation}
which can be expressed in terms of the $\sigma$-form of the fifth Painlev\'e equation \cite{JMMS}.

 It is discovered recently that,  at finite  temperature, the positions of the fermions can be described by the  MNS random matrix ensemble \cite{DDMS,  LW, MNS}. By analyzing the grand canonical version of the MNS ensemble,  the scaling limits of the correlation kernel at finite temperature are obtained both in the bulk and at the edge of the spectrum in the large $N$ limit.  These limiting kernels depend on the temperature and  generalize the classical sine kernel and Airy kernel.
For instance,  the largest eigenvalue of a MNS matrix, and consequently the position of the rightmost fermion at finite  temperature, 
is given by the determinant of the deformed  Airy kernel
\begin{equation}\label{eq:KAif}
K_{\Ai}^{T}(\lambda,\mu;c)=\int_{-\infty}^{+\infty}\frac{\Ai(\lambda+r)\Ai(\mu+r)}{1+e^{-cr}}dr.
\end{equation}
Here, the parameter $c$ depends explicitly on the temperature.  
In the zero-temperature limit  $c\to+\infty$,  the finite-temperature  Airy kernel is reduced to the classical Airy kernel. 
 It is  discovered by Johansson in \cite{J} that the finite-temperature  Airy kernel determinant  interpolates  between the classical Gumbel distribution and the Tracy-Widom distribution.
Remarkably,  the finite-temperature  Airy kernel determinant characterizes  also the solution of the Kardar-Parisi-Zhang (KPZ) equation with the narrow wedge initial condition \cite{ACQ, CC, KPZ}.
Recently,  the asymptotics of the finite-temperature  deformed Airy kernel determinant have been obtained in several different regimes,  providing useful information on the lower tail asymptotics for the narrow wedge solution of the KPZ equation and the initial data of the determinant solutions to the KdV equation  \cite{CC, CCR, CCR2}.

Similarly, the finite-temperature generalization of the sine kernel is obtained in the bulk  regime \cite{ DDMS,J,  LW}
 \begin{equation}\label{eq:KInt}
K(\lambda,\mu; s)=\int_{0}^{+\infty} \cos\left( \pi (\lambda-\mu)\tau\right)
\sigma(\tau;s)d\tau.
\end{equation}
Here $\sigma(\tau;s)$  denotes the  Fermi distribution
\begin{equation}\label{eq:sigma}
\sigma(\tau;s)=\frac{1}{1+\exp( \tau^2-s )}, \quad \tau\in\mathbb{R}, \quad s\in \mathbb{R}.
\end{equation}
Up to a scaling of variables, the kernel \eqref{eq:KInt} is equivalent to the  finite-temperature sine kernel $K_y^{bulk}(v)$ from \cite[Eq. (27)]{ DDMS}:
$$K(\lambda,\mu; s)=\pi\sqrt{\frac{y}{2} }K_y^{bulk}\left(\pi\sqrt{\frac{y}{2}}(\lambda-\mu)\right).$$
Here, the variable $s$ in \eqref{eq:sigma}  is related to  the inverse temperature parameter $y>0$ in   \cite[Eq. (27)]{ DDMS} through the relation 
\begin{equation}\label{}
e^{s}=e^{y}-1.\end{equation}
As $s\to+\infty$, we  correspondingly  have the low temperature limit $y\to+\infty$; see  \cite[Eq. (27)]{ DDMS} and the discussion following the equation.  Consequently, the probability that no fermions at finite temperature, properly scaled,  lie in the interval $(-\frac{x}{\pi},\frac{x}{\pi})$ is given by the Fredholm determinant
\begin{equation}\label{eq:D0}
D(x, s)=\det(I-K_{\frac{x}{\pi}}),\end{equation}
where $K_{\frac{x}{\pi}}$ is  the integral  operator acting on $L^2(-x/\pi,x/\pi)$ with the kernel \eqref{eq:KInt}.
In the low temperature limit   $s\to+\infty$, we have
\begin{equation}\label{eq:sigmalimit}
\sigma(\sqrt{s}\tau;s)\to \chi_{(-1,1)}(\tau), \quad \tau\in\mathbb{R},
\end{equation}
where $ \chi_{(a, b)}(\tau)$ is the characteristic function on the interval $(a,b)$.
Thus,  the kernel \eqref{eq:KInt} approaches   the classical sine kernel as  $s\to+\infty$  
 \begin{equation}\label{eq:mPIIKernelLim}
\frac{1}{\sqrt{s}}K(\frac{\lambda}{\sqrt{s}},\frac{\mu}{\sqrt{s}}; s)\to \frac{\sin \pi(\lambda-\mu)}{\pi(\lambda-\mu)}.\end{equation}
On the other hand,   the kernel \eqref{eq:KInt} degenerates to the kernel of a uniform Poisson process as $s\to -\infty$  
 \begin{equation}\label{eq:KernelLim0}
e^{-s}K(e^{-s}\lambda, e^{-s}\mu; s)\to \frac{\sqrt{\pi}}{2}\delta(\lambda-\mu),\end{equation}
where $\delta$ is the Dirac delta function. 
Therefore, the process defined by the finite-temperature sine kernel \eqref{eq:KInt}
interpolates between the sine process  and the Poisson process, as observed in \cite{J}. 

 To see how the determinant \eqref{eq:D0} arises in the studies of   one-dimensional Bose gas  \cite{IIKS},  we rewrite \eqref{eq:KInt} in the following form:
  \begin{equation}\label{eq:KIntexp}
K(\lambda,\mu;s)=\frac{1}{2}\int_{\mathbb{R}} \exp\left(-i \pi (\lambda-\mu)\tau\right)
\sigma(\tau;s)d\tau,
\end{equation}
using the fact that $\sigma(\tau;s)$ defined in \eqref{eq:sigma} is an even function.
  Let $K_{x/\pi}$ be the integral operator acting on $L^2(\mathbb{R})$ with the kernel $K(\lambda,\mu;s)$ restricted to the interval 
  $(-x/\pi,x/\pi)$.
We have
\begin{equation}\label{eq:KxDecomp}
K_{x/\pi}=\chi_{(-x/\pi,x/\pi)}\mathcal{F}\sigma  \mathcal{F}^{*} \chi_{(-x/\pi,x/\pi)},
\end{equation}
where $ \mathcal{F} $ and $ \mathcal{F}^{*} $ denote the Fourier  transform and its adjoint 
\begin{equation}\label{eq:Fourier}
 \mathcal{F}(g)(\tau)=\frac{1}{\sqrt{2}}\int_{\mathbb{R}} e^{-i\pi\mu \tau} g(\mu)d\mu, \quad \mathcal{F}^{*}(g)(\tau)=\frac{1}{\sqrt{2}}\int_{\mathbb{R}} e^{i\pi\mu \tau} g(\mu)d\mu.
\end{equation}
Here,  $\chi_{(-x/\pi,x/\pi)}$ and $\sigma$  denote  the multiplication operators with the functions $\chi_{(-x/\pi,x/\pi)}$ and $\sigma$, respectively.
As shown in  \cite[Proposition 3.2]{CT}, both  operators $\chi_{(-x/\pi,x/\pi)}\mathcal{F}\sqrt{\sigma}$ and $\sqrt{\sigma}  \mathcal{F}^{*} \chi_{(-x/\pi,x/\pi)}$ 
are Hilbert-Schmidt operators. From the property $\det(I-T_1T_2)=\det(I-T_2T_1)$ for two Hilbert-Schmidt operators $T_1$ and $T_2$, it follows that
\begin{align}
\det(I- \chi_{(-x/\pi,x/\pi)}\mathcal{F}\sigma  \mathcal{F}^{*} \chi_{(-x/\pi,x/\pi)})
&=\det(I- \sqrt{\sigma} \mathcal{F}^{*}\chi_{(-x/\pi,x/\pi)}\mathcal{F}\sqrt{\sigma}  ).\label{eq:cyclic}
\end{align}
This leads to
 \begin{equation}\label{eq:D}
D(x,s)=\det(I-  K_{\sigma}).
\end{equation}
Here $D(x,s)$ is given in \eqref{eq:D0}.
And $ K_{\sigma}$ is the integral  operator acting on $L^2(\mathbb{R})$ with the kernel 
 \begin{equation}\label{eq:K0}
K_{\sigma}(\lambda,\mu)=\sqrt{\sigma(\lambda;s)}~\frac{\sin\left(x(\lambda-\mu)\right)}{\pi(\lambda-\mu)}\sqrt{\sigma(\mu;s)},
\end{equation}
with  $\sigma(\lambda;s)$ defined in \eqref{eq:sigma}.
Notably, the determinant of \eqref{eq:K0} is  precisely the one investigated  in the pioneering  work  \cite{IIKS} of  Its, Izergin, Korepin and Slavnov on one-dimensional  Bose gas and the theory of  integral Fredholm operators. 
It was shown in \cite{IIKS} that the second derivative of the logarithm of the determinant, 
 \begin{equation}\label{thm:AsyD}
b(x,s)^2=-\partial_x^2\log D(x,s), \end{equation}
satisfies the PDE
\begin{equation}\label{eq:PDEb}
\partial_x\left(\frac{\partial_s\partial_xb}{2b}\right)-\partial_s(b^2)+1=0,
\end{equation}
and that the logarithm of the  determinant $q(x,s)=\log D(x,s)$ satisfies another PDE
\begin{equation}\label{eq:PDEV}
(\partial_s\partial_x^2 q)^2+4(\partial_x^2 q)(2x\partial_s\partial_xq+(\partial_s\partial_x q)^2-2\partial_sq)=0.
\end{equation}
Recently,   Claeys and  Tarricone in  \cite{CT} rederived the PDEs  and showed that the PDEs hold true for a much broader class of
weight functions than the Fermi distribution \eqref{eq:sigma}. Furthermore, it was shown  in  \cite{CT} that the determinant also satisfies a system of integro-differential equations generalizing the fifth Painlev\'e equation.
From the fact that the operator $K_{\sigma}$ has small trace for small $x$,  the initial data of the PDE \eqref{eq:PDEV} was  also derived in \cite{CT} as $x\to0$
\begin{equation}\label{thm:Dsmallx}
 D(x,s)=1-\tr K_{\sigma} +O(x^2)=1-x\int_{\mathbb{R}}\sigma(\lambda;s)d\lambda+O(x^2), 
\end{equation}
with any  fixed $s\in\mathbb{R}$.

Since the  celebrated work of Its, Izergin, Korepin and Slavnov \cite{IIKS}  on the theory of  integrable Fredholm operators, the studies of Fredholm determinants of integrable operators have attracted great interest. In \cite{DIZ},  Deift, Its and Zhou derived the  asymptotics of the  Fredholm of the sine kernel on the union of disjoint intervals as the size of the intervals tends to infinity,  up to an undetermined constant term, with applications in random matrix theory and integrable statistical mechanics. For the one-interval  case, the constant was conjectured to be expressible  
in terms of the derivative of 
 the Riemann zeta-function $\zeta'(-1)$  in  earlier work of Dyson \cite{D} and Widom \cite{W}, 
and was later derived rigorously  by Deift, Its, Krasovsky and Zhou \cite{DIKZ}.  In the two-interval case, this constant  was obtained very recently by Fahs and Krasovsky in \cite{FK}. The asymptotics of the Fredholm determinant of the sine kernel with discontinuities on  consecutive intervals have been derived up to and including the constant  by Charlier in \cite{C21}. The asymptotics of the Fredholm determinants of  the Airy kernel, Bessel kernel, confluent hypergeometric kernel and Pearcey kernel have been explored in several works \cite{BCL, CC20,CM, DXZ21,DXZ22,DZ22, DKV, XD, XZZ,XZ}. Quite recently,  significant progress has been made in the studies of the asymptotics of the finite-temperature deformed Airy kernel determinant in \cite{CC, CCR, CCR2}.

In the present paper, we will study the asymptotics  for the Fredholm determinant \eqref{eq:D} of  the finite-temperature deformed sine kernel \eqref{eq:K0}, the logarithm of which solves the PDE \eqref{eq:PDEV}. This determinant represents the probability that there are no particles in the interval $(-x/\pi,x/\pi)$ in the bulk scaling limit of the  finite-temperature fermion system, as well as the finite-temperature correlation function of  one-dimensional Bose gas. We will derive the asymptotics of the Fredholm determinant \eqref{eq:D} in several different regimes in the $(x,s)$-plane. A third-order phase transition is observed  in the asymptotic expansions for both $x$ and $s$ tending to positive infinity at certain related speed. This phase transition is   shown to be descried by the  Hastings-McLeod solution of the second Painlev\'e equation.  
The main results are stated in the following section.

\section{Statement of results}\label{sec: results}
We derive the asymptotics of the Fredholm determinant $D(x,s)$ defined in  \eqref{eq:D}  when $x$ or $s$ is large \eqref{eq:D} in several different regimes in the $(x,s)$-plane as shown in Theorems \ref{thm: large gap asy x}-\ref{thm:PVasy} below.  The regimes  are illustrated in Fig. \ref{fig:Phase_diagram}.  
\begin{figure}[htb]
\begin{center}
\tikzset{every picture/.style={line width=0.75pt}} 

\begin{tikzpicture}[x=0.75pt,y=0.75pt,yscale=-1,xscale=1]

\draw    (203.43,202.69) -- (444.02,201.76) ;
\draw [shift={(447.02,201.75)}, rotate = 179.78] [fill={rgb, 255:red, 0; green, 0; blue, 0 }  ][line width=0.08]  [draw opacity=0] (6.25,-3) -- (0,0) -- (6.25,3) -- cycle    ;
\draw    (203.43,343.7) -- (204.3,44.66) ;
\draw [shift={(204.31,41.66)}, rotate = 90.17] [fill={rgb, 255:red, 0; green, 0; blue, 0 }  ][line width=0.08]  [draw opacity=0] (6.25,-3) -- (0,0) -- (6.25,3) -- cycle    ;
\draw    (203.43,202.69) .. controls (298.72,171.77) and (348.19,135.35) .. (387,63.7) ;
\draw  [dash pattern={on 4.5pt off 4.5pt}]  (224.24,62.92) .. controls (227.44,88.9) and (221.04,123.53) .. (256.25,181.67) ;
\draw (129,125)  --  (214,70);
\draw [shift={(214,70)}, rotate = 145] [fill={rgb, 255:red, 0; green, 0; blue, 0 }  ][line width=0.08]  [draw opacity=0] (6.25,-3) -- (0,0) -- (6.25,3) -- cycle    ;
\draw (440.4,209.13) node [anchor=north west][inner sep=0.75pt]   [align=left] {x};
\draw (185.93,36.49) node [anchor=north west][inner sep=0.75pt]   [align=left] {s};
\draw (173.13,197.29) node [anchor=north west][inner sep=0.75pt]   [align=left] {$0$};
\draw (366,34) node [anchor=north west][inner sep=0.75pt]   [align=left] {$x=\frac{2}{\pi}\sqrt{s}$};
\draw (247,73) node [anchor=north west][inner sep=0.75pt]   [align=left] {One-gap regime};
\draw (368,131) node [anchor=north west][inner sep=0.75pt]   [align=left] {No-gap regime};
\draw (290,282) node [anchor=north west][inner sep=0.75pt]   [align=left] {No-gap regime};
\draw (79,125) node [anchor=north west][inner sep=0.75pt]   [align=left] {Painlev\'e V regime};
\draw (212,40) node [anchor=north west][inner sep=0.75pt]   [align=left] {$x\sqrt{s}= c$};
\end{tikzpicture}
\end{center}
\caption{Phase diagram showing different types of asymptotics of $D(x,s)$ when $x$ or $s$ is large in the $(x,s)$-plane.   \label{fig:Phase_diagram}
 Near the solid critical curve $x=\frac{2}{\pi}\sqrt{s}$,  the asymptotic behavior exhibits a third-order phase transition; see Remark \ref{rem:PhaseTran}. This transition is  described by  the  Hastings-McLeod solution of the  second Painlev\'e equation similar to the closing of a gap between two soft edges in unitary matrix models \cite{BDJ,CK,CKV}. The asymptotics in the one-gap regime, no-gap regime and the transition regime near the critical curve are given in Theorems \ref{thm: large gap asy x}-\ref{thm:doubleScaling}, respectively. The regime lies on the left of the dashed curve $x\sqrt{s}=c$ corresponds to the  Painlev\'e V regime, where the asymptotic expansion is described by the Painlev\'e V transcendents as $s\to+\infty$ and $x\to 0$ such that $x\sqrt{s}$ remains bounded; see Theorem \ref{thm:PVasy}.
  }
\end{figure}
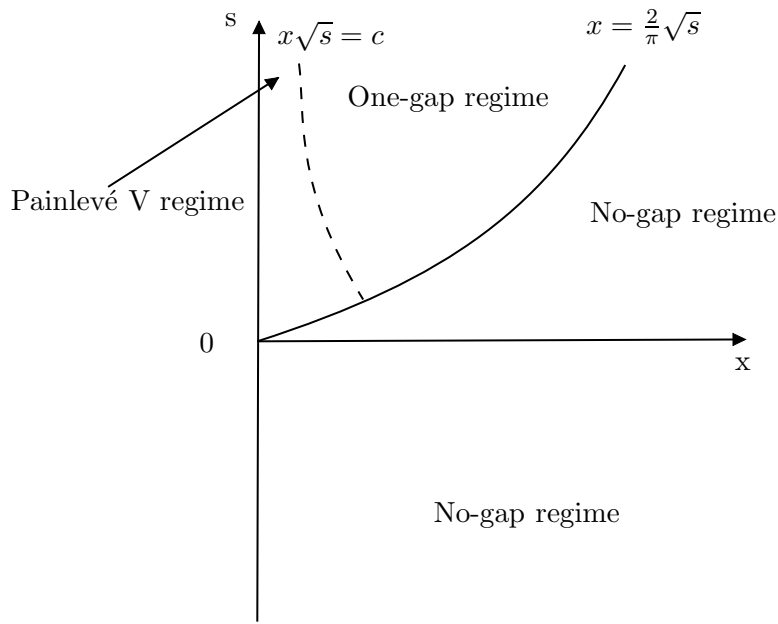

In the no-gap regime lies below the  critical curve $x=\frac{2}{\pi}\sqrt{s}$ as depicted in Fig. \ref{fig:Phase_diagram}, we derive the following  asymptotic expansion of  $D(x,s)$.

\begin{thm}\label{thm: large gap asy x}(Large $x$ asymptotics of  $D(x,s)$ in no-gap regime) Let $D(x,s)$ be defined in \eqref{eq:D}, 
we have the   asymptotic expansion as  $x\to+\infty$
\begin{equation}\label{thm:AsyD}
\log D(x,s)=-\frac{2 x}{\pi}\int_{\mathbb{R}}\sigma(\lambda;s)\lambda^2d\lambda+c_0(s)+O(e^{-cx^2}),\end{equation}
where $c$ is some positive constant. The term $c_0(s)$ is given by
\begin{equation}\label{eq:c0}
c_0(s)=\frac{1}{2\pi^2}\int_{-\infty}^s\left(\int_{\mathbb{R}}\sigma(\lambda;\tau)d\lambda\right)^2d\tau.
\end{equation}
with $\sigma(\lambda;s)$   defined in \eqref{eq:sigma}. The error term is uniform for  
$s\in(-\infty, M_0x^2 ]$  for any constant $0<M_0< \frac{\pi^2}{4}$. 
\end{thm}

\begin{rem}
The formula \eqref{thm:AsyD} is consistent with the one derived in \cite[Eq.(9.9)]{IIKS}  as $x\to+\infty$  for fixed $s$. 
Specifically, for fixed $s$, the formula  \eqref{thm:AsyD} describes the large $x$ asymptotics of  the  probability of there being no free fermions in $(-\frac{x}{\pi}, \frac{x}{\pi})$ at finite temperature. 
At zero temperature, the probability of finding no free fermions in $(-\frac{x}{\pi}, \frac{x}{\pi})$ is characterized  by the classical sine kernel, which has different large gap asymptotics as $x\to+\infty$:
\begin{equation}\label{eq:sinDet}\log \det(I-K_{\sin,x})= -\frac{x^2}{2}-\frac{\log x}{4}+\frac{\log 2}{12}+3\zeta'(-1)+O(1/x),\end{equation}
where $K_{\sin,x}$ denotes the integral operator acting on on $L^2(-x/\pi,x/\pi)$  with the sine kernel  given in \eqref{eq:Ksine}
and $\zeta$ is the Riemann zeta-function; see \cite{DIKZ,D}. 
\end{rem}


In the one-gap regime lies above the  critical curve $x=\frac{2}{\pi}\sqrt{s}$ as shown in Fig. \ref{fig:Phase_diagram}, we have the following  asymptotic expansion of  $D(x,s)$.

\begin{thm}\label{thm: DAsylargs} (Large $x$ and $s$ asymptotics of  $D(x,s)$ in one-gap regime)
As $s\to+\infty$ and $x\to+\infty$, we have  the asymptotic expansion 
\begin{equation}\label{thm:Dasylargs}
\log D(x,s)= -\frac{sx^2}{2}\left(1-\frac{\pi^2}{24}\frac{x^2}{s}\right)+O(s),\end{equation}
where the error term is uniform for $\frac{x}{\sqrt{s}} \in [M_2, \frac{2 }{\pi }-M_1s^{-\frac{2}{3}}]$ with a  sufficiently large  constant $M_1>0$ and any constant $0<M_2<\frac{2 }{\pi }$.
\end{thm}

\begin{rem}\label{rem:PhaseTran}
From Theorems \ref{thm: large gap asy x} and \ref{thm: DAsylargs},  we observe a third-order phase transition near the critical curve  $x=\frac{2}{\pi}\sqrt{s}$ in the  asymptotics of $\log D(x;s)$ as both $x$ and $s$ tend to positive infinity in  such a way  that $x=\frac{2l}{\pi} \sqrt{s} $ with $l>0$.  To see this, we first derive the asymptotics in the regime $l>1$ using Theorem \ref{thm: large gap asy x}. 
It follows from \eqref{eq:sigma} that $\sigma(\sqrt{s}\lambda;s)\to\chi_{(-1,1)}(\lambda)$ as $s\to+\infty$.
Therefore, we have  as $s\to+\infty$
\begin{equation}\label{eq:SigmaAsy1}
\int_{\mathbb{R}}\sigma(\lambda;s) \lambda^2 d\lambda=s^{\frac{3}{2}}\int_{\mathbb{R}}\sigma(\sqrt{s}\lambda;s) \lambda^2 d\lambda
 \sim \frac{2}{3}s^{\frac{3}{2}}, 
\end{equation}
and 
\begin{equation}\label{eq:SigmaAsy2}
c_0(s)  \sim \frac{s^2}{\pi^2},
\end{equation}
with $c_0(s)$ defined in \eqref{eq:c0}.
 Substituting these into \eqref{thm:AsyD}, we have for $l>1$
 \begin{align}\label{eq:Dlimit1}
\log D(x,s) \sim C_+(l)sx^2, ~~\mbox{with} ~~C_+(l)=-\frac{2}{3 l}+\frac{1}{4l^2},
\end{align}
as $s\to+\infty$.
For $0<l<1$, it follows from \eqref{thm:Dasylargs}  that
\begin{align}\label{eq:Dlimit2}
\log D(x,s) \sim C_{-}(l)sx^2, ~~\mbox{with} ~~C_{-}(l) =-\frac{1}{2}+\frac{1}{12}l^2,
\end{align}
as $s\to+\infty$.
We  expand further the coefficients $C_{\pm}(l)$ near $l=1$
\begin{equation}\label{eq:F1}
C_+(l)=-\frac{5}{12}+\frac{1}{6}(l-1)+\frac{1}{12}(l-1)^2-\frac{1}{3}(l-1)^3+O(l-1)^4,
\end{equation}
and
\begin{equation}\label{eq:F2}
C_{-}(l)=-\frac{5}{12}+\frac{1}{6}(l-1)+\frac{1}{12}(l-1)^2. 
\end{equation}
Comparing \eqref{eq:F1} with \eqref{eq:F2}, we observe that the  third derivative of the 
 coefficient of the  leading order asymptotics of  $\log D(x,s)$  is discontinuous at $l=1$. This phenomenon is known   in the literature as  a third-order phase transition \cite{BDJ, GW,  J98}.  
\end{rem}

%

In the transition regime near the critical curve $x=\frac{2}{\pi} \sqrt{s}$,  the asymptotics can be described by  the  Hastings-McLeod solution of the  second Painlev\'e equation
\begin{equation}\label{eq:PII}
u''(x)=xu(x)+2u^3(x),\end{equation} 
subject to the boundary conditions
\begin{equation}\label{eq:PIIAsy}
u(x)\sim \Ai(x), \quad x\to +\infty, \qquad  u(x)\sim \sqrt{-\frac{x}{2}}, \quad x\to -\infty.\end{equation} 

\begin{thm}\label{thm:doubleScaling}
Let $\frac{\pi}{2}\frac{x}{\sqrt{s}}=1+\left(\frac{\pi}{4s} \right)^{\frac{2}{3}}y $, 
we have as $s\to+\infty$
\begin{equation}\label{eq:DasyScaling}
\log D(x,s)=-\frac{2x}{\pi}\int_{\mathbb{R}}\sigma(\lambda;s)\lambda^2d\lambda+c_0(s)-\int^{+\infty}_{y}(\tau-y)u^2(\tau)d\tau+O(s^{-1/6}),
\end{equation} 
where  $c_0$ is given in  \eqref{eq:c0}, and $u(\tau)$ is the Hastings-McLeod solution of the  second Painlev\'e equation.  The error term is uniform for $-M_3\leq y\leq M_4s^{\frac{2}{3}}$ with any constants $M_3>0$ and $M_4>0$. \end{thm}

\begin{rem} The integral involving the  Hastings-McLeod solution in \eqref{eq:DasyScaling} corresponds to the Tracy-Widom distribution \cite{TW}
$$\log F_{TW}(y)=-\int^{+\infty}_{y}(\tau-y)u^2(\tau)d\tau.$$
This Painlev\'e II transition near the critical curve is  similar to the closing of a gap  in unitary matrix models \cite{BDJ,CK,CKV}. 
From \eqref{eq:PIIAsy}, we see that $\log F_{TW}(y)$ tends to zero exponentially fast as $y\to+\infty$, and 
\begin{equation}\label{eq:TWexpand}
\log F_{TW}(y)\sim \frac{1}{12}y^3, \quad y\to-\infty.
\end{equation}
As a consequence,  the asymptotic expansion \eqref{eq:DasyScaling} degenerates to \eqref{thm:AsyD} and also \eqref{eq:Dlimit1} when $y$ grows like $M_4s^{\frac{2}{3}}$ for any constant $M_4>0$.  Futhermore, let $l=\frac{\pi}{2}\frac{x}{\sqrt{s}}=1+\left(\frac{\pi}{4s} \right)^{\frac{2}{3}}y$, we have $ \frac{1}{3}(l-1)^3sx^2 \sim \frac{1}{12}  y^3$ for large $s$. 
It follows from \eqref{eq:Dlimit1}-\eqref{eq:F2} that the difference  between the leading asymptotics of $\log D(x,s)$  in the  regimes separated by the critical curve $\frac{\pi}{2}\frac{x}{\sqrt{s}}=1$ is given approximately by $ \frac{1}{3}(l-1)^3sx^2$. 
 Therefore, the difference in the  leading asymptotics  of $\log D(x,s)$ across the  regimes separated by the critical curve should be compared with  the tail of the Tracy-Widom distribution in \eqref{eq:TWexpand}.
\end{rem}

As $s\to+\infty$ and $x\to 0$ in such a way that $x\sqrt{s}$ remains bounded, we have $\sigma(\sqrt{s}\lambda;s)\to\chi_{(-1,1)}(\lambda)$.
We show that the determinant $D(x,s)$ degenerates to the classical sine kernel determinant, which can be expressed in terms of 
the solution of the $\sigma$-form of the fifth Painlev\'e equation \cite{JMMS}
\begin{equation}\label{eq:sigmaPV}
(xv'')^2+4(4v-4xv'-v'^2)(v-xv')=0,\end{equation} 
with the boundary conditions 
 \begin{equation}\label{eq:PVasy}
 v(x)\sim -\frac{2}{\pi} x, \quad x\to 0, \qquad v(x)\sim -x^2-\frac{1}{4}, \quad x\to+\infty.
\end{equation}

\begin{thm}\label{thm:PVasy}As $s\to+\infty$ and $x\to 0$ such that $\sqrt{s}x$ remains bounded,   we have 
\begin{equation}\label{eq:DPV1}
\log D(x,s)=\log \det(I-K_{\sin,\sqrt{s}x})+O\left(\frac{x}{s^2}\right), \end{equation}
or equivalently 
\begin{equation}\label{eq:DPV}
\log D(x,s)=\int_{0}^{\sqrt{s}x}\frac{v(\tau)}{\tau}d\tau+
O\left(\frac{x}{s^2}\right), \end{equation}
where $K_{\sin,x}$ denotes the integral operator acting on $L^2(-x/\pi,x/\pi)$ with the sine kernel given in \eqref{eq:Ksine}  
  and $v(\tau)$ solves the $\sigma$-form of the fifth Painlev\'e equation \eqref{eq:sigmaPV} with the boundary conditions 
\eqref{eq:PVasy}.
\end{thm}

\begin{rem} From \eqref{eq:PVasy} and  \eqref{eq:DPV}, we  formally  obtain  the leading-order term $-\frac{1}{2}sx^2$ in the  asymptotics of  $\log D(x,s)$ as both $s$ and  $\sqrt{s}x$ tend to positive infinity, which is consistent with 
 the  large gap asymptotics of   the classical sine kernel determinant shown in \eqref{eq:sinDet}.  The leading  term is also consistent with \eqref{thm:Dasylargs} after  formally taking  the limit $l=\frac{\pi}{2}\frac{x}{\sqrt{s}}\to 0$. We will study these regimes rigorously  in a separate paper.
\end{rem}

\begin{rem} We consider in the present paper  the asymptotics of the determinant of the finite-temperature sine kernel
with the special Fermi weight defined by \eqref{eq:sigma}, which  describes  precisely the gap probability of the  finite-temperature free fermions in the bulk scaling regime.  It is expected that the results of the paper would hold for more general class of weight functions. Particularly,
it would be interesting to extend the results to the weight function, considered in \cite{CT}, of the form $w(\lambda^2-s)$ with some function $w$ 
 sufficiently smooth  and decaying sufficiently fast at positive infinity. We will leave this problem to  a further investigation.
 \end{rem}

The rest of the paper is organized as follows. In Section \ref{RHforDet}, we express the logarithmic derivatives of the 
determinant \eqref{eq:D} with respect to $x$ and $s$ in terms of the solution of a Riemann-Hilbert (RH) problem by  using   the theory of  integrable  operators developed in \cite{IIKS, DIZ}.  In Sections \ref{sec: large x}-\ref{sec:PVasy}, we perform the Deift-Zhou nonlinear steepest descent analysis \cite{DMVZ1, DMVZ2, DZ} of the RH problem in several different regimes. 
The proofs of Theorems \ref{thm: large gap asy x}-\ref{thm:PVasy} are then provided in Section \ref{sec:ProofTheorems} by using the results obtained in Sections \ref{sec: large x}-\ref{sec:PVasy}.

\section{RH problem for the finite-temperature sine kernel determinant }\label{RHforDet}
In this section, we relate the determinant \eqref{eq:D} to the solution of a RH problem  using 
the general theory of integrable operators developed in \cite{IIKS, DIZ}.

For this purpose, we define 
\begin{equation}\label{def:fh}
\mathbf{ f}(\lambda)=\sqrt{\sigma(\lambda;s)}\left(
                               \begin{array}{c}
                                 e^{ix\lambda}  \\
                                 e^{-ix\lambda}\\
                                 \end{array}
                             \right), \qquad  \mathbf{ h}(\mu)=\frac{\sqrt{\sigma(\mu;s)}}{2\pi i}\left(
                               \begin{array}{c}
                                  e^{-ix\mu} \\
                                 - e^{ix\mu} \\
                                 \end{array}
                             \right),
\end{equation}
with $\sigma(\lambda;s)$ defined in \eqref{eq:sigma}.
Then, the kernel \eqref{eq:K0} can be expressed as 
\begin{equation}\label{eq:K}
K_{\sigma}(\lambda,\mu)=\frac{\mathbf{ f}(\lambda)^{t}\mathbf{ h}(\mu)}{\lambda-\mu},
\end{equation} 
where $\mathbf{ f}(\lambda)^{t}$ denotes the transpose of the vector $\mathbf{ f}(\lambda)$.
According to \cite{IIKS, DIZ}, the  logarithmic derivative of the determinant \eqref{eq:D} can be expressed in terms of the solution of the following RH problem.

\begin{rhp} \label{RHY}
The function $Y(\lambda)=Y(\lambda; x,s)$  satisfies the properties.
\begin{itemize}
\item[\rm (1)] 
$Y(\lambda)$ is analytic in
  $\mathbb{C}\backslash \mathbb{R}$.

\item[\rm (2)]   $Y(\lambda)$  has continuous boundary values $Y_+(\lambda)$ and $Y_-(\lambda)$ as $\lambda$ approaches the real axis
from the positive and negative sides, respectively.  And 
they satisfy the jump relation
 \begin{equation}\label{eq:YJump}Y_+(\lambda)
 =Y_-(\lambda) \left(
                               \begin{array}{cc}
                              1-\sigma( \lambda;s)&   e^{2ix\lambda}\sigma( \lambda;s)\\
                                 -  e^{-2ix\lambda} \sigma( \lambda;s) &  1+ \sigma( \lambda;s)\\
                                 \end{array}
                             \right),
\qquad \lambda\in \mathbb{R},
\end{equation}
where $\sigma(\lambda;s)$ is defined in \eqref{eq:sigma}.

\item[\rm (3)]  The behavior of $Y(\lambda)$ at infinity is
  \begin{equation}\label{eq:YInfinity}Y(\lambda)=I+\frac {Y_{1}}{\lambda}+O\left (\frac 1 {\lambda^2}\right).\end{equation}
\end{itemize}
\end{rhp}

\begin{pro}\label{pro:diffIdentity}  
We have the differential identities
\begin{equation}\label{eq:DY}
\partial_x\log D(x,s)=-2i (Y_1)_{11}, 
\end{equation}
and 
\begin{equation}\label{eq:BC2}
\partial_s \log D(x,s)=-\int_{\mathbb{R}}  \mathbf{ h}^{t}(\lambda)Y_{+}^{-1}(\lambda)\frac{d}{d\lambda}(Y_{+}(\lambda)\mathbf{f}(\lambda)) \partial_s \log \sigma (\lambda;s)d\lambda,\end{equation}
where $\mathbf{ f}$ and $\mathbf{ h}$ are defined in  \eqref{def:fh} and $Y_1$ is the coefficient in the large $\lambda$ expansion of $Y(\lambda)$ in \eqref{eq:YInfinity}. 
\end{pro}

\begin{proof}
According to \cite{IIKS}, there exists a unique solution to the above RH problem 
\begin{equation}\label{eq:Ysolution}Y(\lambda)=I-\int_{\mathbb{R} } \mathbf{ F}(\mu)\mathbf{ h}^{t}(\mu)\frac{d\mu}{\mu-\lambda},\end{equation}
where
  \begin{equation}\label{eq:YFH}
  \mathbf{ F}(\mu)=(I- K_{\sigma})^{-1} \mathbf{ f}(\mu).\end{equation}
  The operator $I-K_{\sigma}$ is invertible as shown in \cite{CT}.
  Moreover, the resolvent $K_{\sigma}(I-K_{\sigma})^{-1}$ is also an integrable operator with the kernel given by
  \begin{equation}\label{eq:Rkernel}
R(\lambda,\mu)=\frac{  \mathbf{F}^t(\lambda)  \mathbf{H}(\mu)}{\lambda-\mu},\end{equation}  
where 
\begin{equation}\label{eq:FH}
  \mathbf{F}(\lambda)=Y_{+}(\lambda)  \mathbf{f}(\lambda), \quad    \mathbf{H}(\mu)=Y_{+}^{-t}(\mu)  \mathbf{h}(\mu),\end{equation}  
with $\lambda,\mu\in\mathbb{R}$; see \cite{IIKS}.

It follows from \eqref{eq:D} that
 \begin{equation}\label{eq:logDF1}
\partial_x\log D(x,s)=-\mathrm{ tr} \left((I-K_{\sigma})^{-1}  \partial_x K_{\sigma}\right).\end{equation}
By \eqref{eq:K}, we have
 \begin{equation}\label{eq:dsin}
\partial_x K_{\sigma}(\lambda,\mu) =i\mathbf{ f}(\lambda)^{t}\sigma_3\mathbf{ h}(\mu)=i(\sigma_3\mathbf{ h}(\mu))^{t}\mathbf{ f}(\lambda),\end{equation} 
where $\sigma_3$ is one of the Pauli matrices:
                     \begin{equation}\label{eq:PauliMatrices}
                     \sigma_1=\left(
                               \begin{array}{cc}
                                 0&1\\
                              1& 0 \\
                                 \end{array}
                             \right),  \quad  \sigma_2=\left(
                               \begin{array}{cc}
                                 0&-i\\
                              i& 0 \\
                                 \end{array}
                             \right), \quad  \sigma_3=\left(
                               \begin{array}{cc}
                                 1&0\\
                              0& -1 \\
                                 \end{array}
                             \right). \end{equation}
Therefore, we have
 \begin{equation}\label{eq:logDF2}
\partial_x \log D(x,s)=-i\int_{\mathbb{R}}(\sigma_3 \mathbf{ h}(\mu))^{t} \mathbf{ F}(\mu) d\mu.\end{equation}
 The differential identity
\eqref{eq:DY} then follows from \eqref{eq:YInfinity}, \eqref{eq:Ysolution}, \eqref{eq:YFH} and \eqref{eq:logDF2}.

Similarly, using \eqref{eq:D},  \eqref{eq:K}, \eqref{eq:YFH} and Jacobi’s identity, we have
 \begin{equation}\label{eq:dslogD}
\partial_s \log D(x,s)=-\mathrm{ tr} \left((I-K_{\sigma})^{-1} \partial_s K_{\sigma}\right)
=-\mathrm{ tr} \left((I-K_{\sigma})^{-1}  K_{\sigma}\partial_s \log \sigma\right). \end{equation}
Therefore, the  differential identity \eqref{eq:BC2}  follows from \eqref{eq:Rkernel}, \eqref{eq:FH} and \eqref{eq:dslogD}.
 
 \end{proof}

\section{Large $x$ asymptotics of $Y$   
}\label{sec: large x}

In this section, we carry out the Deift-Zhou \cite{DMVZ1, DMVZ2, DZ}  nonlinear steepest descent analysis of the RH problem for $Y(\lambda)$ as  $x\to+\infty $ for $s\in(-\infty, M_0x^2 ]$  with any constant $0<M_0< \frac{\pi^2}{4}$.
\subsection{First transformation: $Y\to S$}
The jump matrix in \eqref{eq:YJump} can be factorized in the following way
 \begin{align}\label{eq:TFac}\left(
                               \begin{array}{cc}
                              1-\sigma( x\lambda;s)&   e^{2ix^2\lambda}\sigma(x \lambda;s)\\
                                 -  e^{-2ix^2\lambda} \sigma( x\lambda;s) &  1+ \sigma( x\lambda;s)\\
                                 \end{array}
                             \right)=&\left(
                               \begin{array}{cc}
                                 1 &0\\
                                 -e^{-2ix^2\lambda}\frac{\sigma(x \lambda;s)}{1- \sigma( x\lambda;s)}& 1 \\
                                 \end{array}
                             \right)\\ \nonumber
                            &\times \left( 1-\sigma(x \lambda;s)\right)^{\sigma_3} \left( \begin{array}{cc}
                                 1 & e^{2ix^2\lambda}\frac{\sigma(x \lambda;s)}{1- \sigma( x\lambda;s)}\\
                                0 & 1 \\
                                 \end{array}\right).
                             \end{align}  
Based on the above factorization, we  define 
\begin{equation}\label{eq:S}
S(\lambda)=Y(x\lambda)\left\{        \begin{array}{cc}
\left(
                               \begin{array}{cc}
                                 1 & -e^{2ix^2\lambda}\frac{\sigma(x\lambda;s )}{1- \sigma(x\lambda;s )}\\
                                0 & 1 \\
                                 \end{array}
                             \right), &   \lambda\in \Omega_2, \\
                                                             \left(
                               \begin{array}{cc}
                                 1 &0\\
                                 -e^{-2ix^2\lambda}\frac{\sigma(x\lambda;s )}{1- \sigma(x\lambda;s )}& 1 \\
                                 \end{array}
                             \right), &  \lambda\in \Omega_3,\\
                                                        
                               I,            &  \lambda\in \Omega_1\cup \Omega_4.             
                                                            \end{array}
                      \right.
\end{equation}
The regions are illustrated in Fig. \ref{fig:S}, the boundaries of which will be specified later.
From the function $\sigma(\lambda;s)$ defined in \eqref{eq:sigma}, we have 
 \begin{equation}\label{eq:DeformWelldefine}
\frac{\sigma(\lambda;s)}{1- \sigma(\lambda;s)}=\exp(-\lambda^2+s),
\end{equation}
which is an entire function in the complex plane. 
Therefore, it follows from \eqref{eq:S} that $S(\lambda)$ satisfies the following RH problem.

\begin{rhp} \label{RHPT}
The function  $S(\lambda)$  satisfies the following properties.
\begin{itemize}
\item[\rm (1)]  $S(\lambda)$ is analytic in $\mathbb{C}\setminus \{ \mathbb{R}\cup\Gamma_1\cup\Gamma_2\}$ with the contours shown in Fig. \ref{fig:S}.

 \item[\rm (2)]  $S(\lambda)$  satisfies the jump condition
\begin{equation}\label{eq:Sjump}
S_{+}(\lambda)=S_{-}(\lambda)\left\{        \begin{array}{cc}
\left(
                               \begin{array}{cc}
                                 1 &e^{2ix^2\lambda}\frac{\sigma(x\lambda;s)}{1- \sigma(x\lambda;s)}\\
                                 0 & 1 \\
                                 \end{array}
                             \right), &   \lambda\in \Gamma_1,\\

                       \left( 1-\sigma(x\lambda;s)\right)^{\sigma_3},      &  \lambda\in\mathbb{R},\\
                                                      \left(
    \begin{array}{cc}
                                 1 &0\\
                                 -e^{-2ix^2\lambda}\frac{\sigma(x\lambda;s)}{1- \sigma(x\lambda;s)}& 1 \\
                                 \end{array}
                             \right),&   \lambda\in \Gamma_2. 
                                 \end{array}
                             \right.
\end{equation}
 \item[\rm (3)]   As $\lambda\to\infty$, we have
  \begin{equation}\label{eq:SInfinity}S(\lambda)=I+O\left(\frac{1}{\lambda}\right).\end{equation}                     
  \end{itemize}
  
  \end{rhp}
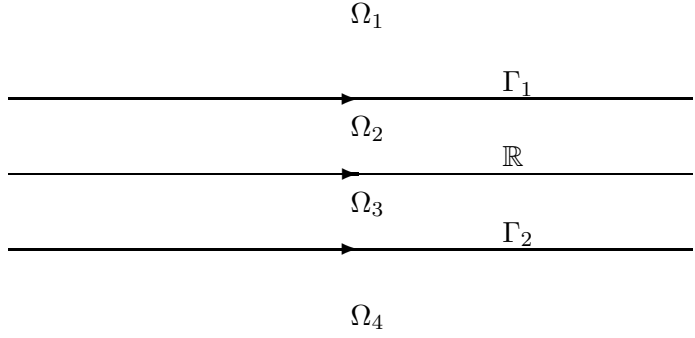
\begin{figure}[htb]
\begin{center}
   \setlength{\unitlength}{1truemm}
   \begin{picture}(100,50)(-5,20)

       \put(25,40){\line(1,0){30}}
          \put(55,40){\line(1,0){30}}
             \put(25,40){\line(-1,0){30}}
  \put(25,50){\line(1,0){30}}
          \put(55,50){\line(1,0){30}}
             \put(25,50){\line(-1,0){30}}
 \put(25,30){\line(1,0){30}}
          \put(55,30){\line(1,0){30}}
             \put(25,30){\line(-1,0){30}}

 \put(40,40){\thicklines\vector(1,0){1}}
  \put(40,50){\thicklines\vector(1,0){1}}
   \put(40,30){\thicklines\vector(1,0){1}}

       \put(60,51){$\Gamma_1$}
   
 \put(60,31){$\Gamma_2$}

 \put(60,41){$\mathbb{R}$}

  \put(40,60){$\Omega_1$}
   \put(40,45){$\Omega_2$}
    \put(40,35){$\Omega_3$}
     \put(40,20){$\Omega_4$}

\end{picture}
   \caption{Jump contours and regions for the RH problem for $S$.}
   \label{fig:S}
\end{center}
\end{figure}

\subsection{Global parametrix}

From  \eqref{eq:DeformWelldefine}, we have 
 \begin{equation}\label{eq:SJ}
e^{\pm 2ix^2\lambda}\frac{\sigma(x\lambda;s)}{1- \sigma(x\lambda;s)}=\exp(-x^2(\lambda\mp i)^2-x^2+s 
).
\end{equation}
Therefore, we may choose  appropriate contours $\Gamma_1$ and $\Gamma_2$ such that the  function in \eqref{eq:SJ} is exponentially small  for $\lambda$ on these contours as $x\to+\infty$. 
Neglecting the jumps along  $\Gamma_1$ and $\Gamma_2$, we arrive at the following approximate RH problem with the only jump  along 
the real axis
\begin{equation}\label{eq:NJump}
N_{+}(\lambda) =N_{-}(\lambda)   \left( 1-\sigma(x\lambda;s)\right)^{\sigma_3},  \quad \lambda\in\mathbb{R},
\end{equation}
and  the asymptotic behavior at infinity
\begin{equation}\label{eq:NInfty}
N(\lambda) =I+O\left(\frac{1}{\lambda}\right).
\end{equation}
It is direct to see that 
\begin{equation}\label{eq:N}
N(\lambda) =d(\lambda)^{\sigma_3}
\end{equation}
with 
\begin{equation}\label{eq:d}
d(\lambda)=\exp\left(\frac{1}{2\pi i} \int_{\mathbb{R}}\log(1-\sigma(x\zeta;s)) \frac{d\zeta}{ \zeta-\lambda}\right).
\end{equation}
If $\Im\lambda>0$, we have
\begin{equation}\label{eq:Int-est}
\Re \left(\frac{1}{2\pi i} \int_{\mathbb{R}}\log(1-\sigma(x\zeta;s))\frac{d\zeta}{ \zeta-\lambda}\right)=
\frac{\Im\lambda}{2\pi} \int_{\mathbb{R}}\log(1-\sigma(x\zeta;s)) \frac{d\zeta}{ |\zeta-\lambda|^2}
<0,
\end{equation}
since $\sigma(x\zeta;s)\in(0,1)$ for $\zeta\in\mathbb{R}$.
Therefore, we find that for  $\Im\lambda>0$
\begin{equation}\label{eq:d-est1}
|d(\lambda)|<1.
\end{equation}
Similarly, we have for  $\Im\lambda<0$
\begin{equation}\label{eq:d-est2}
|d(\lambda)^{-1}|<1.
\end{equation}
For later use, we derive the asymptotic expansion for $\log d(\lambda)$ as $x\to+\infty$. After performing integration by parts on \eqref{eq:d} and using the relation $\frac{\sigma'(\lambda;s)}{1-\sigma(\lambda;s)}=-2\lambda\sigma(\lambda;s)$, we have
\begin{equation}\label{eq:logh}
\log d(\lambda)=-\frac{x^2}{\pi i} \int_{\mathbb{R}}\log (\lambda-\zeta)\zeta\sigma(x\zeta;s)d\zeta,
\end{equation}
where $\arg(\lambda-\zeta)\in(-\pi, \pi)$.
Replacing $\sigma(x\zeta;s)$ by  $\chi_{(-b,b)}(\zeta)$ with $b=\frac{\sqrt{s}}{x}$, we have as $x\to+\infty$
\begin{align}\label{eq:loghEst}
\log d(\lambda)&=-\frac{x^2}{\pi i} \int_{-b}^b\log (\lambda-\zeta)\zeta d\zeta+O(1) \nonumber\\
&=\frac{x^2}{2\pi i}\left(2b\lambda+(\lambda^2-b^2)\log\left(\frac{\lambda-b}{\lambda+b}\right)\right)+O(1),
\end{align}
where $\arg(\lambda\pm b)\in(-\pi,\pi)$ and the error term is uniform for $\lambda\in \mathbb{C}\setminus \mathbb{R}$.

\subsection{Small-norm RH problem}

We define
\begin{equation}\label{eq:1R}
R(\lambda) =S(\lambda)N(\lambda)^{-\sigma_3}.
\end{equation}
Then, it follows from \eqref{eq:Sjump}, \eqref{eq:SJ} and  \eqref{eq:N} that $R(\lambda)$ satisfies the following RH problem.
\begin{rhp} \label{RHPR}
The function  $R(\lambda)$  satisfies the following properties.
\begin{itemize}
\item[\rm (1)]  $R(\lambda)$ is analytic in $\mathbb{C}\setminus \{\Gamma_1\cup\Gamma_2\}$.

 \item[\rm (2)]  $R(\lambda)$  satisfies the jump condition
\begin{equation}\label{eq:Rjump}
R_{+}(\lambda)=R_{-}(\lambda)J_{R}(\lambda), ~ J _{R}(\lambda)= \left\{        \begin{array}{cc}
\left(
                               \begin{array}{cc}
                                 1 &\exp(-x^2(\lambda- i)^2-x^2+s )d(\lambda)^2\\
                                 0 & 1 \\
                                 \end{array}
                             \right),&   \lambda\in \Gamma_1,\\

                                                      \left(
    \begin{array}{cc}
                                 1 &0\\
                                 -\exp(-x^2(\lambda+ i)^2-x^2+s) d(\lambda)^{-2} & 1 \\
                                 \end{array}
                             \right), &   \lambda\in \Gamma_2. 
                                 \end{array}
                             \right.
\end{equation}
 \item[\rm (3)]   As $\lambda\to\infty$, we have
       \begin{equation}\label{eq:RInfinity} R(\lambda)=I+O\left(\frac{1}{\lambda}\right).\end{equation}                
  \end{itemize}
  \end{rhp}
 
 \subsection{Large $x$ asymptotics of $R$ for $s\in(-\infty, c_1x^2]$ with $0<c_1<1$ }
 If $s\in(-\infty, c_1x^2]$ with any constant $0<c_1<1$, we choose the contours 
  \begin{equation}\label{eq:Gamma1}\Gamma_1=\{\lambda\in\mathbb{C}: \Im \lambda=1\} \quad \mbox{and} \quad  \Gamma_2=\{\lambda\in\mathbb{C}: \Im \lambda=-1\}.
  \end{equation}
  Then,  from \eqref{eq:d-est1}, \eqref{eq:d-est2}  and \eqref{eq:Rjump}, we obtain
  \begin{equation}\label{eq:JRExpan}J_{R}(\lambda)=I+O\left(\exp\left(-(x^2-s+x^2|\Re \lambda|^2)\right)\right), \end{equation}   
where the error term is uniform for $\lambda\in\Gamma_1\cup\Gamma_2$ and 
$s\in(-\infty, c_1x^2]$. 

From  \eqref{eq:JRExpan}, we see that $R(\lambda)$ satisfies   a small-norm RH problem for large $x$. 
By standard  arguments of small-norm RH problems \cite{ Deift, DZ}, we have the expansion as $x\to+\infty$
\begin{equation}\label{eq:RExpand1}
R(\lambda) =I+O\left(\frac{\exp(-(x^2-s))}{1+|\lambda|}\right),
\end{equation}
and 
\begin{equation}\label{eq:dRExpand1}
\frac{d}{d\lambda}R(\lambda) =O\left(\frac{\exp(-(x^2-s))}{1+|\lambda|^2}\right),
\end{equation}
where the error term is uniform for $\lambda\in \mathbb{C}\setminus\{ \Gamma_1\cup\Gamma_2\}$ and 
 $s\in(-\infty, c_1x^2 ]$  with any constant $0<c_1< 1$. 
 
 \subsection{Large $x$ asymptotics of $R$ for $s\in[ c_1x^2, c_2x^2 ]$  with $0<c_1<c_2< \frac{\pi^2}{4}$ }

In this subsection, we consider the region  $s\in[ c_1x^2, c_2x^2 ]$  with any constant $0<c_1<c_2< \frac{\pi^2}{4}$.  From \eqref{eq:loghEst}, we have
\begin{equation}\label{eq:loghEst1}
\log \left(\exp(-x^2(\lambda- i)^2-x^2+s )d(\lambda)^2\right)
=x^2\varphi(\lambda)+ O(1),
\end{equation}
where
\begin{equation}\label{eq:varphi}
\varphi(\lambda)
= -\lambda^2+2i\lambda+b^2 +\frac{1}{\pi i}\left(2b\lambda+(\lambda^2-b^2)\log\left(\frac{\lambda-b}{\lambda+b}\right)\right),
\end{equation}
where $b=\frac{\sqrt{s}}{x}$ and the principle branches for $\log(\lambda\pm b)$ are taken such that $\arg(\lambda\pm b)\in(-\pi,\pi)$. From \eqref{eq:varphi}, we derive  the following properties:
\begin{equation}\label{eq:varphiPro1}
\Re \varphi(\lambda)
= -\lambda^2+b^2<0, \quad \lambda \in (-\infty, -b )\cup (b, +\infty ), 
\end{equation}
\begin{equation}\label{eq:varphiPro21}
 \Re \varphi_{+}(\lambda)=0, \quad \lambda\in ( -b, b ),
\end{equation}
\begin{equation}\label{eq:varphiPro22}
\frac{d}{d\lambda} \Im \varphi_{+}(\lambda)
= 2-\frac{4b}{\pi}-\frac{2}{\pi}\lambda\log\left|\frac{\lambda-b}{\lambda+b}\right|>0, \quad  \lambda\in ( -b, b ),
\end{equation}
and 
\begin{equation}\label{eq:varphiPro3}
\Re \varphi(\lambda)\sim \frac{2b}{\pi} \Im \lambda \log|\lambda\mp b|<0, \quad \Im \lambda>0, \quad \lambda\to \pm b.
\end{equation}
It is noted that the restriction $b=\frac{\sqrt{s}}{x}<\frac{\pi}{2}$ ensures that the inequality in  \eqref{eq:varphiPro22} holds.
 From \eqref{eq:varphiPro1}-\eqref{eq:varphiPro3},  there exists a sufficiently small constant $\delta>0$  such that
 \begin{equation}\label{eq:Revarphi}
\Re \varphi(\lambda)<-c|\lambda|^2,  \quad \Im \lambda=\delta,
\end{equation}  
with some constant $c>0$. 
In the derivation of  \eqref{eq:Revarphi} for $|\Re \lambda\pm b |<r$ with sufficiently small $r>0$, we use \eqref{eq:varphiPro3}. The estimate of  \eqref{eq:Revarphi} for 
$|\Re \lambda|>b+r/2$ follows from \eqref{eq:varphiPro1}. And the estimate of \eqref{eq:Revarphi} for 
$\Re \lambda\in (-b+r/2, b-r/2)$  is obtained by using  \eqref{eq:varphiPro21}, \eqref{eq:varphiPro22} and the Cauchy-Riemann equations.
Now, we choose the contours
  \begin{equation}\label{eq:Gamma2}\Gamma_1=\{\lambda\in\mathbb{C}: \Im \lambda=\delta\} \quad  \mbox{ and} \quad \Gamma_2=\{\lambda\in\mathbb{C}: \Im \lambda=-\delta\}.\end{equation}
Then, from \eqref{eq:Rjump}, \eqref{eq:loghEst1} and \eqref{eq:Revarphi},  we  have 
\begin{equation}\label{eq:JRExpan1}J_{R}(\lambda)=I+O(\exp(-cx^2|\lambda|^2)), \end{equation}   
where $c$ is some positive constant and the error term is uniform for $\lambda\in\Gamma_1$ defined in \eqref{eq:Gamma2} and $s\in[ c_1x^2, c_2x^2 ]$.  Similarly, we have  \eqref{eq:JRExpan1} uniformly for $\lambda\in\Gamma_2$. 

Similar to \eqref{eq:RExpand1} and  \eqref{eq:dRExpand1},  we have from \eqref{eq:JRExpan1}
\begin{equation}\label{eq:RExpand12}
R(\lambda) =I+O\left(\frac{\exp(-cx^2)}{1+|\lambda|}\right),
\end{equation}
and 
\begin{equation}\label{eq:dRExpand12}
\frac{d}{d\lambda}R(\lambda) =O\left(\frac{\exp(-cx^2)}{1+|\lambda|^2}\right),
\end{equation}
where $c$ is some positive constant and the error term is uniform for $\lambda\in \mathbb{C}\setminus \{\Gamma_1\cup\Gamma_2\}$ and
 $s\in[ c_1x^2, c_2x^2 ]$  with any constant $0<c_1<c_2< \frac{\pi^2}{4}$.

\section{Large $s$ asymptotics of $Y$ for $x=\frac{2l}{\pi} \sqrt{s} $ with $M_2\leq l\leq 1-M_1s^{-\frac{2}{3}}$}\label{sec:asyYlarges}
In this section, we will carry out the nonlinear steepest descent analysis of the RH problem for $Y(\lambda)$ as $s\to+\infty$ for $x=\frac{2l}{\pi} \sqrt{s} $, $l\in [M_2, 1-M_1s^{-\frac{2}{3}}]$ with  a sufficiently large constant $M_1>0$ and any constant $0<M_2<1$.

\subsection{First transformation: $Y\to T$}
To simplify the jump condition, we introduce the transformation
 \begin{equation}\label{eq:T}
T(\lambda)=Y(\sqrt{s}\lambda)e^{ix\sqrt{s}\lambda\sigma_3}.\end{equation}
Then, $T(\lambda)$ satisfies the RH problem below.
\begin{rhp} \label{RHhPsi}
The function  $T(\lambda)$   satisfies the following properties.
\begin{itemize}
\item[\rm (1)]  $T(\lambda)$ is analytic in $\mathbb{C}\setminus \mathbb{R}$.

 \item[\rm (2)] $T(\lambda)$ satisfies the jump condition
 \begin{equation}\label{eq:TJump1}T_+(\lambda)=T_-(\lambda)\left(
                               \begin{array}{cc}
                                 1-\sigma(\sqrt{s}\lambda;s) & \sigma(\sqrt{s}\lambda;s)\\
                                 -\sigma(\sqrt{s}\lambda;s) & 1+\sigma(\sqrt{s}\lambda;s) \\
                                 \end{array}
                             \right), 
                             \end{equation}    
                             for $\lambda \in \mathbb{R}$.
 \item[\rm (3)]   As $\lambda\to\infty$, we have
  \begin{equation}\label{eq:TInfinity}T(\lambda)=\left(I+O\left(\frac{1}{\lambda}\right)\right)e^{ix\sqrt{s}\lambda\sigma_3}.\end{equation}
  \end{itemize}
  \end{rhp}

\subsection{Second transformation: $T\to A$}
As $s\to +\infty$, we have $ \sigma(\sqrt{s}\lambda;s)\to \chi_{(-1,1)}(\lambda)$. Therefore, we may factorize 
the jump matrix in \eqref{eq:TJump1} as follows:
\begin{equation}\label{eq:AFac1}\left(
                               \begin{array}{cc}
                                 1-\sigma(\sqrt{s}\lambda;s) & \sigma(\sqrt{s}\lambda;s)\\
                                 -\sigma(\sqrt{s}\lambda;s)& 1+\sigma(\sqrt{s}\lambda;s) \\
                                 \end{array}
                             \right)=\left(
                               \begin{array}{cc}
                                 1 &0\\
                                 1 & 1 \\
                                 \end{array}
                             \right)\left(
                               \begin{array}{cc}
                                 1 - \sigma(\sqrt{s}\lambda;s)& 1 \\
                                 -1& 0 \\
                                 \end{array}
                             \right) \left(
                               \begin{array}{cc}
                                 1 &-1\\
                                 0 & 1 \\
                                 \end{array}
                             \right),
                             \end{equation}    
for $\lambda\in(-\lambda_0,\lambda_0)$ and     
      \begin{align}\nonumber 
      \left(
                               \begin{array}{cc}
                               \sigma(\sqrt{s}\lambda;s)& -1+\sigma(\sqrt{s}\lambda;s)\\
                                 1+ \sigma(\sqrt{s}\lambda;s)& \sigma(\sqrt{s}\lambda;s) \\
                                 \end{array}
                             \right)&=\left(
                               \begin{array}{cc}
                                 1 &0\\
                                 -\frac{\sigma(\sqrt{s}\lambda;s)}{1- \sigma(\sqrt{s}\lambda;s)}& 1 \\
                                 \end{array}
                             \right)
                             \left(1-\sigma(\sqrt{s}\lambda;s)\right)^{\sigma_3} \\ \label{eq:AFac2} 
                     &~~~~~~~~~~~~~~~~~~~~~~~~~~~\times       \left( \begin{array}{cc}
                                 1 &\frac{\sigma(\sqrt{s}\lambda;s)}{1- \sigma(\sqrt{s}\lambda;s)}\\
                                0 & 1 \\
                                 \end{array}\right),
                             \end{align}    
   for $\lambda\in(-\infty,-\lambda_0)\cup(\lambda_0, \infty)$.  The endpoint $\lambda_0$ will be specified later.      
   
Based on \eqref{eq:AFac1} and \eqref{eq:AFac2},  we introduce the second transformation
\begin{equation}\label{eq:A}
A(\lambda)=T(\lambda)\left\{        \begin{array}{cc}
\left(
                               \begin{array}{cc}
                                 1 &- \frac{\sigma(\sqrt{s}\lambda;s)}{1- \sigma(\sqrt{s}\lambda;s)}\\
                                0 & 1 \\
                                 \end{array}
                             \right),&   \lambda\in \Omega_1\cup\Omega_3,\\
                                 \left(
                               \begin{array}{cc}
                                 1 &1\\
                                 0 & 1 \\
                                 \end{array}
                             \right),&  \lambda\in \Omega_2,\\
                             \left(
                               \begin{array}{cc}
                                 1 &0\\
                                 -\frac{\sigma(\sqrt{s}\lambda;s)}{1- \sigma(\sqrt{s}\lambda;s)}& 1 \\
                                 \end{array}
                             \right),&  \lambda\in \Omega_4\cup \Omega_6,\\
                             \left(
                               \begin{array}{cc}
                                 1 &0\\
                                 1 & 1 \\
                                 \end{array}
                             \right),&  \lambda\in \Omega_5.
                                 \end{array}
                             \right.
\end{equation}
Here, the regions are illustrated in Fig. \ref{fig:A}. The contours $\Sigma_k$, $k=1,2,3,4$ are chosen such that $\Re (\lambda^2-\lambda_0^2)=0$.  From \eqref{eq:DeformWelldefine}, it follows that the function $\frac{\sigma(\sqrt{s}\lambda;s)}{1- \sigma(\sqrt{s}\lambda;s)}=\exp(-s(\lambda^2-1))$ is entire.
Therefore, it follows from \eqref{eq:TJump1}, \eqref{eq:TInfinity} and \eqref{eq:A}  that the function $A(\lambda)$  satisfies the following RH problem.

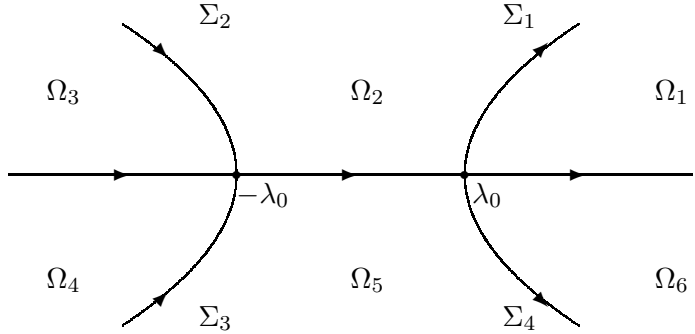
\begin{figure}[h]
\begin{center}
   \setlength{\unitlength}{1truemm}
   \begin{picture}(100,65)(-5,12)
              \put(25,40){\line(1,0){30}}
          \put(55,40){\line(1,0){30}}
               \put(25,40){\line(-1,0){30}}
     \qbezier(10,60)(40,40)(10,20)
 \qbezier(70,60)(40,40)(70,20)

 \put(40,40){\thicklines\vector(1,0){1}}
  \put(10,40){\thicklines\vector(1,0){1}}
   \put(70,40){\thicklines\vector(1,0){1}}
 
       \put(15,56.5){\thicklines\vector(1,-1){1}}
       
       \put(15,23.5){\thicklines\vector(1,1){1}}
      
       \put(65,23.5){\thicklines\vector(1,-1){1}}
     
       \put(65,56.5){\thicklines\vector(1,1){1}}

       \put(60,60){$\Sigma_1$}
      \put(20,60){$\Sigma_2$}
 \put(20,20){$\Sigma_3$}
\put(60,20){$\Sigma_4$}

  \put(80,50){$\Omega_1$}
   \put(40,50){$\Omega_2$}
    \put(0,50){$\Omega_3$}
     \put(0,25){$\Omega_4$}
   \put(40,25){$\Omega_5$}
    \put(80,25){$\Omega_6$}

       \put(25,40){\thicklines\circle*{1}}
       \put(55,40){\thicklines\circle*{1}}

       \put(25,36.3){$-\lambda_0$}
       \put(56,36.3){$\lambda_0$}
\end{picture}
   \caption{Jump contours and regions for the RH problem for $A$.}
   \label{fig:A}
\end{center}
\end{figure}

\vskip 2cm
\begin{rhp} \label{RHPA}
The function  $A(\lambda)$  satisfies the following properties.
\begin{itemize}
\item[\rm (1)]  $A(\lambda)$ is analytic in $\mathbb{C}\setminus \Sigma_A$, where $\Sigma_A=\cup_{1}^4\Sigma_k\cup \mathbb{R}$ is shown in Fig. \ref{fig:A}. 

 \item[\rm (2)]  $A(\lambda)$  satisfies the jump condition
 \begin{equation}\label{eq:AJump}A_+(\lambda)=A_-(\lambda)\left\{        \begin{array}{cc}
 \left(
                               \begin{array}{cc}
                                 1 & \frac{1}{1- \sigma(\sqrt{s}\lambda;s)}\\
                                0 & 1 \\
                                 \end{array}
                             \right),&  \lambda\in \Sigma_1\cup\Sigma_2,\\
\left(
                               \begin{array}{cc}
                                 1 &0\\
                                 -\frac{1}{1- \sigma(\sqrt{s}\lambda;s)} & 1 \\
                                 \end{array}
                             \right),&  \lambda\in \Sigma_3\cup\Sigma_4,\\

                             \left(
                               \begin{array}{cc}
                                 1-\sigma(\sqrt{s}\lambda;s)&1\\
                                 -1& 0\\
                                 \end{array}
                             \right),&  \lambda\in (-\lambda_0,\lambda_0),\\
                             \left(
                               1-\sigma(\sqrt{s}\lambda;s)
                             \right)^{\sigma_3}, &  \lambda\in(-\infty,-\lambda_0)\cup (\lambda_0,+\infty). 
                                 \end{array}
                             \right.
\end{equation}
 \item[\rm (3)]   As $\lambda\to\infty$, we have
  \begin{equation}\label{eq:AInfinity}A(\lambda)=
  \left( I+ O\left(\frac{1}{\lambda}\right)\right)e^{ix\sqrt{s}\lambda\sigma_3}.
\end{equation}
                  
  \end{itemize}
  \end{rhp}

\subsection{Third transformation: $A\to B$}
To normalize the behavior of $A(\lambda)$ at infinity, we introduce the transformation
\begin{equation}\label{eq:B}
B(\lambda)=(1-\sigma(\sqrt{s}\lambda_0;s))^{\pm\frac{1}{2}\sigma_3}
A(\lambda)\exp(-g(\lambda)\sigma_3)(1-\sigma(\sqrt{s}\lambda_0;s))^{\mp\frac{1}{2}\sigma_3},\end{equation}
for $\pm \Im \lambda>0$.
Here, the $g$-function satisfies the asymptotic behavior as $\lambda\to\infty$
\begin{equation}\label{eq:gInfty}
g(\lambda)=ix\sqrt{s}\lambda+O(1/\lambda).
\end{equation}
To simplify the jump condition,  we assume $g(\lambda)$ is analytic in $\mathbb{C}\setminus\mathbb{R}$ and satisfies the following properties:
\begin{equation}\label{eq:phase}
g_+(\lambda)-g_{-}(\lambda)-V(\lambda)+V(\lambda_0)=0, \quad \lambda\in(-\infty,-\lambda_0)\cup(\lambda_0,+\infty),
\end{equation}
with 
\begin{equation}\label{eq:V}
V(\lambda)=\log(1-\sigma(\sqrt{s}\lambda;s)),
\end{equation}
and 
\begin{equation}\label{eq:phase2}
g_+(\lambda)+g_{-}(\lambda)=0,  \quad \lambda\in(-\lambda_0,\lambda_0).
\end{equation}
 By the Sokhotski-Plemelj formula, we obtain  from \eqref{eq:gInfty}-\eqref{eq:phase2} that 
\begin{align}\label{eq:g}
g(\lambda)&=\sqrt{\lambda^2-\lambda_0^2} \left(ix\sqrt{s}+\frac{1}{2\pi i}\int_{J}\frac{V(\zeta)-V(\lambda_0)}{\sqrt{\zeta^2-\lambda_0^2}}\frac{d\zeta}{ \zeta-\lambda}\right)\nonumber\\
&=\sqrt{\lambda^2-\lambda_0^2} \left(ix\sqrt{s}+\frac{1}{\pi i}\int_{\lambda_0}^{+\infty}\frac{V(\zeta)-V(\lambda_0)}{\sqrt{\zeta^2-\lambda_0^2}}\frac{\zeta d\zeta}{ \zeta^2-\lambda^2}\right),
\end{align}
where   $J=(-\infty,-\lambda_0)\cup(\lambda_0,+\infty)$ and $ \sqrt{\lambda^2-\lambda_0^2}$ takes the branch cut along $[-\lambda_0, \lambda_0]$ and behaves like $\lambda$ as $\lambda\to\infty$.


\begin{figure}[h]
\begin{center}
   \setlength{\unitlength}{1truemm}
   \begin{picture}(100,50)(-5,20)
       
       \put(25,40){\line(1,0){30}}

\qbezier(10,60)(40,40)(10,20)
 \qbezier(70,60)(40,40)(70,20)

 \put(40,40){\thicklines\vector(1,0){1}}
 
       \put(15,56.5){\thicklines\vector(1,-1){1}}
       
       \put(15,23.5){\thicklines\vector(1,1){1}}
      
       \put(65,23.5){\thicklines\vector(1,-1){1}}
     
       \put(65,56.5){\thicklines\vector(1,1){1}}

       \put(60,60){$\Sigma_1$}
      \put(20,60){$\Sigma_2$}
 \put(20,20){$\Sigma_3$}
\put(60,20){$\Sigma_4$}

       \put(25,40){\thicklines\circle*{1}}
       \put(55,40){\thicklines\circle*{1}}

       \put(25,36.3){$-\lambda_0$}
       \put(56,36.3){$\lambda_0$}
\end{picture}
   \caption{Jump contours and regions for the RH problem for $B$.}
   \label{fig:B}
\end{center}
\end{figure}
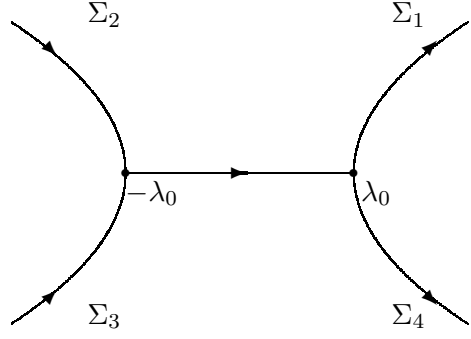
Then, $B(\lambda)$ satisfies the following RH problem.

\begin{rhp} \label{RHPB}
The function  $B(\lambda)$  satisfies the following properties.
\begin{itemize}
\item[\rm (1)]  $B(\lambda)$ is analytic in $\mathbb{C}\setminus \Sigma_B$, where $\Sigma_B=\cup_{1}^4~\Sigma_{k}\cup(-\lambda_0,\lambda_0)$ is described in  Fig. \ref{fig:B}.

 \item[\rm (2)]  $B(\lambda)$  satisfies the jump condition
\begin{equation}\label{eq:Bjump}
B_{+}(\lambda)=B_{-}(\lambda)\left\{        \begin{array}{ll}
\left(
                               \begin{array}{cc}
                                 1 &\exp(2g(\lambda)-V(\lambda)+V(\lambda_0))\\
                                 0 & 1 \\
                                 \end{array}
                             \right), &   \lambda\in\Sigma_1\cup\Sigma_2,\\
                             \left(
                               \begin{array}{cc}
                                 1 &0\\
                                 - \exp(-2g(\lambda)-V(\lambda)+V(\lambda_0)) & 1 \\
                                 \end{array}
                             \right), &  \lambda\in\Sigma_3\cup\Sigma_4,\\
                             \left(
                               \begin{array}{cc}
                                 \exp(-2g_{+}(\lambda)+V(\lambda)-V(\lambda_0)) & 1 \\
                                 -1&0\\
                                 \end{array}
                             \right), &  \lambda\in(-\lambda_0,\lambda_0).
                                 \end{array}
                             \right.
\end{equation}
 \item[\rm (3)]   As $\lambda\to\infty$, we have
  \begin{equation}\label{eq:BInfinity}B(\lambda)=I+O\left(\frac{1}{\lambda}\right).\end{equation}         
  \end{itemize}
  \end{rhp}

We denote
\begin{equation}\label{eq: phi}
\phi(\lambda)=g(\lambda)\mp\frac{1}{2} (V(\lambda)-V(\lambda_0)),  \quad \pm \Im \lambda>0.
\end{equation}
Then, the jump condition in \eqref{eq:Bjump} can be written as follows:
\begin{equation}\label{eq:Bjump1}
B_{+}(\lambda)=B_{-}(\lambda)\left\{        \begin{array}{ll}
\left(
                               \begin{array}{cc}
                                 1 &\exp(2\phi(\lambda))\\
                                 0 & 1 \\
                                 \end{array}
                             \right),&   \lambda\in\Sigma_1\cup\Sigma_2,\\
                             \left(
                               \begin{array}{cc}
                                 1 &0\\
                                 - \exp(-2\phi(\lambda)) & 1 \\
                                 \end{array}
                             \right),&  \lambda\in\Sigma_3\cup\Sigma_4,\\
                             \left(
                               \begin{array}{cc}
                                 \exp(-2\phi_{+}(\lambda)) & 1 \\
                                 -1&0\\
                                 \end{array}
                             \right),&  \lambda\in(-\lambda_0,\lambda_0).
                                 \end{array}
                             \right.
\end{equation}

In the subsequent RH analysis, it will be convenient if  the function $\phi(\lambda)$ vanishes at the endpoints $\pm \lambda_0$ of order $\frac{3}{2}$.  In view of \eqref{eq:g} and \eqref{eq: phi}, 
we use the freedom to choose the endpoint $\lambda_0$ such that 
\begin{equation}\label{eq:lambda0}
\frac{1}{\pi}\int_{\lambda_0}^{+\infty}\frac{V(\zeta)-V(\lambda_0)}{(\zeta^2-\lambda_0^2)^{3/2}}\zeta d\zeta=x\sqrt{s}.
\end{equation}
From \eqref{eq:V}, we have the derivative of $V(\lambda)$ with respect to $\lambda$
\begin{equation}\label{eq:derV}
V'(\lambda)=2s\lambda\sigma(\sqrt{s}\lambda;s).
\end{equation}
Therefore, using integration by parts, the  equation \eqref{eq:lambda0} is equivalent to 
\begin{equation}\label{eq:lambda01}
\int_{\lambda_0}^{+\infty}\frac{\zeta\sigma(\sqrt{s}\zeta;s)}{\sqrt{\zeta^2-\lambda_0^2}} d\zeta=\int_{0}^{+\infty}\sigma\left(\sqrt{s}(\zeta^2+\lambda_0^2)^{\frac{1}{2}};s\right) d\zeta=\frac{\pi}{2}\frac{x}{\sqrt{s}}=l\in(0,1). 
\end{equation}
We will show in next section that such an endpoint $\lambda_0>0$ exists and is unique, and that the function $\phi(\lambda)$ vanishes at the endpoints $\pm \lambda_0$ of order $\frac{3}{2}$
under the condition \eqref{eq:lambda0} or  equivalently \eqref{eq:lambda01}.

\subsection{Properties of the $g$-function and $\phi$-function}
In this section, we derive some useful properties of the $g$-function and the $\phi$-function, formulated in a series of lemmas. These properties play important roles in the asymptotic analysis of the RH problem for $B(\lambda)$. 

For this purpose, we define
\begin{equation}\label{eq:L}
L(s)=\int_0^{+\infty}\sigma(\sqrt{s}\zeta;s)d\zeta.
\end{equation}
 We derive the asymptotic expansion of the above integral over the intervals 
$[0, 1]$ and $[1,+\infty)$ separately. From \eqref{eq:sigma},  we have 
\begin{equation}\label{eq:LAsy0}
\int_0^{1}\sigma(\sqrt{s}\zeta;s)d\zeta=\int_0^{1}\frac{1}{1+\exp(s(\zeta^2-1))}d\zeta=1-\int_0^{1}\frac{\exp(s(\zeta^2-1))}{1+\exp(s(\zeta^2-1))}d\zeta.
\end{equation}
As $s\to+\infty$, the  major contribution to the asymptotics of the last integral in \eqref{eq:LAsy0} arises from the neighborhood of $\zeta=1$. Therefore, we have as $s\to+\infty$
\begin{equation}\label{eq:LAsy1}
\int_0^{1}\frac{\exp(s(\zeta^2-1))}{1+\exp(s(\zeta^2-1))}d\zeta=\frac{1}{4s}+O(1/s^2).
\end{equation}
Similarly, we have 
as $s\to+\infty$
\begin{equation}\label{eq:LAsy2}
\int_{1}^{+\infty}\sigma(\sqrt{s}\zeta;s)d\zeta=\int_{1}^{+\infty} \frac{\exp(-s(\zeta^2-1))}{1+\exp(-s(\zeta^2-1))}d\zeta=\frac{1}{4s}+O(1/s^2).
\end{equation}
Substituting \eqref{eq:LAsy0}-\eqref{eq:LAsy2} into \eqref{eq:L}, we have the asymptotic expansion as $s\to+\infty$
\begin{equation}\label{eq:LAsy}
L(s)=1+O(1/s^2). 
\end{equation}

\begin{lem}\label{lem: lambda0}
Suppose $0<\frac{\pi}{2}\frac{x}{\sqrt{s}}\leq L(s)$, 
then there exists a unique solution $\lambda_0$ to the equation \eqref{eq:lambda01}. 
Moreover, we have the asymptotic expansion as $s\to +\infty$
\begin{equation}\label{eq:lambda0Est}
\lambda_0^2=1-\frac{\pi^2 x^2}{4s}+O(1/s^2),
\end{equation}
 where the error term is uniform for $\frac{\pi}{2}\frac{x}{\sqrt{s}}\in [M, L(s)]$ with any constant $0<M<L(s)$. 
\end{lem}
\begin{proof}
It follows from \eqref{eq:sigma} that the second integral in \eqref{eq:lambda01} is strictly decreasing from the value $L(s)$ defined in \eqref{eq:L} to zero as $\lambda_0$ increases from zero to positive infinity.  Therefore, there exists a unique positive solution $\lambda_0$ to the equation \eqref{eq:lambda01}  for $0<\frac{\pi}{2}\frac{x}{\sqrt{s}}\leq L(s)$. 

We proceed to derive the asymptotics of $\lambda_0$. For this purpose, we split the first integral in \eqref{eq:lambda01} into three parts
\begin{align}\label{eq:dginfty0}
 \int_{\lambda_0}^{+\infty} \zeta\sigma(\sqrt{s}\zeta;s)\frac{d\zeta}{\sqrt{\zeta^2-\lambda_0^2}} 
 = &\int_{\lambda_0}^{1} \frac{\zeta d\zeta}{\sqrt{\zeta^2-\lambda_0^2}} -\int_{\lambda_0}^{1}\zeta(1-\sigma(\sqrt{s}\zeta;s))\frac{d\zeta}{\sqrt{\zeta^2-\lambda_0^2}}\\ \nonumber
 &+ \int_{1}^{+\infty} \zeta\sigma(\sqrt{s}\zeta;s)\frac{d\zeta}{\sqrt{\zeta^2-\lambda_0^2}}.
 \end{align}
 From \eqref{eq:sigma}, we have as $s\to+\infty$
\begin{align}
\int_{\lambda_0}^{1} \zeta(1-\sigma(\sqrt{s} \zeta;s))\frac{d\zeta}{\sqrt{ \zeta^2-\lambda_0^2}}&=\int_{\lambda_0}^{1} \zeta\sigma(\sqrt{s} \zeta;s)\exp(s(1- \zeta^2))\frac{d \zeta}{\sqrt{ \zeta^2-\lambda_0^2}} \nonumber\\
&=\frac{\sigma(\sqrt{s};s)}{2s\sqrt{1-\lambda_0^2}}+O(1/s^2) \nonumber\\
&=\frac{1}{4s\sqrt{1-\lambda_0^2}}+O(1/s^2),
\label{eq:dgint1}
\end{align}
uniformly for  $0\leq\lambda_0\leq1-\delta$ with any constant $0<\delta<1$. Similarly, we have from \eqref{eq:sigma}
\begin{align}
\int_{1}^{+\infty} \zeta\sigma(\sqrt{s}\zeta;s)\frac{d\zeta}{\sqrt{\zeta^2-\lambda_0^2}}
&=\frac{1}{4s\sqrt{1-\lambda_0^2}}+O(1/s^2).
\end{align}
Therefore, we have as $s\to+\infty$
 \begin{equation}\label{eq:dginfty}
 \int_{\lambda_0}^{+\infty} \zeta\sigma(\sqrt{s}\zeta;s)\frac{d\zeta}{\sqrt{\zeta^2-\lambda_0^2}}  =
 \int_{\lambda_0}^{1} \frac{\zeta d\zeta}{\sqrt{\zeta^2-\lambda_0^2}}+O(1/s^2)= \sqrt{1-\lambda_0^2}+O(1/s^2).\end{equation}
Then, it follows from \eqref{eq:lambda01} and \eqref{eq:dginfty} that
\begin{equation}\label{eqIntAsy}
 \sqrt{1-\lambda_0^2}=\frac{\pi x}{2\sqrt{s}}+O(1/s^2).
\end{equation}
 Therefore, we have \eqref{eq:lambda0Est}.
 \end{proof}
\begin{rem} 
In the regime $x=\frac{2l}{\pi} \sqrt{s} $, $l\in [M_2, 1-M_1s^{-\frac{2}{3}}]$  with $M_1>0$ and $0<M_2<1$  considered in this section, it follows from  \eqref{eq:LAsy} that for sufficiently large  $s$
\begin{equation}\label{eq:LM}
l=\frac{\pi}{2}\frac{x}{\sqrt{s}}<L(s).
\end{equation}
As  a consequence of Lemma \ref{lem: lambda0}, the  endpoint $\lambda_0>0$ exists and is unique for $(x,s)$ in this regime. 
\end{rem}

We now derive the properties of the $g$-function and the $\phi$-function.
\begin{lem}\label{lem:dg}
Suppose  $\frac{\pi}{2}\frac{x}{\sqrt{s}}\in [M, L(s)]$ with any constant $0<M<L(s)$ and $L(s)$ given in \eqref{eq:L}, then we have 
\begin{align}\label{eq:dg}
g'(\lambda)=\frac{2}{\pi i}s\lambda\sqrt{\lambda^2-\lambda_0^2} \int_{\lambda_0}^{+\infty} \frac{ \zeta \sigma(\sqrt{s} \zeta;s)}{\sqrt{ \zeta^2-\lambda_0^2}}\frac{d \zeta}{ \zeta^2-\lambda^2},
\end{align}
where $\sqrt{\lambda^2-\lambda_0^2}$ takes the branch cut along $[-\lambda_0, \lambda_0]$ and behaves like $\lambda$ as $\lambda\to\infty$.
  Furthermore, we have the asymptotic expansion as $\lambda\to\infty$
\begin{equation}\label{eq:dgExpand}
g(\lambda) =  ix\sqrt{s}\lambda+\frac{g_1}{\lambda}+O(1/\lambda^3), \end{equation}
with
\begin{equation}\label{eq:dgExpandg1}
 g_1
 =- \frac{ix\sqrt{s}}{2}\lambda_0^2-\frac{2is}{3\pi} \left(1-\lambda_0^2\right)^{\frac{3}{2}} +O(1/s), \quad s\to+\infty.
 \end{equation}

\end{lem}

\begin{proof} 

From \eqref{eq:g}, we have 
\begin{align}\label{eq:dgIni}
g'(\lambda)=&\frac{\lambda}{\sqrt{\lambda^2-\lambda_0^2}}  \left(ix\sqrt{s}+\frac{1}{\pi i}\int_{\lambda_0}^{+\infty}\frac{V(\zeta)-V(\lambda_0)}{\sqrt{\zeta^2-\lambda_0^2}}\frac{\zeta d\zeta}{ \zeta^2-\lambda^2}\right)\nonumber\\
&+\frac{ 1}{\pi i}\lambda\sqrt{\lambda^2-\lambda_0^2}\int_{\lambda_0}^{+\infty}\frac{V(\zeta)-V(\lambda_0)}{\sqrt{\zeta^2-\lambda_0^2}}\frac{2\zeta d\zeta}{( \zeta^2-\lambda^2)^2}.
\end{align}
Using \eqref{eq:lambda0} and applying integration by parts, we have
\begin{align}\label{eq:dgBypart}
ix\sqrt{s}+&\frac{1}{\pi i}\int_{\lambda_0}^{+\infty}\frac{V(\zeta)-V(\lambda_0)}{\sqrt{\zeta^2-\lambda_0^2}}\frac{\zeta d\zeta}{ \zeta^2-\lambda^2}= \frac{\lambda^2-\lambda_0^2}{\pi i}\int_{\lambda_0}^{+\infty}\frac{V(\zeta)-V(\lambda_0)}{(\zeta^2-\lambda_0^2)^{3/2}}\frac{\zeta d\zeta}{ \zeta^2-\lambda^2}\nonumber\\
&=\frac{\lambda^2-\lambda_0^2}{\pi i}\left(\int_{\lambda_0}^{+\infty}\frac{V'(\zeta)}{\sqrt{\zeta^2-\lambda_0^2}}\frac{ d\zeta}{ \zeta^2-\lambda^2}-\int_{\lambda_0}^{+\infty}\frac{V(\zeta)-V(\lambda_0)}{\sqrt{\zeta^2-\lambda_0^2}}\frac{2\zeta d\zeta}{( \zeta^2-\lambda^2)^2}\right).
\end{align}
From  \eqref{eq:derV}, \eqref{eq:dgIni} and \eqref{eq:dgBypart}, we have \eqref{eq:dg}. 

We proceed to derive the asymptotic expansion of $g(\lambda)$ as $\lambda\to\infty$. In the expression for $g$ given in \eqref{eq:g}, the square function $\sqrt{\lambda^2-\lambda_0^2}$ is analytic in $\mathbb{C}\setminus[-\lambda_0, \lambda_0]$ and has the asymptotic expansion as $\lambda\to\infty$
\begin{equation}\label{eq: squarefuctionExpand}
\sqrt{\lambda^2-\lambda_0^2}=\lambda \left(1-\frac{\lambda_0^2}{2\lambda^2}+O(1/\lambda^4)\right).\end{equation}
We also have  the asymptotic expansion of the integral in \eqref{eq:g} as $\lambda\to\infty$
\begin{equation}\label{eq: gIntegralExpand}
\frac{1}{\pi i}\int_{\lambda_0}^{+\infty}\frac{V(\zeta)-V(\lambda_0)}{\sqrt{\zeta^2-\lambda_0^2}}\frac{\zeta d\zeta}{ \zeta^2-\lambda^2}=-\frac{1}{\pi i\lambda^2} \int_{\lambda_0}^{+\infty}\frac{V(\zeta)-V(\lambda_0)}{\sqrt{\zeta^2-\lambda_0^2}}\zeta d\zeta +O(1/\lambda^4).\end{equation}
Substituting \eqref{eq: squarefuctionExpand} and  \eqref{eq: gIntegralExpand} into  \eqref{eq:g}, we see that  $g(\lambda)$ has expansion of the form given in \eqref{eq:dgExpand} with certain coefficient $g_1$ as $\lambda\to\infty$. To derive the coefficient $g_1$,  we use \eqref{eq:lambda01} to rewrite the equation  \eqref{eq:dg} as follows:
\begin{equation}\label{eq:dg1}
g'(\lambda)=\frac{ix \sqrt{s} \lambda }{  \sqrt{\lambda^2-\lambda_0^2}} \left(1-\frac{2\sqrt{s}}{\pi x} \int_{\lambda_0}^{+\infty} \sigma(\sqrt{s}\zeta;s)\sqrt{\zeta^2-\lambda_0^2}\frac{\zeta d\zeta}{\zeta^2-\lambda^2} \right).\end{equation}
From  \eqref{eq:g} and \eqref{eq:dg1}, we have
\begin{equation}\label{eq: g1}
 g_1= -\frac{ix\sqrt{s}}{2}\lambda_0^2-\frac{2is}{\pi}\int_{\lambda_0}^{+\infty}\zeta (\zeta ^2-\lambda_0^2)^{\frac{1}{2}}\sigma(\sqrt{s}\zeta;s )d\zeta.\end{equation}
Similar to \eqref{eq:dginfty}, we have 
\begin{equation}\label{eq: g1estI}
\int_{\lambda_0}^{+\infty}\zeta (\zeta ^2-\lambda_0^2)^{\frac{1}{2}}\sigma(\sqrt{s}\zeta;s )d\zeta=\int_{\lambda_0}^{1}\zeta (\zeta ^2-\lambda_0^2)^{\frac{1}{2}}d\zeta+O(1/s^2).\end{equation}
This, together with \eqref{eq: g1}, implies \eqref{eq:dgExpandg1}.
\end{proof}

\begin{lem}
Under the assumptions of Lemma \ref{lem:dg}, we have
\begin{equation}\label{eq:phisign1}
\phi(\lambda)
= \left(\lambda^2-\lambda_0^2\right)^{\frac{3}{2}}(c_1(s)+O(\lambda^2-\lambda_0^2)), \quad \lambda\to\pm\lambda_0,
\end{equation}
where $\arg (\lambda\pm\lambda_0)\in(-\pi,\pi)$ and the coefficient
\begin{equation}\label{eq:coef1}
c_1(s)=  \frac{2is}{3\pi \sqrt{1-\lambda_0^2}}+O(1/s), \quad  s\to+\infty.
\end{equation}
\end{lem}
\begin{proof}
From \eqref{eq:g}, \eqref{eq: phi} and \eqref{eq:lambda0},  it follows that $\phi(\lambda)$ has  asymptotic expansion of the form \eqref{eq:phisign1}. We proceed to  derive the asymptotics of the coefficient $c_1(s)$ in \eqref{eq:phisign1} as $s\to+\infty$.  For convenience, we suppose $|\lambda\pm \lambda_0|<\delta$ with a sufficiently small $\delta>0$ such that $ \lambda_0+\delta<1$. Similar to  \eqref{eq:dginfty},  we obtain, after replacing $\sigma(\sqrt{s}\lambda;s)$ by  $\chi_{(-1,1)}(\lambda)$ in  \eqref{eq:dg},  the following asymptotics:
\begin{equation}\label{eq:dglarges}
g'(\lambda) =I_{0}(\lambda)(1+O(1/s^2)), 
\end{equation}
where
\begin{equation}\label{eq:dg0}
I_0(\lambda)=\frac{2}{\pi i}s\lambda\sqrt{\lambda^2-\lambda_0^2} \int_{\lambda_0}^{1} \frac{\zeta}{\sqrt{\zeta^2-\lambda_0^2}(\zeta^2-\lambda^2)}d\zeta,
\end{equation}
and the square function $\sqrt{\lambda^2-\lambda_0^2}$ takes the branch cut along $[-\lambda_0, \lambda_0]$ and behaves like $\lambda$ as $\lambda\to\infty$. 
 By a change of variables, we have 
\begin{align}\label{eq:intlarges}
\int_{\lambda_0}^{1} \frac{\zeta }{\sqrt{\zeta^2-\lambda_0^2}(\zeta^2-\lambda^2)}d\zeta&=\int_{0}^{\sqrt{1-\lambda_0^2}}\frac{1}{\zeta^2+\lambda_0^2-\lambda^2}d\zeta
\nonumber\\
&=\frac{1}{2\sqrt{\lambda^2-\lambda_0^2}} \log\left(\frac{\sqrt{\lambda^2-\lambda_0^2}-\sqrt{1-\lambda_0^2}}{\sqrt{\lambda^2-\lambda_0^2}+\sqrt{1-\lambda_0^2}}\right),
\end{align}
where the logarithm takes the principal branch.  
Therefore, we have
\begin{equation}\label{eq:I0}
I_0(\lambda) =\frac{1}{\pi i}s\lambda \log\left(\frac{\sqrt{\lambda^2-\lambda_0^2}-\sqrt{1-\lambda_0^2}}{\sqrt{\lambda^2-\lambda_0^2}+\sqrt{1-\lambda_0^2}}\right). \end{equation}
From  \eqref{eq:I0}, we have the asymptotic expansion  \begin{equation}\label{eq:g0behavior}
I_0(\lambda)-s\lambda =- \frac{2s\lambda }{\pi i} \sqrt{\frac{\lambda^2-\lambda_0^2}{1-\lambda_0^2}}\sum_{k=0}^{\infty}\frac{1}{2k+1} \left(\frac{\lambda^2-\lambda_0^2}{1-\lambda_0^2}\right)^{k}.\end{equation}
Integrating  both sides of \eqref{eq:g0behavior} and comparing it with \eqref{eq: phi} and \eqref{eq:phisign1}, we have \eqref{eq:coef1}.
\end{proof}

\begin{lem}\label{lem:phiEst}
Under the assumptions of Lemma \ref{lem:dg}, there exists some constant $c>0$ such that
\begin{equation}\label{eq:phiEst1}
\phi_{+}(\lambda)\geq cs \left|\lambda_0^2-\lambda^2\right|^{\frac{3}{2}}, \quad \lambda\in(-\lambda_0, \lambda_0),
\end{equation} 
\begin{equation}\label{eq:phiEst2}
\Re \phi(\lambda)
\leq  -cs |\lambda^2-\lambda_0^2|^{\frac{3}{2}} (|\lambda|+1)^{-2}, \quad \lambda\in \Sigma_1\cup\Sigma_2,\end{equation}
and
\begin{equation}\label{eq:phiEst3}
\Re \phi(\lambda)
\geq cs |\lambda^2-\lambda_0^2|^{\frac{3}{2}} (|\lambda|+1)^{-2}, \quad \lambda\in \Sigma_3\cup\Sigma_4.\end{equation} 
\end{lem}
\begin{proof}
For $-\lambda_0<\lambda<\lambda_0$,  after a change of variables $\eta=\sqrt{\zeta^2-\lambda_0^2}$, we obtain
\begin{equation}\label{eq:Identity}
\int_{\lambda_0}^{+\infty} \frac{\zeta }{\sqrt{\zeta^2-\lambda_0^2}(\zeta^2-\lambda^2)}d\zeta=\int_{0}^{+\infty} \frac{1 }{\eta^2+\lambda_0^2-\lambda^2}d\eta=\frac{\pi }{2\sqrt{\lambda_0^2-\lambda^2}}.
\end{equation}
This, together with \eqref{eq:derV} and \eqref{eq:dg}, implies
\begin{align}\label{eq:dphi0}
\phi_{+}'(\lambda)=g'_{+}(\lambda)-\frac{1}{2}V'(\lambda)=\frac{2}{\pi  }s\lambda\sqrt{\lambda_0^2-\lambda^2} \int_{\lambda_0}^{+\infty} \frac{\zeta(\sigma(\sqrt{s}\zeta;s)-\sigma(\sqrt{s}\lambda;s))}{\sqrt{\zeta^2-\lambda_0^2}(\zeta^2-\lambda^2)}d\zeta.\end{align}
It is seen from \eqref{eq:sigma} that the integrand in \eqref{eq:dphi0} is negative. Therefore,  from \eqref{eq:sigma} and \eqref{eq:lambda0Est}, we have
\begin{align}\label{eq:Inq}
\int_{\lambda_0}^{+\infty} \frac{\zeta(\sigma(\sqrt{s}\zeta;s)-\sigma(\sqrt{s}\lambda;s))}{\sqrt{\zeta^2-\lambda_0^2}(\zeta^2-\lambda^2)}d\zeta
&\leq \int_{1}^{+\infty} \frac{\zeta(\sigma(\sqrt{s}\zeta;s)-\sigma(\sqrt{s}\lambda;s)) }{\sqrt{\zeta^2-\lambda_0^2}(\zeta^2-\lambda^2)}d\zeta \nonumber\\
&\leq (\sigma(\sqrt{s};s)-\sigma(\sqrt{s}\lambda_0;s))\int_{1}^{+\infty} \frac{\zeta }{\sqrt{\zeta^2-\lambda_0^2}(\zeta^2-\lambda^2)}d\zeta\nonumber\\
&\leq \left(\frac{1}{2}-\sigma(\sqrt{s}\lambda_0;s)\right)\int_{1}^{+\infty} \frac{1 }{\zeta^2}d\zeta\nonumber\\
&\leq -\frac{1}{3},
\end{align}
for $s$ sufficiently large. To get the last inequality, we use the fact that $\sigma(\sqrt{s}\lambda_0;s)\to 1$ as $s\to+\infty$.
Therefore, we have 
\begin{equation}\label{eq:dphiEst1}
\phi_{+}'(\lambda)\leq -\frac{2}{3\pi}s\lambda\sqrt{\lambda_0^2-\lambda^2}, \quad 0\leq \lambda<\lambda_0,\end{equation}
and 
\begin{equation}\label{eq:dphiEst2}
\phi_{+}'(\lambda)\geq -\frac{2}{3\pi}s\lambda\sqrt{\lambda_0^2-\lambda^2}, \quad -\lambda_0<\lambda<0.\end{equation}
Since $\phi(\pm \lambda_0)=0$, we have from \eqref{eq:dg}-\eqref{eq:dphiEst2} the estimate \eqref{eq:phiEst1}.

We proceed to derive  \eqref{eq:phiEst2} and \eqref{eq:phiEst3}. For $\lambda\in\Sigma_1\cup \Sigma_2\cup \Sigma_3\cup \Sigma_4$,  we have $\Re \lambda^2=\lambda_0^2<1$. Therefore, we have as $s\to+\infty$
\begin{equation}\label{eq:gsign01}
\Re (V(\lambda_0)-V(\lambda))=\log \left|\frac{1+\exp(-s(\lambda^2-1))}{1+\exp(-s(\lambda_0^2-1))}\right|= \log \left|\frac{1+\exp(s(\lambda^2-1))}{1+\exp(s(\lambda_0^2-1))}\right|\to 0,
\end{equation}
where $V(\lambda)$ is defined in \eqref{eq:V}.
By \eqref{eq: phi}, the next task is to estimate $g(\lambda)$. For this purpose,  we rewrite  \eqref{eq:dg} using integration by parts 
\begin{equation}\label{eq:dgExp}
g'(\lambda)=\frac{is\lambda}{\pi }\int_{\lambda_0}^{+\infty} \log\left(\frac{\sqrt{\lambda^2-\lambda_0^2}-\sqrt{\zeta^2-\lambda_0^2}}{\sqrt{\lambda^2-\lambda_0^2}+\sqrt{\zeta^2-\lambda_0^2}}\right) d\sigma(\sqrt{s}\zeta),
\end{equation}
where the square roots and the logarithm take the principal branches.
Integrating \eqref{eq:dgExp} and using the fact that $g(\lambda_0)=0$, we have 
\begin{equation}\label{eq:gExp0}
g(\lambda)=\frac{is}{\pi } \int_{\lambda_0}^{\lambda}\int_{\lambda_0}^{+\infty} \eta \log\left(\frac{\sqrt{\eta^2-\lambda_0^2}-\sqrt{\zeta^2-\lambda_0^2}}{\sqrt{\eta^2-\lambda_0^2}+\sqrt{\zeta^2-\lambda_0^2}}\right) d\sigma(\sqrt{s}\zeta) d\eta.
\end{equation}
Sine  $\arg(\lambda^2-\lambda_0^2)=\frac{\pi}{2}$ for $\lambda\in\Sigma_1$, we obtain after a change of variables $\eta^2-\lambda_0^2=i \tau$ that 
\begin{equation}\label{eq:gExp}
g(\lambda)=-\frac{s}{2\pi } \int_0^{|\lambda^2-\lambda_0^2|}\int_{\lambda_0}^{+\infty} \log\left(\frac{e^{\frac{1}{4}\pi i}\sqrt{\tau}-\sqrt{\zeta^2-\lambda_0^2}}{e^{\frac{1}{4}\pi i}\sqrt{\tau}+\sqrt{\zeta^2-\lambda_0^2}}\right) d\sigma(\sqrt{s}\zeta) d\tau.
\end{equation}
Thus, we have
\begin{equation}\label{eq:Re}
\Re g(\lambda)=-\frac{s}{2\pi } \int_0^{|\lambda^2-\lambda_0^2|}\int_{\lambda_0}^{+\infty} \log\left|\frac{e^{\frac{1}{4}\pi i}\sqrt{\tau}-\sqrt{\zeta^2-\lambda_0^2}}{e^{\frac{1}{4}\pi i}\sqrt{\tau}+\sqrt{\zeta^2-\lambda_0^2}}\right| d\sigma(\sqrt{s}\zeta) d\tau.
\end{equation}
Since both the integrand and $d\sigma(\sqrt{s}\zeta)$ are negative, we have 
\begin{equation}\label{eq:Int0}
\int_{\lambda_0}^{+\infty} \log\left|\frac{e^{\frac{1}{4}\pi i}\sqrt{\tau}-\sqrt{\zeta^2-\lambda_0^2}}{e^{\frac{1}{4}\pi i}\sqrt{\tau}+\sqrt{\zeta^2-\lambda_0^2}}\right| d\sigma(\sqrt{s}\zeta)
\geq
 \int_{\lambda_0}^{1}\log\left|\frac{e^{\frac{1}{4}\pi i}\sqrt{\tau}-\sqrt{\zeta^2-\lambda_0^2}}{e^{\frac{1}{4}\pi i}\sqrt{\tau}+\sqrt{\zeta^2-\lambda_0^2}}\right| d\sigma(\sqrt{s}\zeta).
\end{equation}
As $s\to+\infty$, we have 
\begin{equation}\label{eq:Int1}
 \int_{\lambda_0}^{1}\log\left|\frac{e^{\frac{1}{4}\pi i}\sqrt{\tau}-\sqrt{\zeta^2-\lambda_0^2}}{e^{\frac{1}{4}\pi i}\sqrt{\tau}+\sqrt{\zeta^2-\lambda_0^2}}\right| d\sigma(\sqrt{s}\zeta)= \log\left|\frac{e^{\frac{1}{4}\pi i}\sqrt{\tau}+\sqrt{1-\lambda_0^2}}{e^{\frac{1}{4}\pi i}\sqrt{\tau}-\sqrt{1-\lambda_0^2}}\right|+O(1/s),
\end{equation}
where the error term is uniform for $\tau\geq0$.
Therefore,  for $\lambda\in\Sigma_1$ and sufficiently large $s$, we have
\begin{equation}\label{eq:Reg1}
-\Re g(\lambda)\geq \frac{s}{4\pi } \int_0^{|\lambda^2-\lambda_0^2|}\log\left|\frac{e^{\frac{1}{4}\pi i}\sqrt{\tau}+\sqrt{1-\lambda_0^2}}{e^{\frac{1}{4}\pi i}\sqrt{\tau}-\sqrt{1-\lambda_0^2}}\right| d\tau\geq cs |\lambda^2-\lambda_0^2|^{\frac{3}{2}} (|\lambda|+1)^{-2},
\end{equation}
with  some constat  $c>0$. To obtain the last inequality, we use the asymptotic behaviors of the integrand near zero and positive infinity. The estimate \eqref{eq:Reg1}  also holds for $\lambda\in \Sigma_2$, which can be derived  from \eqref{eq:dgExp} in the same manner.

Substituting \eqref{eq:gsign01} and \eqref{eq:Reg1} into \eqref{eq: phi}, we obtain the estimate \eqref{eq:phiEst2}. 
From \eqref{eq:g} and \eqref{eq: phi}, we have the symmetry relation
\begin{equation}\label{eq:Symphi}
\phi(\lambda)=-\phi(-\lambda),
\end{equation}
for $\lambda\in \mathbb{C}\setminus{\mathbb{R}}$. Thus, the  estimate \eqref{eq:phiEst3}  follows directly from \eqref{eq:phiEst2} and  \eqref{eq:Symphi}. We complete the proof of the lemma.
\end{proof} 

At the end of this section, we  derive the expression for $\phi_+'(\lambda)$ with $\lambda>\lambda_0$ for later use. 
From \eqref{eq: phi}, \eqref{eq:derV} and \eqref{eq:gExp0}, we have
\begin{equation}\label{eq:dphiexp0}
\phi'_+(\lambda)=\frac{is\lambda}{\pi } \int_{\lambda_0}^{+\infty}  \log\left(\frac{\left|\sqrt{\lambda^2-\lambda_0^2}-\sqrt{\zeta^2-\lambda_0^2}\right|}{\sqrt{\lambda^2-\lambda_0^2}+\sqrt{\zeta^2-\lambda_0^2}}\right) d\sigma(\sqrt{s}\zeta),
\end{equation}
for $\lambda>\lambda_0$.
The above expression can be simplified to 
\begin{equation}\label{eq:dphiexp1}
\phi'_+(\lambda)=\frac{is\lambda}{\pi } \int_{\lambda_0}^{+\infty}  \log |\lambda^2-\zeta^2|-2\log\left(\sqrt{\lambda^2-\lambda_0^2}+\sqrt{\zeta^2-\lambda_0^2}\right) d\sigma(\sqrt{s}\zeta),
\end{equation}
for $\lambda>\lambda_0$.

\subsection{Fourth transformation: $B\to C$}

To fix the endpoints of the interval $(-\lambda_0,\lambda_0)$,  it is convenient  to introduce the transformation 
\begin{equation}\label{eq:C}
C(\lambda)=B(\lambda_0\lambda).
\end{equation}
Then, it follows from \eqref{eq:BInfinity} and \eqref{eq:Bjump1} that $C(\lambda)$ satisfies the following RH problem.

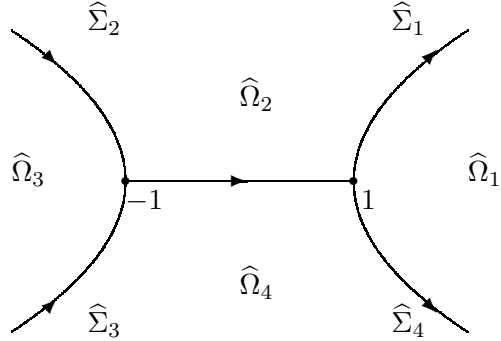
\begin{figure}[h]
\begin{center}
   \setlength{\unitlength}{1truemm}
   \begin{picture}(100,55)(-5,12)
              \put(25,40){\line(1,0){30}}
     \qbezier(10,60)(40,40)(10,20)
 \qbezier(70,60)(40,40)(70,20)

 \put(40,40){\thicklines\vector(1,0){1}}
 
       \put(15,56.5){\thicklines\vector(1,-1){1}}
       
       \put(15,23.5){\thicklines\vector(1,1){1}}
      
       \put(65,23.5){\thicklines\vector(1,-1){1}}
     
       \put(65,56.5){\thicklines\vector(1,1){1}}

       \put(60,60){$\widehat{\Sigma}_1$}
      \put(20,60){$\widehat{\Sigma}_2$}
 \put(20,20){$\widehat{\Sigma}_3$}
\put(60,20){$\widehat{\Sigma}_4$}

  \put(70,40){$\widehat{\Omega}_1$}
   \put(40,50){$\widehat{\Omega}_2$}
    \put(10,40){$\widehat{\Omega}_3$}
   \put(40,25){$\widehat{\Omega}_4$}

       \put(25,40){\thicklines\circle*{1}}
       \put(55,40){\thicklines\circle*{1}}

       \put(25,36.3){$-1$}
       \put(56,36.3){$1$}
\end{picture}
   \caption{Jump contours and regions for the RH problem for $C$.}
   \label{fig:C}
\end{center}
\end{figure}

\begin{rhp} \label{RHPC}
The function  $C(\lambda)$  satisfies the following properties.
\begin{itemize}
\item[\rm (1)]  $C(\lambda)$ is analytic in $\mathbb{C}\setminus \Sigma_C$, where  $\Sigma_C=\cup_{1}^4~\widehat{\Sigma}_{k}\cup(-1,1)$ is shown in Fig. \ref{fig:C}.

 \item[\rm (2)]  $C(\lambda)$  satisfies the jump condition
\begin{equation}\label{eq:Cjump}
C_{+}(\lambda)=C_{-}(\lambda)\left\{        \begin{array}{ll}
\left(
                               \begin{array}{cc}
                                 1 &\exp(2\phi(\lambda_0\lambda))\\
                                 0 & 1 \\
                                 \end{array}
                             \right),&   \lambda\in\widehat{\Sigma}_1\cup\widehat{\Sigma}_2,\\
                             \left(
                               \begin{array}{cc}
                                 1 &0\\
                                 - \exp(-2\phi(\lambda_0\lambda)) & 1 \\
                                 \end{array}
                             \right),&  \lambda\in\widehat{\Sigma}_3\cup\widehat{\Sigma}_4,\\
                             \left(
                               \begin{array}{cc}
                                 \exp(-2\phi_{+}(\lambda_0\lambda)) & 1 \\
                                 -1&0\\
                                 \end{array}
                             \right),&  \lambda\in(-1,1),
                                 \end{array}
                             \right.
\end{equation}
where $\phi(\lambda)$ is given in \eqref{eq: phi}.
 \item[\rm (3)]   As $\lambda\to\infty$, we have
  \begin{equation}\label{eq:CInfinity}C(\lambda)=I+O\left(\frac{1}{\lambda}\right).\end{equation}         
  \end{itemize}
  \end{rhp}

%
%
%
%
%
%
%
%
%
%
%
%
%
%

\subsection{Global parametrix}

Ignoring, for the moment, the jumps on the contours $\cup_{1}^4~\widehat{\Sigma}_{k}$ in the RH problem for $C(\lambda)$, we arrive at the 
 following approximate RH problem for $P^{(\infty)}(\lambda)$.

\begin{rhp} \label{RHPN}
The function  $P^{(\infty)}(\lambda)$  satisfies the following properties.
\begin{itemize}
\item[\rm (1)]  $P^{(\infty)}(\lambda)$ is analytic in $\mathbb{C}\setminus (-1, 1)$.

 \item[\rm (2)]  $P^{(\infty)}(\lambda)$  satisfies the jump condition
\begin{equation}\label{eq:Njump}
P^{(\infty)}_{+}(\lambda)=P^{(\infty)}_{-}(\lambda)                           \left(
                               \begin{array}{cc}
                                 0 &1\\
                              -1& 0 \\
                                 \end{array}
                             \right) \end{equation}
for $\lambda\in(-1,1)$.
 \item[\rm (3)]   As $\lambda\to\infty$, we have
  \begin{equation}\label{eq:NInfinity}P^{(\infty)}(\lambda)=I+O\left(\frac{1}{\lambda}\right).\end{equation}         
  \end{itemize}
  \end{rhp}

The solution to the above RH problem can be constructed as follows: 

 \begin{equation}\label{eq:Gp}P^{(\infty)}(\lambda)=\frac{I+i\sigma_1}{\sqrt{2}}\left(\frac{\lambda-1}{\lambda+1}\right)^{\frac{1}{4}\sigma_3}  \frac{I-i\sigma_1}{\sqrt{2}},
\end{equation}  
where $\sigma_1$  is defined in  \eqref{eq:PauliMatrices} and $\left(\frac{\lambda-1}{\lambda+1}\right)^{\frac{1}{4}} $ takes the branch cut along $[-1,1]$ and behaves like $1$ as $\lambda\to\infty$. 
From \eqref{eq:Gp},  we have the asymptotic expansion for $P^{(\infty)}(\lambda)$   as $\lambda\to\infty$
   \begin{equation}\label{eq:GpExpand}P^{(\infty)}(\lambda)
  =I-\frac{\sigma_2}{2\lambda}  +O\left(\frac{1}{\lambda^2}\right),  
\end{equation}  
where $ \sigma_2$ is defined in  \eqref{eq:PauliMatrices}.
\subsection{Local parametrix}

Let $U(\pm1,\delta)$ be an open disc centered at $\pm 1$ with radius  $\delta>0$. We look for the
 parametrices satisfying the same jump condition as $C(\lambda)$ in $U(\pm1,\delta)$ and matching with $P^{(\infty)}(\lambda)$ on the boundaries $\partial U(\pm 1,\delta)$. 
 
\begin{rhp} \label{RHPl}
We look for a matrix-valued  function  $ P^{(\pm1)}(\lambda)$  satisfying the following properties.
\begin{itemize}
\item[\rm (1)]  $P^{(\pm 1 )}(\lambda)$ is analytic in $U(\pm 1,\delta)\setminus \Sigma_{C}$.

 \item[\rm (2)]  $P^{(\pm1)}(\lambda)$  satisfies the same  jump condition as $C(\lambda)$ on $\Sigma_{C} \cap U(\pm 1,\delta)$.
 \item[\rm (3)]   On $\partial U(\pm  1,\delta)$, we have as $s\to+\infty$
  \begin{equation}\label{eq:Math}P^{(\pm 1)}(\lambda)=\left(I+O\left(\frac{1}{s\lambda_0^3}\right)\right)P^{(\infty)}(\lambda),\end{equation}         
where $\lambda_0$ is given in   \eqref{eq:lambda0Est}.  \end{itemize}
  \end{rhp}

To construct the solution to the above RH problem, we define
\begin{equation}\label{eq:f}
f(\lambda)
=\left\{
                               \begin{array}{cc}
                             \left(\frac{3}{2}\phi(\lambda_0\lambda)\right)^{\frac{2}{3}} & \Im \lambda>0,\\
                              
                           \left(-\frac{3}{2}\phi(\lambda_0\lambda)\right)^{\frac{2}{3}}  & \Im \lambda<0.
                                 \end{array}
                             \right.
\end{equation}  
From \eqref{eq:phisign1} and \eqref{eq:coef1}, we  choose the branch in \eqref{eq:f} such that $f(\lambda)$ is a  conformal mapping in $U(-1, \delta)$:
\begin{equation}\label{eq:flocal2}
f(\lambda) \sim 2 \lambda_0^2\left(\frac{3}{2}|c_1(s)|\right)^{\frac{2}{3}}(\lambda+1), \quad \lambda\to- 1,
\end{equation}
with $c_1(s)$ given in \eqref{eq:coef1}. It follows from \eqref{eq:coef1} that
\begin{equation}\label{eq:fcoeff}
 2 \lambda_0^2\left(\frac{3}{2}|c_1(s)|\right)^{\frac{2}{3}} \sim  2\pi^{-\frac{2}{3}} (1-\lambda_0^2)^{-\frac{1}{3}}  s^{\frac{2}{3}} \lambda_0^2,
\end{equation}
as $s\to+\infty$.
From  \eqref{eq:lambda0Est}, we have as $s\to+\infty$
\begin{equation}\label{eq:sLambda0}
s^{\frac{2}{3}}\lambda_0^2 \sim s^{\frac{2}{3}}(1+\frac{\pi x}{2\sqrt{s}})(1-\frac{\pi x}{2\sqrt{s}}).
\end{equation}
Therefore, in the regime we consider, we have 
\begin{equation}\label{eq:sLambda1}
s^{\frac{2}{3}}\lambda_0^2 \geq s^{\frac{2}{3}}(1-\frac{\pi x}{2\sqrt{s}})\geq M_1,\end{equation}
 where $M_1>0$ is sufficiently large.
 
The local parametrix near $-1$ can be constructed as follows:
 \begin{equation}\label{eq:Pl}
 P^{(-1)}(\lambda)=E(\lambda)\Phi^{(\Ai)}(f(\lambda))e^{-\phi(\lambda_0\lambda)\sigma_3},
\end{equation}  
where
 \begin{equation}\label{eq:E}
E(\lambda)=\sqrt{2} P^{(\infty)}(\lambda) \begin{pmatrix}i & 1\\ -i &1\end{pmatrix}^{-1}
f(\lambda)^{\frac{\sigma_{3}}{4} },\end{equation}  
and the branch of $f(\lambda)^{\frac{1}{4} }$ is chosen such that $\arg f(\lambda)\in(0,2\pi)$.

\begin{figure}[h]
\begin{center}
   \setlength{\unitlength}{1truemm}
   \begin{picture}(70,65)(-5,8)
       
       \put(25,40){\line(1,0){40}}
            \put(25,40){\line(-1,-1){25}}
       \put(25,40){\line(-1,1){25}}
     
       \put(10,55){\thicklines\vector(1,-1){1}}
       \put(10,25){\thicklines\vector(1,1){1}}
     \put(50,40){\thicklines\vector(1,0){1}}
       \put(-2,11){$\Sigma_3^{(\Ai)}$}

       \put(-2,67){$\Sigma_{2}^{(\Ai)}$}
     
       \put(25,40){\thicklines\circle*{1}}
       \put(40,55){$\mathtt{I}$}
            \put(0,40){$\mathtt{II}$}
          \put(40,25){$\mathtt{III}$}

       \put(25,36.3){$0$}
       \put(54,35){$\Sigma_1^{(\Ai)}$}
\end{picture}
   \caption{Jump contours and regions for the RH problem for $\Phi^{(\Ai)}$.}
   \label{fig:Airy}
\end{center}
\end{figure}
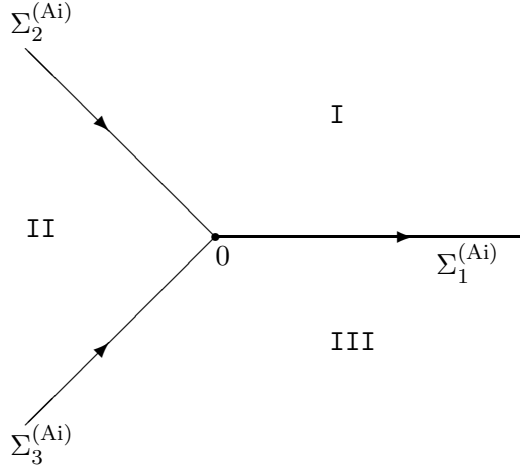

Here, $\Phi^{(\Ai)}(\lambda)$ satisfies the following RH problem.
\begin{rhp}
\item[\rm (1)] $\Phi^\mathrm{(Ai)}(\lambda)$ is analytic for $\lambda\in \mathbb{C}\setminus \cup^3_{k=1}\Sigma_{k}^{(\Ai)}$,  
where the contours are depicted in Fig. \ref{fig:Airy}. \item[\rm (2)] $\Phi^\mathrm{(Ai)}(\lambda)$ satisfies the jump condition 
\begin{equation}\label{eq: AiryJump}\Phi^\mathrm{(Ai)}_+(\lambda)=\Phi^\mathrm{(Ai)}_-(\lambda)J_k, \quad \lambda\in\Sigma_{k}^{(\Ai)}, 
~k=1,2,3,\end{equation} 
 where
\begin{equation*}
J_1=\begin{pmatrix}1 & 1 \\ -1 & 0 \end{pmatrix},\
J_2=\begin{pmatrix}1 & 1 \\ 0 & 1 \end{pmatrix},  \ J_3=\begin{pmatrix}1 & 0 \\ -1 & 1 \end{pmatrix}.
\end{equation*}
\item[\rm (3)] 
As $\lambda\to\infty$, we have 
\begin{equation}\label{eq: Airysinfty}
\Phi^{(\mathrm{Ai})}(\lambda)=\lambda^{-\frac{\sigma_{3}}{4} }\frac{1}{\sqrt{2}} \begin{pmatrix}i & 1\\ -i &1\end{pmatrix}\left(I
+O\left(\frac{1}{\lambda^{\frac{3}{2}}}\right)
\right) 
e^{\frac{2}{3} \lambda^{\frac{3}{2}} \sigma_{3}},
\end{equation}
where the branch cut for  the fractional powers is taken along $(0,+\infty)$ with 
$\arg (\lambda)\in(0,2\pi)$, and the error term is uniform  for $\lambda$ in a full neighborhood of infinity. 
\end{rhp}

The above RH problem is related to the standard Airy parametrix $\Phi^{(\mathrm{Ai})}_0(\lambda)$ in \cite[Eq. (7.9)]{DMVZ2} by the following elementary transformation:
\begin{equation}\label{eq: AiryP0} 
\Phi^{(\mathrm{Ai})}(\lambda)=\sqrt{2\pi}e^{\frac{1}{6}\pi i}\sigma_3\Phi^{(\mathrm{Ai})}_0(\lambda)
\left\{
                               \begin{array}{cc}
                              \sigma_1,& \Im\lambda>0,\\
                              \sigma_3,& \Im\lambda<0,
                                 \end{array}
                             \right.
\end{equation}
with the Pauli matrices $\sigma_1$ and $\sigma_3$ defined in  \eqref{eq:PauliMatrices}.
 Therefore, the solution to the above RH problem can be constructed  with the Airy functions
\begin{equation}\label{eq: AiryP} 
\Phi^{(\mathrm{Ai})}(\lambda)=\sqrt{2\pi}e^{\frac{1}{6}\pi i}\begin{pmatrix}
\mathrm{Ai}(e^{-\frac{2}{3}\pi i}\lambda) & \mathrm{Ai}(\lambda) \\
-e^{-\frac{2}{3}\pi i} \mathrm{Ai}^{\prime}(e^{-\frac{2}{3}\pi i}\lambda) &- \mathrm{Ai}^{\prime}(\lambda)
\end{pmatrix} e^{i \frac{\pi}{6} \sigma_{3}}
\left\{
                               \begin{array}{cc}
                              I,& \lambda\in \mathtt{I},\\
                              J_2^{-1},& \lambda\in \mathtt{II},\\
                           J_1^{-1},& \lambda\in \mathtt{III};
                                 \end{array}
                             \right.
\end{equation}
see also \cite[Section 3.7]{B} and \cite[Section 4.4]{CCG} for a similar construction.

It follows from \eqref{eq:Njump} and \eqref{eq:E} that $E(\lambda)$ is analytic for $\lambda\in U(-1,\delta)$. This, together with \eqref{eq: AiryJump}, implies that $P^{(-1)}(\lambda)$ satisfies the  same  jump condition as $C(\lambda)$ on $\Sigma_{C} \cap U(- 1,\delta)$ as given in \eqref{eq:Cjump}.  
On account of  \eqref{eq:flocal2}-\eqref{eq:sLambda1}, we may choose a sufficiently large constant $M_1$ in  \eqref{eq:sLambda1}  such that $|f(\lambda)|$ is sufficiently large for $ \lambda\in\partial U(- 1,\delta)$  as $s\to+\infty$.
Therefore,  the function $\Phi^{(\Ai)}(f(\lambda))$ can be expanded according to  \eqref{eq: Airysinfty} on the boundary $ \partial U(- 1,\delta)$. 
From \eqref{eq:Pl}, \eqref{eq:E}   and \eqref{eq: Airysinfty}, we then obtain the matching condition \eqref{eq:Math}.
Thus, $ P^{(-1)}(\lambda)$ constructed in \eqref{eq:Pl}  satisfies the  required properties.

It follows from  \eqref{eq:Cjump},  \eqref{eq:CInfinity} and  \eqref{eq:Symphi}  that 
 \begin{equation}\label{eq:SymC} \sigma_1C(-\lambda)\sigma_1=C(\lambda),\end{equation}
 where $\sigma_1$ is one of the Pauli matrices defined in  \eqref{eq:PauliMatrices}.
  Therefore, the local parametrix in  $U(1,\delta)$ can be constructed as follows
 \begin{equation}\label{eq:Pr}
 P^{(1)}(\lambda)=\sigma_1 P^{(-1)}(-\lambda)\sigma_1,  \quad \lambda\in U(1,\delta).
\end{equation}

\subsection{Small-norm  RH problem}
We define    
\begin{equation}\label{eq:Rlarges}
R(\lambda) =\left\{
                               \begin{array}{cc}
                              C(\lambda)P^{(-1)}(\lambda)^{-1},& |\lambda+1|<\delta,\\
                               C(\lambda)P^{(1)}(\lambda)^{-1},& |\lambda-1|<\delta,\\
                           C(\lambda)P^{(\infty)}(\lambda)^{-1},& |\lambda\pm 1|>\delta.\\
                                 \end{array}
                             \right.
\end{equation}
Then, $R(\lambda)$ satisfies the RH problem below. 
  \begin{rhp}
The function  $R(\lambda)$  satisfies the following properties.
\begin{itemize}
\item[\rm (1)]  $R(\lambda)$ is analytic in $\mathbb{C}\setminus\Sigma_{R}$, with $\Sigma_{R}=\partial U(-1, \delta)\cup \partial U(1, \delta)\cup (\Sigma_C\setminus \{U(-1, \delta)\cup U(1, \delta)\})$.

 \item[\rm (2)]  $R(\lambda)$  satisfies the jump condition
\begin{equation}\label{eq:Rjump2}
R_{+}(\lambda)=R_{-}(\lambda)J_{R}(\lambda), 
%
%
\end{equation}
with 
\begin{equation}\label{eq:Rjump2}
 J _{R}(\lambda)= \left\{        \begin{array}{cc}

                            P^{(\pm1)}(\lambda)P^{(\infty)}(\lambda)^{-1} ,&   \lambda\in\partial U(\pm1, \delta),\\

                            P^{(\infty)}(\lambda)J_C(\lambda)P^{(\infty)}(\lambda)^{-1}, &   \lambda\in \Sigma_C\setminus \{U(-1, \delta)\cup U(1, \delta)\}. 
                                 \end{array}
                             \right.
\end{equation}
 \item[\rm (3)]   As $\lambda\to\infty$, we have
       \begin{equation}\label{eq:RInfinity2} R(\lambda)=I+O\left(\frac{1}{\lambda}\right).\end{equation}                
  \end{itemize}
  \end{rhp}

It follows  from \eqref{eq:phiEst1}-\eqref{eq:phiEst3} and \eqref{eq:Math} that
 \begin{equation}\label{eq:RlargesJump}
J_{R}(\lambda) =I+O\left(\frac{1}{s\lambda_0^3(1+|\lambda|^2)}\right),
\end{equation}
for $\lambda\in \Sigma_{R}$.   It follows  from \eqref{eq:sLambda1} that $s\lambda_0^3\geq M_1^{\frac{3}{2}}$. Therefore,  for sufficiently large $M_1$,  the error term in \eqref{eq:RlargesJump} is small.
By  similar arguments for small-norm RH problems in  \cite{DZ, Deift}, we have as $s\to+\infty$
\begin{equation}\label{eq:Rasylarges}
R(\lambda)=I +O\left(\frac{1}{s\lambda_0^3(1+|\lambda|)}\right),
\end{equation}
and 
\begin{equation}\label{eq:dRasylarges}
\frac{d}{d\lambda}R(\lambda)=O\left(\frac{1}{s\lambda_0^3(1+|\lambda|^2)}\right),
\end{equation}
 where the error term is uniform for $\lambda\in \mathbb{C}\setminus \Sigma_{R}$  and $ \frac{\pi }{2 }\frac{x}{\sqrt{s} }\in [M_2, 1-M_1s^{-\frac{2}{3}}]$ with  a sufficiently large constant $M_1>0$ and any constant $0<M_2<1$.

\section{Large $s$ asymptotics of $Y$ as $\frac{\pi}{2}\frac{x}{\sqrt{s}}\to 1$}\label{sec:PIIasy}

\subsection{Case: $L(s)\leq  \frac{\pi}{2} \frac{x}{\sqrt{s}}\leq M$ with any constant $M>1$}

We consider in this section the case $L(s)\leq  \frac{\pi}{2} \frac{x}{\sqrt{s}}\leq M$, with  any constant $M>1$ and $L(s)\sim 1$ as given in \eqref{eq:L} and \eqref{eq:LAsy}.  The  first few transformations $Y\to T\to A\to B$ are defined in the same way as those given in \eqref{eq:T}, \eqref{eq:A} and \eqref{eq:B}, with the  endpoint $\lambda_0=0$  therein.
 
Similar to  \eqref{eq:gInfty} and \eqref{eq:phase}, the $g$-function is analytic in $\mathbb{C}\setminus\mathbb{R}$ and satisfies  the following properties:
\begin{equation}\label{eq:doublegphase}
g_+(\lambda)-g_{-}(\lambda)-V(\lambda)+V(0)=0, \quad \lambda\in\mathbb{R},
\end{equation}
with $V(\lambda)$ defined in \eqref{eq:V} and 
\begin{equation}\label{eq:gdoubleinfty}
g(\lambda)=ix\sqrt{s}\lambda+O(1/\lambda), \quad \lambda\to\infty. \end{equation}

\subsubsection{Properties of the $g$-function and $\phi$-function}

We derive some useful properties of the $g$-function defined by \eqref{eq:doublegphase} and \eqref{eq:gdoubleinfty}, and of the associated $\phi$-function.

\begin{lem}Suppose $g(\lambda)$ satisfies \eqref{eq:doublegphase} and \eqref{eq:gdoubleinfty}, 
then we have 
\begin{equation}\label{eq:gdouble}
g(\lambda)
=i\lambda\left(x\sqrt{s}-\frac{s}{\pi }\int_{\mathbb{R}}\sigma(\sqrt{s}\zeta;s)d\zeta\right)+\frac{s}{\pi i} \int_0^{\lambda}\eta\left(\int_{\mathbb{R}}\frac{\sigma(\sqrt{s}\zeta;s)}{\zeta-\eta}d\zeta\right)d\eta,\end{equation}
\begin{equation}\label{eq:dgdouble}
g'(\lambda)
=i\left(x\sqrt{s}-\frac{s}{\pi }\int_{\mathbb{R}}\sigma(\sqrt{s}\zeta;s)d\zeta\right)+\frac{s\lambda}{\pi i}\int_{\mathbb{R}}\frac{\sigma(\sqrt{s}\zeta;s)}{\zeta-\lambda}d\zeta,\end{equation}
 where $\lambda\in\mathbb{C}\setminus\mathbb{R}$, and in \eqref{eq:gdouble} the path of integration from $0$ to $\lambda$ does not cross the real axis. 
Moreover, we have  as $\lambda\to\infty$
\begin{equation}\label{eq:gdoubleExpand}
g(\lambda)=ix\sqrt{s}\lambda+\frac{s}{\pi i \lambda} \int_{\mathbb{R}}\zeta^2\sigma(\sqrt{s}\zeta;s)d\zeta+O(1/\lambda^3). \end{equation}
\end{lem}
\begin{proof}
From \eqref{eq:derV},  \eqref{eq:doublegphase} and \eqref{eq:gdoubleinfty},  we have
\begin{equation}\label{eq:g'phase}
g'_+(\lambda)-g'_{-}(\lambda)=2s\lambda\sigma(\sqrt{s}\lambda;s), \quad \lambda\in\mathbb{R},\end{equation}
and
\begin{equation}\label{eq:g'Infty}
g'(\lambda)=ix\sqrt{s}+O(1/\lambda^2), \quad \lambda\to\infty.
\end{equation}
By the Sokhotski-Plemelj formula, we have \eqref{eq:dgdouble}. 
From \eqref{eq:doublegphase} and \eqref{eq:gdoubleinfty}, it follows  that 
\begin{equation}\label{eq:gSym}
g(\lambda)=-g(-\lambda), \quad \lambda\in\mathbb{C}\setminus\mathbb{R}. \end{equation}
Therefore,  integrating \eqref{eq:dgdouble} yields \eqref{eq:gdouble}.
Using \eqref{eq:sigma},   \eqref{eq:dgdouble} and \eqref{eq:gSym}, we obtain \eqref{eq:gdoubleExpand}.
\end{proof}

It is noted that the expression for $g'$  in \eqref{eq:dgdouble} is consistent with \eqref{eq:dg} when $\lambda_0=0$.

\begin{lem}\label{lem:gdoubleEst}
For  large enough $s$, there exists some constant $c>0$ such that
\begin{equation}\label{eq:gdoubleEst1}
\Re g(\lambda)
\leq  -cs |\lambda|^{3} (|\lambda|+1)^{-2}, \quad \arg\lambda=\frac{\pi}{2}\pm \frac{\pi}{4},\end{equation}
and
\begin{equation}\label{eq:gdoubleEst2}
\Re g(\lambda)
\geq cs  |\lambda|^{3} (|\lambda|+1)^{-2}, \quad  \arg\lambda=-\frac{\pi}{2}\pm \frac{\pi}{4}.
\end{equation} 
\end{lem}
\begin{proof}
We derive the estimates by using the  expression for $g(\lambda)$ given in \eqref{eq:gdouble}. 
We first estimate the first integral in  \eqref{eq:gdouble}. Using the definition of $L(s)$ given in \eqref{eq:L}, we have
\begin{equation}\label{eq:gTerm2}
x\sqrt{s}-\frac{s}{\pi }\int_{\mathbb{R}}\sigma(\sqrt{s}\zeta;s)d\zeta =x\sqrt{s}-\frac{2s}{\pi }L(s).
\end{equation}
Under the condition $L(s)\leq  \frac{\pi}{2} \frac{x}{\sqrt{s}}$ in this case, we have 
\begin{equation}\label{eq:gTerm3}
x\sqrt{s}-\frac{s}{\pi }\int_{\mathbb{R}}\sigma(\sqrt{s}\zeta;s)d\zeta \geq 0. 
\end{equation}
Therefore, we have for $\Im\lambda>0$
\begin{equation}\label{eq:gdoubleEst10}
\Re \left(x\sqrt{s}-\frac{s}{\pi }\int_{\mathbb{R}}\sigma(\sqrt{s}\zeta;s)d\zeta\right) i \lambda 
 \leq 0.\end{equation}

Next, we  estimate the second integral  in  \eqref{eq:gdouble}. Denote
\begin{equation}\label{eq:g0}
g_0(\lambda)=\frac{s}{\pi i} \int_0^{\lambda}\eta\left(\int_{\mathbb{R}}\frac{\sigma(\sqrt{s}\zeta;s)}{\zeta-\eta}d\zeta\right)d\eta=\frac{2s}{\pi i}\int_0^{\lambda} \left(\eta^2
\int_{0}^{+\infty}\frac{\sigma(\sqrt{s}\zeta;s)}{\zeta^2-\eta^2}d\zeta \right)d\eta.
\end{equation}
Using integration by parts, we have
\begin{equation}\label{eq:dg0}
g'_0(\lambda)=-\frac{s\lambda}{\pi i }\int_{0}^{+\infty} \log\left(\frac{\lambda-\zeta}{\lambda+\zeta}\right) d\sigma(\sqrt{s}\zeta),
\end{equation}
where  $\arg(\lambda\pm\zeta)\in(-\pi,\pi)$.
Therefore, we have
\begin{equation}\label{eq:g01}
g_0(\lambda)=\frac{is}{\pi } \int_{0}^{\lambda}\int_{0}^{+\infty} \eta \log\left(\frac{\eta-\zeta}{\eta+\zeta}\right)d\sigma(\sqrt{s}\zeta) d\eta,
\end{equation}
where  $\arg(\eta\pm\zeta)\in(-\pi,\pi)$.
Similar to the estimate of   \eqref{eq:gExp0}, we have 
\begin{equation}\label{eq:Reg0}
\Re g_0(\lambda)\leq -cs  |\lambda|^{3} (|\lambda|+1)^{-2},
\end{equation}
where $\arg \lambda=\frac{\pi}{2}\pm \frac{\pi}{4}$ and $c$ is  some positive constant.  Combining \eqref{eq:gdoubleEst10} with \eqref{eq:Reg0}, we have  \eqref{eq:gdoubleEst1}. The estimate \eqref{eq:gdoubleEst2} then follows from \eqref{eq:gdoubleEst1} and the symmetry relation \eqref{eq:gSym}.
\end{proof}

We define for $\pm\Im \lambda>0$
\begin{equation}\label{eq:phi0}
\phi_0(\lambda)=g_0(\lambda)\mp\frac{1}{2}(V(\lambda)-V(0))=
2s\int_0^{\lambda}\eta\left(\frac{1}{2\pi i}\int_{\mathbb{R}} \frac{\sigma(\sqrt{s}\zeta;s)}{\zeta-\eta}d\zeta\mp\frac{\sigma(\sqrt{s}\eta;s)}{2}\right) d\eta.
\end{equation}

\begin{lem}
Let $\phi_0(\lambda)$ be defined by \eqref{eq:phi0}, we have as $\lambda\to0$
\begin{equation}\label{eq:dphi0Asy}
\phi_0(\lambda)=c_0(s)\lambda^3(1+O(\lambda)),
\end{equation}
where the coefficient
\begin{equation}\label{eq:c3}
c_0(s)= \frac{2i}{3\pi} s+O(1/s), \quad s\to+\infty.\end{equation}
\end{lem}
\begin{proof}
It follows from \eqref{eq:phi0} that $\phi_0(\lambda)$ is analytic near the origin and $\phi_0(0)=0$. To calculate the leading asymptotics of $\phi_0(\lambda)$ as $\lambda\to0$, it is convenient to consider the derivative of  $\phi_0(\lambda)$
\begin{equation}\label{eq:Intphi}
\phi_0'(\lambda)=2s\lambda\left(\frac{1}{2\pi i}\int_{\mathbb{R}} \sigma(\sqrt{s}\zeta;s)\frac{d\zeta}{\zeta-\lambda}-\frac{\sigma(\sqrt{s}\lambda;s)}{2}\right),
\end{equation}
for $\Im \lambda>0$. We rewrite the integral in the above equation as 
\begin{equation}\label{eq:Intphi2}
I_0(\lambda)=\frac{1}{2\pi i}\int_{\mathbb{R}} \sigma(\sqrt{s}\zeta;s)\frac{d\zeta}{\zeta-\lambda}-\frac{\sigma(\sqrt{s}\lambda;s)}{2}=\frac{1}{2\pi i}\int_{\Gamma}\sigma(\sqrt{s}\zeta;s)\frac{d\zeta}{\zeta-\lambda}- \frac{\sigma(\sqrt{s}\lambda;s)}{2},
\end{equation}
where the integral path is along the real axis with infinitesimal semicircular  centered at the origin $\Gamma=(-\infty, -r]\cup \{r e^{\pi i\theta}: \theta\in(-1,0) \}\cup[r,+\infty)$. 
From \eqref{eq:sigma} and \eqref{eq:Intphi2}, we have
\begin{equation}\label{eq:dvphi0}
I_0(0)=0 \qquad\mbox{and}  \qquad I'_0(0)=\frac{1}{2\pi i }\int_{\Gamma} \sigma(\sqrt{s}\zeta;s)\frac{d\zeta}{\zeta^2}.
\end{equation}
For sufficiently small $r>0$, we may replace $\sigma(\sqrt{s}\zeta;s)$ by the characteristic function on $\Gamma \cap U(0,1)$ in \eqref{eq:dvphi0} and obtain as $s\to+\infty$
\begin{equation}\label{eq:dvphiAsy}
I_0'(0)=-\frac{1}{\pi i}+O(1/s^2).
\end{equation}
Therefore, we arrive at the expansion \eqref{eq:dphi0Asy} with the coefficient given in \eqref{eq:c3}. This completes the proof of the lemma.
\end{proof}

\subsubsection{Local parametrix near the origin}
It is seen from \eqref{eq:gdoubleEst1}, \eqref{eq:gdoubleEst2} and \eqref{eq:gsign01} (with $\lambda_0=0$) that the jumps described in \eqref{eq:Bjump} approach  the identity matrix exponentially fast for $\lambda$ bounded away from the origin.
The next task is to construct a local parametrix  near the origin, which satisfies the same jump condition as $B(\lambda)$.

For this purpose, using \eqref{eq:dphi0Asy} we define the following conformal mapping in $U(0,\delta)$ with some constant $0<\delta<1$:
 \begin{equation}\label{eq:f0}
f_0(\lambda)=\left(-\frac{3}{4}\phi_{0}(\lambda)i\right)^{\frac{1}{3}}= a_0(s)\lambda\left(1+O(\lambda)\right),
\end{equation}
with 
 \begin{equation}\label{eq:coeffc1}a_0(s)=(2\pi)^{-\frac{1}{3}}s^{\frac{1}{3}}+O(s^{-5/3}), \quad s\to+\infty.\end{equation}
We also introduce 
\begin{equation}\label{eq:h}
h(\lambda;x, s)=\left(x\sqrt{s}-\frac{s}{\pi} \int_{\mathbb{R}}\sigma(\sqrt{s}\zeta;s)d\zeta\right)\frac{\lambda}{f_0(\lambda)}.\end{equation}
Then, $h(\lambda)=h(\lambda;x, s)$ is analytic for $\lambda\in U(0,\delta)$.  From \eqref{eq:L}, \eqref{eq:LAsy}  and  \eqref{eq:coeffc1}, it follows that
\begin{equation}\label{eq:tau0}
\tau(x,s)=h(0; x,s)=4(2\pi)^{-\frac{2}{3}}s^{\frac{2}{3}}(\frac{\pi x}{2\sqrt{s}}-1)+O(s^{-4/3}),\end{equation}
and 
\begin{equation}\label{eq:dtau0}
\partial_{x}\tau(x,s)=\frac{\sqrt{s}}{a_0(s)}=(2\pi)^{\frac{1}{3}}s^{\frac{1}{6}}\left(1+O(s^{-2})\right),
\end{equation}
as $s\to+\infty$.
Moreover, from \eqref{eq: phi}, \eqref{eq:gdouble} and  \eqref{eq:phi0}, we have  for $|\lambda|=\delta$
\begin{align}\label{eq:phase3}
i\frac{4}{3}f_0(\lambda)^3+ih(\lambda)f_0(\lambda)- \phi(\lambda)&=
i\frac{4}{3}f_0(\lambda)^3+ih(\lambda)f_0(\lambda)-\phi_0(\lambda)-i\left(x\sqrt{s}-\frac{s}{\pi} \int_{\mathbb{R}}\sigma(\sqrt{s}\zeta;s)d\zeta\right)\lambda
\nonumber\\
&=ih(\lambda)f_0(\lambda)-i\left(x\sqrt{s}-\frac{s}{\pi} \int_{\mathbb{R}}\sigma(\sqrt{s}\zeta;s)d\zeta\right)\lambda\nonumber\\
&=0,
\end{align}
where $\phi(\lambda)$ is defined by \eqref{eq: phi} with $\lambda_0=0$ therein.

 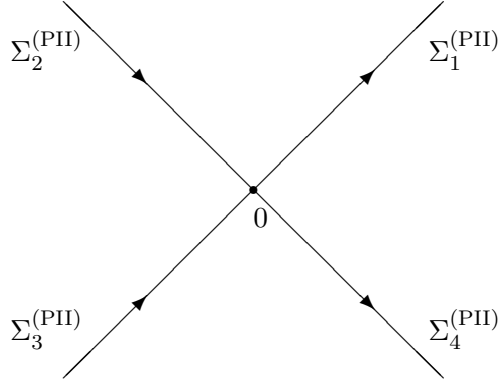
\begin{figure}[h]
\begin{center}
   \setlength{\unitlength}{1truemm}
   \begin{picture}(100,55)(-5,10)
      
       \put(40,40){\line(-1,-1){25}}
       \put(40,40){\line(-1,1){25}}
       \put(40,40){\line(1,1){25}}
       \put(40,40){\line(1,-1){25}}

       \put(25,55){\thicklines\vector(1,-1){1}}
       \put(25,25){\thicklines\vector(1,1){1}}
       \put(55,25){\thicklines\vector(1,-1){1}}
       \put(55,55){\thicklines\vector(1,1){1}}

       \put(8,20){$\Sigma_3^{(\PII)}$}
       \put(8,57){$\Sigma_2^{(\PII)}$}
       \put(63,20){$\Sigma_4^{(\PII)}$}
       \put(63,57){$\Sigma_1^{(\PII)}$}
      
       \put(40,40){\thicklines\circle*{1}}

       \put(40,35){$0$}
     
\end{picture}
   \caption{Jump contours for the RH problem for $\Phi^{(\PII)}$.}
   \label{fig:PII}
\end{center}
\end{figure}

The parametrix near the origin can be constructed as follows:
\begin{equation}\label{eq:P02}
 P^{(0)}(\lambda)=\sigma_1\Phi^{(\PII)}(f_0(\lambda),h(\lambda;x,s))\sigma_1 e^{-\phi(\lambda)\sigma_3}, \quad \lambda\in U(0,\delta),
 \end{equation}
 where $f_0$ and $h$ are defined in  \eqref{eq:f0} and  \eqref{eq:h}, respectively. 
 Here, $\Phi^{(\PII)}(\lambda,t)$ satisfies the following RH problem.

  \begin{rhp}\label{RH: PII} The function $\Phi^{(\PII)}(\lambda)=\Phi^{(\PII)}(\lambda,t)$ satisfies the properties below. 
\begin{itemize}
\item[\rm (1)]  $\Phi^{(\PII)}(\lambda)$ is analytic in $\mathbb{C}\setminus\{\cup_{k=1}^4\Sigma_{k}^{(\PII)}\}$, where the contours are shown in Fig. \ref{fig:PII}.

 \item[\rm (2)] $\Phi^{(\PII)}(\lambda)$ satisfies the jump condition
 \begin{equation}\label{eq:PhiJumps}
  \Phi^{(\PII)}_{+}(\lambda)=  \Phi^{(\PII)}_{-}(\lambda)\left\{
                               \begin{array}{cc}
                                 \left(
                               \begin{array}{cc}
                                 1 &0\\
                              1& 1 \\
                                 \end{array}
                             \right),  &  \lambda\in\Sigma_1^{(\PII)}\cup\Sigma_2^{(\PII)},\\
                              \left(
                               \begin{array}{cc}
                                 1 &-1\\
                             0 & 1 \\
                                 \end{array}
                             \right),  &  \lambda\in\Sigma_3^{(\PII)}\cup\Sigma_4^{(\PII)}.
                                 \end{array}
                             \right.
\end{equation}
  \item[\rm (3)]   As $\lambda\to\infty$, we have
 \begin{equation}\label{eq:PsiInfty}
  \Phi^{(\PII)}(\lambda)=\left(I+\frac{\Phi_{1}(t)}{\lambda}+O\left (\frac 1 {\lambda^2}\right )\right) \exp\left(-i\left(\frac{4}{3}\lambda^3+t\lambda\right)\sigma_3\right),
 \end{equation}
 where \begin{equation}\label{eq:Psi1}
  \Phi_1(t)=-\frac{i}{2}H(t)\sigma_3+\frac{1}{2}u(t)\sigma_2, \qquad H(t)=u'(t)^2-u(t)^4-tu(t)^2,
 \end{equation}
 with $\sigma_2$ and $\sigma_3$ defined in \eqref{eq:PauliMatrices}.
  \end{itemize}
  \end{rhp}

 \begin{rem}\label{rem:PIIPro}  The solution to the above model RH problem can be constructed by using the $\psi$-functions associated with the  second Painlev\'e equation \cite{CK}. The function $u(t)$ appeared in \eqref{eq:Psi1} is  the Hastings-McLeod solution of  \eqref{eq:PII}, which is exponentially small as $t\to+\infty$; see \cite{BDJ, BI, CK, CKV, FIKY}.   The Hastings-McLeod solution  $u(t)$ is pole-free in the region  $\arg t\in [-\frac{\pi}{3}, \frac{\pi}{3}]\cup [\frac{2\pi}{3}, \frac{4\pi}{3}]$ as shown in \cite{B, HXZ}.  Consequently, the  model RH problem for $\Phi^{(\PII)}(\lambda, t)$ is solvable for $t$ in this region.
 From the large $t$ asymptotics of $  \Phi^{(\PII)}(\lambda;t)$ in \cite[Chapter 11]{FIKY},   we  see that the expansion \eqref{eq:PsiInfty} is uniform for $|\arg(t)|< \frac{\pi}{3}-\delta$ for any constant $0<\delta<\frac{\pi}{3}$. \end{rem}

 From \eqref{eq:f0}-\eqref{eq:h} and \eqref{eq:gTerm3}, there exists a sufficiently small constant $\delta>0$ such that $|\arg h(\lambda)|< \frac{\pi}{6}$ for  $|\lambda|< \delta$ and $L(s)\leq  \frac{\pi}{2} \frac{x}{\sqrt{s}}\leq M$ with $M>1$.   Therefore, according to Remark \ref{rem:PIIPro},  $h(\lambda)$ stays away from the poles of the Hastings-McLeod solution, ensuring that $\Phi^{(\PII)}(f_0(\lambda),h(\lambda))$ exists for $|\lambda|< \delta$. Moreover, it follows from \eqref{eq:f0}, \eqref{eq:P02} and \eqref{eq:PhiJumps} that $ P^{(0)}(\lambda)$ satisfies the same jump condition as $B(\lambda)$ in \eqref{eq:Bjump1} with $\lambda_0=0$ therein for $\lambda\in\Sigma_{B} \cap U(0,\delta)$. 

\subsubsection{Small-norm  RH problem}
 We define   
\begin{equation}\label{eq:R2}
R(\lambda) =\left\{
                               \begin{array}{cc}
                              B(\lambda)P^{(0)}(\lambda)^{-1},& |\lambda|<\delta,\\
                           B(\lambda), & |\lambda|>\delta.\\
                                 \end{array}
                             \right.
\end{equation}
Then, $R(\lambda)$ satisfies the RH problem below. 
  \begin{rhp}\label{RH:double}
The function  $R(\lambda)$  satisfies the following properties.
\begin{itemize}
\item[\rm (1)]  $R(\lambda)$ is analytic in $\mathbb{C}\setminus\Sigma_{R}$, with $\Sigma_{R}=\partial U(0, \delta)\cup
 \left(\Sigma_B\setminus U(0, \delta)\right)$.

 \item[\rm (2)]  $R(\lambda)$  satisfies the jump condition
\begin{equation}\label{eq:Rjump3}
R_{+}(\lambda)=R_{-}(\lambda)J_{R}(\lambda), \quad J _{R}(\lambda)= \left\{        \begin{array}{cc}

                            P^{(0)}(\lambda),&   \lambda\in\partial U(0, \delta),\\

                           J_B(\lambda) &   \lambda\in \Sigma_B\setminus U(0, \delta). 
                                 \end{array}
                             \right.
\end{equation}
 \item[\rm (3)]   As $\lambda\to\infty$, we have
       \begin{equation}\label{eq:RInfinity3} R(\lambda)=I+O\left(\frac{1}{\lambda}\right).\end{equation}                
  \end{itemize}
  \end{rhp}

From the estimates \eqref{eq:gdoubleEst1} and \eqref{eq:gdoubleEst2}, we see that  the jump $J_{R}(\lambda)=J_{B}(\lambda)$ tends to the identity matrix exponentially fast for $|\lambda|>\delta$. 
Additionally, from   \eqref{eq:f0}, \eqref{eq:phase3} and  \eqref{eq:PsiInfty}, we find that
 \begin{equation}\label{eq:R3Jump}
J_{R}(\lambda) =I+\frac{\sigma_1\Phi_1(h(\lambda;x,s))\sigma_1}{f_0 (\lambda)}+O(s^{-2/3}),
\end{equation}
where the error term is uniform for $|\lambda|=\delta$ and $L(s)\leq  \frac{\pi}{2} \frac{x}{\sqrt{s}}\leq M$  with any constant $M>1$ and $L(s)$ defined in \eqref{eq:L}. It follows from \eqref{eq:PIIAsy}, \eqref{eq:f0} and \eqref{eq:Psi1} that the second term in \eqref{eq:R3Jump} is of order $s^{-1/3}$.  Therefore, $R(\lambda)$ satisfies a small-norm RH problem for large $s$. Consequently,  from \eqref{eq:f0}, \eqref{eq:coeffc1}, \eqref{eq:tau0} and \eqref{eq:R3Jump}, we have  for $|\lambda|>\delta$ that
\begin{equation}\label{eq:R2}
R(\lambda)=I +\frac{\sigma_1\Phi_1(\tau(x,s))\sigma_1}{ (2\pi)^{-1/3}s^{1/3}\lambda}+O\left(\frac{1}{ s^{2/3}\lambda}\right),
\end{equation}
where  $\Phi_1$ and $\tau(x,s)$ are given in \eqref{eq:Psi1} and \eqref{eq:tau0}, respectively.  The error term is uniform for $|\lambda|>\delta$ and  $L(s)\leq  \frac{\pi}{2} \frac{x}{\sqrt{s}}\leq M$ with any constant $M>1$ and $L(s)$ given in \eqref{eq:L}.

\subsection{Case: $s^{\frac{2}{3}}\left(\frac{\pi}{2}\frac{x}{\sqrt{s}}-L(s)\right)\in [-M, 0]$ with any constant $M> 0$}
In this section, we consider the  large $s$ asymptotics of $Y(\lambda)$ for $s^{\frac{2}{3}}\left(\frac{\pi}{2}\frac{x}{\sqrt{s}}-L(s)\right)\in [-M, 0]$ with any constant $M> 0$ and $L(s)$ given in \eqref{eq:L}. The first few transformations $Y\to T\to A\to B$ are  the same  as those given in \eqref{eq:T}, \eqref{eq:A} and \eqref{eq:B}.  

From the estimates \eqref{eq:phiEst2}-\eqref{eq:phiEst3}, we see that  the jumps in \eqref{eq:Bjump1}  tend to the identity matrix exponentially fast for $\lambda$ bounded away from the origin. Therefore, the global parametrix can be defined as the identity matrix. It is seen from \eqref{eq:lambda0Est} that  the two soft edges $\lambda_0$ and $-\lambda_0$ merge together at the origin as $\frac{\pi}{2}\frac{x}{\sqrt{s}}\to 1$.   We intend to construct a local parametrix near the origin by using the $\psi$-functions associated with  the Hastings-McLeod solution of the second Painlev\'e equation in the following subsection.
\subsubsection{Local parametrix near the origin}
 Using \eqref{eq:phisign1},  we introduce
\begin{equation}\label{eq:hatf}
\hat{f}(\lambda)=\left(-i\frac{3}{4}\phi(\lambda)\right)^{\frac{2}{3}},\end{equation}
where the branch is chosen such that
\begin{equation}\label{eq:asyhatf}
\hat{f}(\lambda)\sim \hat{c}_1(x,s)\left(\lambda^2-\lambda_0^2\right), \quad \lambda\to\pm\lambda_0,
\end{equation}
with 
\begin{equation}\label{eq:hatc1}
\hat{c}_1=\hat{c}_1(x,s)=  \frac{s^{\frac{2}{3}}}{(2\pi)^{\frac{2}{3}}(1-\lambda_0^2)^{\frac{1}{3}}}+O(s^{-4/3}).
\end{equation}
We then define 
 the following conformal mapping near the origin
\begin{equation}\label{eq:rho}
\rho(\lambda)=(\hat{f}(\lambda)-\hat{f}(0))^{\frac{1}{2}}\sim  \hat{c}_1^{\frac{1}{2}}\lambda,
\end{equation}
where $\arg \left(\hat{f}(\lambda)-\hat{f}(0) \right)\in(-\pi,\pi)$.
It follows from \eqref{eq:lambda0Est} and \eqref{eq:hatf} that
\begin{equation}\label{eq:tau}
\hat{\tau}(x,s)= 2\hat{f}(0)= - \frac{2s^{\frac{2}{3}}}{(2\pi)^{\frac{2}{3}}(1-\lambda_0^2)^{\frac{1}{3}}}\lambda_0^2+O(s^{-2/3})= -4(2\pi)^{-\frac{2}{3}}s^{\frac{2}{3}}(1-\frac{\pi x}{2\sqrt{s}})+O(s^{-2/3}).
\end{equation}
For later use, we derive the estimate  by the Taylor expansion
\begin{align}\nonumber
i\frac{4}{3}\rho(\lambda)^3+i\hat{\tau}(x,s)\rho(\lambda)- \phi(\lambda)&=\frac{4i\hat{f}(0)^2}{3\rho(\lambda)}\left(\left(\frac{\hat{f}(\lambda)}{\hat{f}(0)}-1\right)^2+\frac{3}{2}\left(\frac{\hat{f}(\lambda)}{\hat{f}(0)}-1\right)\right.\\  \nonumber
&~~~~~~~~~~~~~~~~~~~~~~~~~~~~~ \left.-\frac{\hat{f}(\lambda)^2}{\hat{f}(0)^2}\left(1-\frac{\hat{f}(0)}{\hat{f}(\lambda)}\right)^{\frac{1}{2}}\right)\nonumber\\
&= -\frac{i}{2}\hat{f}(0)^2/\rho(\lambda)\left(1+O(\lambda_0^2)\right), \label{eq:phase1} 
\end{align}
uniformly for $|\lambda|=\delta$ with some constant $\delta>0$.

\begin{figure}[h]
\begin{center}
   \setlength{\unitlength}{1truemm}
   \begin{picture}(100,68)(-5,8)
       
       \put(25,40){\line(1,0){30}}

       \put(25,40){\line(-1,-1){25}}
       \put(25,40){\line(-1,1){25}}
       \put(55,40){\line(1,1){25}}
       \put(55,40){\line(1,-1){25}}

       \put(10,55){\thicklines\vector(1,-1){1}}
       \put(10,25){\thicklines\vector(1,1){1}}
       \put(70,25){\thicklines\vector(1,-1){1}}
       \put(70,55){\thicklines\vector(1,1){1}}
        \put(40,40){\thicklines\vector(1,0){1}}

       \put(-2,11){$\widehat{\Sigma}_3^{(\PII)}$}
       \put(-2,67){$\widehat{\Sigma}_2^{(\PII)}$}
       \put(80,11){$\widehat{\Sigma}_4^{(\PII)}$}
       \put(80,67){$\widehat{\Sigma}_1^{(\PII)}$}
      
       \put(25,40){\thicklines\circle*{1}}
       \put(55,40){\thicklines\circle*{1}}

       \put(25,36.3){$-b_0$}
       \put(54,36.3){$b_0$}
\end{picture}
   \caption{Jump contours for the RH problem for $\widehat{\Phi}^{(\PII)}$.}
   \label{fig:hatPII}
\end{center}
\end{figure}

The parametrix can be constructed for $\lambda\in U(0,\delta)$
\begin{equation}\label{eq:P0}
\widehat{ P}^{(0)}(\lambda)=\sigma_1 \widehat{\Phi}^{(\PII)}(\rho(\lambda), \hat{\tau}(x,s); \rho(\lambda_0) )\sigma_1e^{-\phi(\lambda)\sigma_3},
 \end{equation}
 where $\rho$ and $\hat{\tau}$ are defined in \eqref{eq:rho} and \eqref{eq:tau}, respectively.
Here $\widehat{\Phi}^{(\PII)}(\lambda,t; b_0)$ satisfies the following RH problem.

  \begin{rhp} The function $\widehat{\Phi}^{(\PII)}(\lambda)=\widehat{\Phi}^{(\PII)}(\lambda,t; b_0)$ satisfies the properties below. 
\begin{itemize}
\item[\rm (1)]  $\widehat{\Phi}^{(\PII)}(\lambda)$ is analytic in $\mathbb{C}\setminus\{\cup_{k=1}^4\widehat{\Sigma}_{k}^{(\PII)}\cup (-b_0,b_0)\}$; see Fig. \ref{fig:hatPII}.

 \item[\rm (2)] $\widehat{\Phi}^{(\PII)}(\lambda)$ satisfies the jump condition
 \begin{equation}\label{eq:hatPhiJumps}
 \widehat{\Phi}^{(\PII)}_{+}(\lambda)= \widehat{\Phi}^{(\PII)}_{-}(\lambda)\left\{
                               \begin{array}{cc}
                                 \left(
                               \begin{array}{cc}
                                 1 &0\\
                              1& 1 \\
                                 \end{array}
                             \right)  &  \lambda\in\widehat{\Sigma}_1^{(\PII)}\cup\widehat{\Sigma}_2^{(\PII)},\\
                              \left(
                               \begin{array}{cc}
                                 1 &-1\\
                             0 & 1 \\
                                 \end{array}
                             \right)  &  \lambda\in\widehat{\Sigma}_3^{(\PII)}\cup\widehat{\Sigma}_4^{(\PII)},\\
                                   \left(
                               \begin{array}{cc}
                                 0 &-1\\
                              1& 1 \\
                                 \end{array}
                             \right)  &  \lambda\in(-b_0,b_0).
                            
                                 \end{array}
                             \right.
\end{equation}

  \item[\rm (3)]   The large $\lambda$ asymptotic of $\widehat{\Phi}^{(\PII)}(\lambda)$ is given by \eqref{eq:PsiInfty}. 
  \end{itemize}
  \end{rhp}
  
The above RH problem is equivalent to the one for $ \Phi^{(\PII)}(\lambda)$ after some elementary transformation by deforming the original jump contours in \eqref{eq:PhiJumps} to those  in \eqref{eq:hatPhiJumps}; see Figs. \ref{fig:PII} and 
\ref{fig:hatPII}.  Since $\hat{\tau}(x,s)$ is real and the Hastings-McLeod solution is pole free on the real axis, 
the parametrix $ \widehat{P}^{(0)}(\lambda)$ in \eqref{eq:P0} exists for $|\lambda|< \delta$.  Moreover, it follows from \eqref{eq:hatf}, \eqref{eq:P0} and \eqref{eq:hatPhiJumps} that $ \widehat{P}^{(0)}(\lambda)$ satisfies the same jump condition as $B(\lambda)$ in \eqref{eq:Bjump1}  for $\lambda\in\Sigma_{B} \cap U(0,\delta)$.

 \subsubsection{Small-norm  RH problem}
 We define 
\begin{equation}\label{eq:Rdouble1}
R(\lambda) =\left\{
                               \begin{array}{cc}
                              B(\lambda)\widehat{P}^{(0)}(\lambda)^{-1},& |\lambda|<\delta,\\
                           B(\lambda),& |\lambda|>\delta,\\
                                 \end{array}
                             \right.
\end{equation}
where  $\widehat{P}^{(0)}(\lambda)$ is defined in  \eqref{eq:P0}.  
Then, $R(\lambda)$ satisfies RH problem \ref{RH:double}. It follows from the estimates \eqref{eq:phiEst2}-\eqref{eq:phiEst3} that  the jumps in \eqref{eq:Bjump1} tend to the identity matrix exponentially fast for $|\lambda|>\delta$.
Using \eqref{eq:phase1}, \eqref{eq:P0}, \eqref{eq:PsiInfty} and \eqref{eq:Rdouble1}, we obtain the asymptotic expansion as   $s\to+\infty$
\begin{equation}\label{eq:R2Jump}
J_{R}(\lambda) =I+J_1(\lambda)+O(s^{-2/3}),
\end{equation}
 where 
\begin{equation}\label{eq:J1}
J_{1}(\lambda) =\frac{\sigma_1\Phi_1(\hat{\tau})\sigma_1-\frac{i}{2}\hat{f}(0)^2\sigma_3}{\rho (\lambda)}, \quad |\lambda|=\delta.
\end{equation}
Here, the error term is uniform for   $|\lambda|=\delta $ and $s^{\frac{2}{3}}\left(\frac{\pi}{2}\frac{x}{\sqrt{s}}-L(s)\right)\in [-M, 0]$ with any constant $M> 0$. It follows from \eqref{eq:rho} that $J_1(\lambda)=O(s^{-1/3})$. Therefore, $R(\lambda)$ satisfies a small-norm RH problem for large $s$.
According to the theory of small-norm RH problems \cite{DZ, Deift}, we obtain from \eqref{eq:rho},  \eqref{eq:R2Jump} and \eqref{eq:J1} that 
\begin{equation}\label{eq:RExpand}
R(\lambda) =I+\frac{R_1}{\lambda}+O\left(\frac{1}{ s^{2/3}\lambda}\right),\quad |\lambda|>\delta,
\end{equation}
where
\begin{equation}\label{eq:R1}
R_1 =\left(\frac{i}{2}(H(\hat{\tau})-\hat{f}(0)^2)\sigma_3-\frac{1}{2}u(\hat{\tau})\sigma_2\right) \hat{c}_1^{-1/2},
\end{equation}
with $ \hat{c}_1$ and $\hat{\tau}$ defined in \eqref{eq:hatc1} and  \eqref{eq:tau}.
Therefore, we have 
\begin{equation}\label{eq:R11}
(R_1)_{11} =\frac{i}{2}(2\pi)^{1/3} (H(\hat{\tau})-\frac{\hat{\tau}^2}{4})s^{-\frac{1}{3}}+O(s^{-\frac{2}{3}}),
\end{equation}
where $\hat{\tau}$ is  given in \eqref{eq:tau} and the error term is uniform for    $s^{\frac{2}{3}}\left(\frac{\pi}{2}\frac{x}{\sqrt{s}}-L(s)\right)\in [-M, 0]$ with any constant $M> 0$ and $L(s)$ defined in \eqref{eq:L}.

 \section{Large $s$ asymptotics of $Y$ for $x\sqrt{s}$  bounded} \label{sec:PVasy}
 
In this section, we consider the asymptotics  of $Y(\lambda)$ as $s\to +\infty$ and $x\to 0$ in such a way that $t=\sqrt{s}x$ remains  bounded.

The first transformation $Y\to T$ is defined the same as \eqref{eq:T}. As $ s\to+\infty$,  we have $ \sigma(\sqrt{s}\lambda;s)\to\chi_{(-1,1)}(\lambda)$. 
Replacing $\sigma(\sqrt{s}\lambda;s)$ by $\chi_{(-1,1)}(\lambda)$, we arrive at the following  approximate RH problem.

\subsection{Global parametrix}

  \begin{rhp} \label{RH:PV}The function $\Phi^{(\PV)}(\lambda)=\Phi^{(\PV)}(\lambda,t)$ satisfies the properties below. 
\begin{itemize}
\item[\rm (1)]  $\Phi^{(\PV)}(\lambda)$ is analytic in $\mathbb{C}\setminus [-1,1]$.

 \item[\rm (2)] $\Phi^{(\PV)}(\lambda)$ satisfies the jump condition
  \begin{equation}\label{eq:PsiVJumpAsy}\Phi^{(\PV)}_{+}(\lambda)=\Phi^{(\PV)}_{-}(\lambda)                          
  \left(
                               \begin{array}{cc}
                                 0 & 1\\
                                 -1 & 2 \\
                                 \end{array}
                             \right),
                             \end{equation}   
                             for $\lambda\in(-1,1)$.
  \item[\rm (3)]   As $\lambda\to\infty$, we have
  \begin{equation}\label{eq:PsiVInfinity}\Phi^{(\PV)}(\lambda)=\left(I+\frac{\Phi_1^{(\PV)}(t)}{\lambda}+O\left(\frac{1}{\lambda^2}\right)\right)e^{it\lambda\sigma_3}.\end{equation}
    \item[\rm (4)]   As $\lambda\to\pm1$, we have
  
  \begin{equation}\label{eq:Psiv1}\Phi^{(\PV)}(\lambda)=E^{(\pm 1)}(\lambda)         
    \left(
                               \begin{array}{cc}
                                 1 & \frac{1}{2\pi i}\log\left(\frac{\lambda- 1}{\lambda+1}\right)\\
                                 0 & 1 \\
                                 \end{array}
                             \right)              
  \left(
                               \begin{array}{cc}
                                 1 & 0\\
                                 -1 & 1 \\
                                 \end{array}
                             \right),
                             \end{equation}   
where $E^{(\pm1)}(\lambda)$ is analytic near $\lambda=\pm 1 $ and the logarithm takes the principal branch.
  \end{itemize}
  \end{rhp}
  
  \begin{rem}\label{rem:PVPro} 
  The   model RH problem  above corresponds to  the classical sine kernel determinant; see  RH problem \ref{RHY}, where $\sigma(\lambda;s)$ is replaced by $\chi_{(-1,1)}(\lambda)$; also refer to \cite[Chapter 9]{A}. 
  The solution to this  RH problem can be constructed  using the $\psi$-functions  associated with the fifth Painlev\'e equation. Let $v(t)=-2it(\Phi_1^{(\PV)})_{11}(t)$ with $\Phi_1^{(\PV)}(t)$ given in \eqref{eq:PsiVInfinity}.  Then, the function $v(t)$ is a smooth solution of  the $\sigma$-form of the fifth Painlev\'e equation \eqref{eq:sigmaPV} for $t\in(0,+\infty) $,  satisfying the boundary conditions given in \eqref{eq:PVasy}; see \cite[Chapter 9]{A}. 
\end{rem}

 \subsection{Small-norm  RH problem}
We define 
\begin{equation}\label{eq:R4}
R(\lambda)=T(\lambda)\left\{  
\begin{array}{c}
\Phi^{(\PV)}(\lambda,\sqrt{s}x)^{-1}, \quad  |\lambda\pm1|>\delta, \\
\left(E^{(\mp 1)}(\lambda)         
    \left(
                               \begin{array}{cc}
                                 1 & \frac{1}{2\pi i}\int_{\mathbb{R}} \frac{\sigma(\sqrt{s}\zeta;s)}{\zeta-\lambda}d\zeta\\
                                 0 & 1 \\
                                 \end{array}
                             \right)              
  \left(
                               \begin{array}{cc}
                                 1 & 0\\
                                 -1 & 1 \\
                                 \end{array}
                             \right)\right)^{-1}, \quad  |\lambda\pm1|<\delta,
\end{array}
\right. 
\end{equation}
where $T(\lambda)$ is defined in  \eqref{eq:T}, $\Phi^{(\PV)}(\lambda,t)$ is the solution to RH problem \ref{RH:PV} and $E^{(\pm 1)}$ is given in \eqref{eq:Psiv1}.

From \eqref{eq:TJump1}, \eqref{eq:TInfinity} and \eqref{eq:PsiVJumpAsy}-\eqref{eq:R4}, it follows that  $R(\lambda)$ satisfies the following RH problem. 
  \begin{rhp}\label{RH:doublePV}
The function  $R(\lambda)$  satisfies the following properties.
\begin{itemize}
\item[\rm (1)]  $R(\lambda)$ is analytic in $\mathbb{C}\setminus\Sigma_{R}$, with $\Sigma_{R}=\partial U(-1, \delta)\cup \partial U(1, \delta)
 \cup(-\infty, -1-\delta)\cup(-1+\delta, 1-\delta)\cup(1+\delta,+\infty)$.

 \item[\rm (2)]  $R(\lambda)$  satisfies the jump condition
 \begin{equation}
R_{+}(\lambda)=R_{-}(\lambda)J_{R}(\lambda),\end{equation}
where
\begin{equation}\label{eq:JR4}
J _{R}(\lambda)= \left\{        \begin{array}{cc}

                          E^{(\mp 1)}(\lambda)         
    \left(
                               \begin{array}{cc}
                                 1 & -\frac{1}{2\pi i}\int_{\mathbb{R}} \frac{\widehat{\sigma}(\lambda;s)}{\zeta-\lambda}d\zeta                                 \\
                                 0 & 1 \\
                                 \end{array}
                             \right)              
  E^{(\mp 1)}(\lambda)^{-1},&   |\lambda\pm1|=\delta,\\

                          \Phi^{(\PV)}_{-}(\lambda,\sqrt{s}x) \left(
                               \begin{array}{cc}
                                 1-\widehat{\sigma}(\lambda;s) & \widehat{\sigma}(\lambda;s)\\
                                 -\widehat{\sigma}(\lambda;s) & 1+\widehat{\sigma}(\lambda;s) \\
                                 \end{array}
                             \right)\Phi^{(\PV)}_{-}(\lambda,\sqrt{s}x)^{-1}, &  \lambda\in \mathbb{R}, ~ |\lambda\pm1|>\delta, 
                                 \end{array}
                             \right.
\end{equation}
and 
 \begin{equation}\label{eq:hatsigma} 
  \widehat{\sigma}(\lambda;s)=\sigma(\sqrt{s}\lambda;s)-\chi_{(-1,1)}(\lambda).\end{equation} 
 \item[\rm (3)]   As $\lambda\to\infty$, we have
       \begin{equation}\label{eq:RInfinity3} R(\lambda)=I+O\left(\frac{1}{\lambda}\right).\end{equation}                
  \end{itemize}
  \end{rhp}

From \eqref{eq:sigma} and \eqref{eq:hatsigma}, we have 
 \begin{equation}\label{eq:hatsigma1} 
  \widehat{\sigma}(\lambda;s)=\frac{1}{1+\exp(s(\lambda^2-1))}, \quad \lambda\in(-\infty,-1)\cup(1,+\infty),\end{equation} 
  and 
   \begin{equation}\label{eq:hatsigma2} 
  \widehat{\sigma}(\lambda;s)=-\frac{1}{1+\exp(s(1-\lambda^2))}, \quad \lambda\in(-1,1).\end{equation} 
  Therefore, we have
   \begin{equation}\label{eq:hatsigma3} 
 | \widehat{\sigma}(\lambda;s)|\leq\exp(-s|\lambda^2-1|), \quad \lambda\in\mathbb{R}, \end{equation} 
and 
      \begin{equation}\label{eq:RJump}
   \frac{1}{2\pi i} \int_{\mathbb{R}} \frac{\widehat{\sigma}(\zeta;s)}{\zeta-\lambda}d\zeta=O\left(1/s^2\right),\end{equation}
uniformly for $|\lambda\pm1|=\delta$. 
From \eqref{eq:PsiVInfinity} and \eqref{eq:Psiv1},  we see that $ \Phi^{(\PV)}_{-}(\lambda,\sqrt{s}x)$ and its inverse are bounded uniformly for $\lambda\in\mathbb{R}$, $|\lambda\pm1|>\delta$ and  $\sqrt{s}x\in[0,M]$ for any constant $M>0$.
This, together with \eqref{eq:hatsigma3}, implies that there exists some positive constant $c>0$ such that 
 \begin{equation}\label{eq:JRPVest1} 
J_R(\lambda) =I+O\left(\exp(-cs|\lambda^2-1|)\right), \quad \lambda\in\mathbb{R}, \quad |\lambda\pm1|>\delta. \end{equation} 
From  \eqref{eq:RJump} and the fact that $ E^{(\pm 1)}(\lambda)$ is analytic in $\lambda$, we have
 \begin{equation}\label{eq:JRPv}
J_R(\lambda) =I+O(1/s^2),
\end{equation}
uniformly for $|\lambda\pm1|=\delta$.
Thus, $R(\lambda)$ satisfies a small-norm RH problem  for large $s$. According to the theory of small-norm RH problems \cite{DZ, Deift}, we have as $s\to+\infty$
 \begin{equation}\label{eq:RPv}
R(\lambda) =I+O\left(\frac{1}{ s^{2} \lambda}\right),
\end{equation}
where the error term is uniform for $\lambda\in \mathbb{C}\setminus\Sigma_{R}$ and $\sqrt{s}x\in[0,M]$ for any constant $M>0$.

\section{Proof of Theorems   } \label{sec:ProofTheorems}

\subsection{Proof of Theorem \ref{thm: large gap asy x}  } \label{sec:ProofTheorem1}
Tracing back the transformations  \eqref{eq:S} and \eqref{eq:1R}, we have
\begin{equation}\label{eq:Yasy2}
Y_{+}(x\lambda) =R(\lambda)N_{+}(\lambda)  \left(
                               \begin{array}{cc}
                                 1 &e^{2ix^2\lambda}\frac{\sigma(x\lambda;s )}{1- \sigma(x\lambda;s )}\\
                                 0 & 1 \\
                                 \end{array}
                             \right), 
\end{equation}
for  $\lambda\in\mathbb{R}$. 
From \eqref{def:fh}, we have
 for  $\lambda\in\mathbb{R}$
  \begin{equation}\label{eq:YSasy1}
Y_+(\lambda)\mathbf{f}(\lambda) =\sqrt{\sigma(\lambda;s)}R(\lambda/x)N_+(\lambda/x)  \left(
                               \begin{array}{c}
                                  e^{ix\lambda}\frac{1}{1- \sigma(\lambda;s)}\\
                                e^{-ix\lambda} \\
                                 \end{array}
                             \right) ,
\end{equation}
and 
\begin{equation}\label{eq:YSasy2}
\mathbf{h}(\lambda)^{t}Y_+(\lambda)^{-1} =\frac{\sqrt{\sigma(\lambda;s)}}{2\pi i} 
\left(e^{-ix\lambda} , -e^{ix\lambda}\frac{1}{1- \sigma(\lambda;s)}\right)
 N_+(\lambda/x)^{-1}R(\lambda/x)^{-1},
\end{equation}
where $N(\lambda)$ is defined in \eqref{eq:N} and  $R(\lambda)$ satisfies the estimates in \eqref{eq:RExpand1} and \eqref{eq:RExpand12}.
Substituting \eqref{eq:YSasy1} and \eqref{eq:YSasy2} into \eqref{eq:BC2} and using \eqref{eq:RExpand1}, \eqref{eq:dRExpand1},  \eqref{eq:RExpand12} and \eqref{eq:dRExpand12}, we obtain as $x\to+\infty$
\begin{equation}\label{eq:Dasy2}
\partial_s \log D(x,s)=\frac{i}{2\pi } \int_{\mathbb{R}}\left(\frac{2ix}{1-\sigma(\lambda;s)}+\frac{d}{d\lambda}\frac{1}{1-\sigma(\lambda;s)}+\frac{2\frac{d}{d\lambda}\log d_{+}(\lambda/x)}{1-\sigma(\lambda;s)}\right) \partial_s \sigma(\lambda;s)d\lambda
+\mathcal{E}(x,s),
\end{equation}
where  $d(\lambda)$ is given in \eqref{eq:d}. The error term $\mathcal{E}(x,s)=O\left(e^{s-x^2}\right)$ is uniform for $s\in(-\infty, c_1x^2 ]$  with any constant $0<c_1<1$,  and $\mathcal{E}(x,s)=O\left(e^{-cx^2}\right)$ is uniform for $s\in[c_1x^2, c_2x^2 ]$ with any constant $0<c_1<c_2<\frac{\pi^2}{4}$. Here, $c$ is  some positive constant.

We proceed to evaluate the integral in \eqref{eq:Dasy2}. Since $\sigma(\lambda;s)$ is an even  function in $\lambda$, we have 
\begin{equation}\label{eq:Int0}
\int_{\mathbb{R}}\partial_s \sigma(\lambda;s)\frac{d}{d\lambda}\frac{1}{1-\sigma(\lambda;s)}d\lambda=0.\end{equation}
From \eqref{eq:sigma}, we have 
\begin{equation}\label{eq:dsigma1}
\partial_s\sigma(\lambda;s)=\sigma(\lambda;s)^2\exp(\lambda^2-s),\end{equation}
\begin{equation}\label{eq:dsigma2}
\frac{d}{d\lambda}\sigma(\lambda;s)=-2\lambda\sigma(\lambda;s)^2\exp(\lambda^2-s),\end{equation}
and
\begin{equation}\label{eq:sigma3}
1-\sigma(\lambda;s)=\sigma(\lambda;s)\exp(\lambda^2-s).\end{equation}
Therefore, we have
\begin{equation}\label{eq:Int2}
\int_{\mathbb{R}}\frac{\partial_s \sigma(\lambda;s)}{1-\sigma(\lambda;s)}d\lambda=-\frac{d}{ds} \int_{\mathbb{R}} \log(1-\sigma(\lambda;s))  d\lambda
=\int_{\mathbb{R}}\sigma(\lambda;s)d\lambda.\end{equation}
From \eqref{eq:logh}, we have
\begin{equation}\label{eq:dh}
\frac{d}{d\lambda}\log d_+(\lambda/x)
=\frac{1}{\pi i}\int_{\mathbb{R}}\zeta\sigma(\zeta;s)\frac{d\zeta}{\zeta-\lambda_+}, \end{equation}
for $\lambda\in\mathbb{R}$.
From \eqref{eq:dsigma1}, \eqref{eq:sigma3} and 
\eqref{eq:dh}, we have
\begin{equation}\label{eq:Int3}
-\frac{1}{\pi i}\int_{\mathbb{R}}\frac{\partial_s \sigma(\lambda;s)}{1-\sigma(\lambda;s)}\frac{d}{d\lambda}\ln d(\lambda/x)d\lambda=\frac{1}{\pi ^2}\int_{\mathbb{R}}\left(\int_{\mathbb{R}}\sigma(\zeta;s)\frac{\zeta}{\zeta-\lambda_{+}} d\zeta\right)\sigma(\lambda;s)d\lambda.\end{equation}
Next, we calculate the double integral
\begin{align}\int_{\mathbb{R}}\left(\int_{\mathbb{R}}\sigma(\zeta;s)\frac{\zeta}{\zeta-\lambda_+} d\zeta\right)\sigma(\lambda;s)d\lambda=&\lim_{\epsilon\to 0^+}\int_{\mathbb{R}^2}\frac{\zeta \sigma(\zeta;s)\sigma(\lambda;s)}{\zeta-\lambda-i\epsilon} d\lambda d\zeta \nonumber\\
&=\frac{1}{2}\lim_{\epsilon\to 0^+}\int_{\mathbb{R}^2}\frac{\zeta \sigma(\zeta;s)\sigma(\lambda;s)}{\zeta-\lambda-i\epsilon} +\frac{\lambda \sigma(\lambda;s)\sigma(\zeta;s)}{\lambda-\zeta-i\epsilon} d\lambda d\zeta\nonumber\\
&=\frac{1}{2}\left(\int_{\mathbb{R}}\sigma(\lambda;s)d\lambda\right)^2. \label{eq:doubleInt}
\end{align}
In the second line of the above equation, we use the change of variables $(\lambda,\zeta)\to(\zeta,\lambda)$ in the double integral.
From \eqref{eq:Int3} and \eqref{eq:doubleInt}, we have
\begin{equation}\label{eq:Int4}
-\frac{1}{\pi i}\int_{\mathbb{R}}\frac{\partial_s \sigma(\lambda;s)}{1-\sigma(\lambda;s)}\frac{d}{d\lambda}\ln d(\lambda/x)d\lambda=\frac{1}{2\pi^2}\left(\int_{\mathbb{R}}\sigma(\lambda;s)d\lambda\right)^2.\end{equation}

 Substituting  \eqref{eq:Int0},  \eqref{eq:Int2} and \eqref{eq:Int4} into \eqref{eq:Dasy2}, we obtain as $x\to+\infty$
\begin{equation}\label{eq:Dasy3}
\partial_s \log D(x,s)= \frac{x}{\pi}\frac{d}{ds}\int_{\mathbb{R}}\ln(1-\sigma(\lambda;s))d\lambda +\frac{1}{2\pi^2}\left(\int_{\mathbb{R}}\sigma(\lambda;s)d\lambda\right)^2+\mathcal{E}(x,s).\end{equation}
Here, the error term $\mathcal{E}(x,s)=O\left(e^{s-x^2}\right)$ is uniform for $s\in(-\infty, c_1x^2 ]$  with any constant $0<c_1<1$ and $\mathcal{E}(x,s)=O\left(e^{-cx^2}\right)$ is uniform for $s\in[c_1x^2, c_2x^2 ]$ with $c_1<c_2<\frac{\pi^2}{4}$  and some constant $c>0$.
Therefore, we have as $x\to+\infty$
\begin{equation}\label{eq:Interror}
 \int_{-\infty}^s\mathcal{E}(x,\tau)d\tau= \int_{-\infty}^{c_1x^2}\mathcal{E}(x,\tau)d\tau+ \int_{c_1x^2}^{s}\mathcal{E}(x,\tau)d\tau=O\left(x^2e^{-cx^2}\right),
 \end{equation}
 uniformly for $s\in(-\infty, c_2x^2]$.
 Using integration by parts, we have
\begin{equation}\label{eq:Intbypart}
\int_{\mathbb{R}}\log (1-\sigma(\lambda;s))d\lambda=-2\int_{\mathbb{R}}\sigma(\lambda;s)\lambda^2d\lambda. \end{equation}
It follows from \eqref{eq:sigma} that $\sigma(\lambda;s)=O(e^{-(\lambda^2+|s|)})$ as $s\to -\infty$. Therefore, we have for $x\geq 0$
\begin{equation}\label{eq:Dinitial}
 \lim_{s\to-\infty}\log D(x,s)=0.
 \end{equation}
Integrating both sides of \eqref{eq:Dasy3} with respect to $s$ over the interval $(-\infty, s]$ and using \eqref{eq:Interror} and \eqref{eq:Intbypart}, we obtain  \eqref{thm:AsyD}.  This completes the proof of Theorem \ref{thm: large gap asy x}. 

\subsection{Proof of Theorem \ref{thm: DAsylargs} } \label{sec:ProofTheorem2}

To prove Theorem \ref{sec:ProofTheorem2}, we adopt a strategy similar to \cite{CC}.  In the first step, we derive the asymptotic expansions of  both the logarithmic derivatives of $D(x,s)$ with respect to  $x$ and $s$ by using  the differential identities \eqref{eq:DY} and \eqref{eq:BC2},  along with the results of the RH analysis for $Y$ performed in Section \ref{sec:asyYlarges}.
In the second step, we integrate the asymptotic approximations of the differential identities  along appropriate contours in the $(x,s)$-plane. 

\subsubsection{Asymptotics of the  logarithmic derivative of $D(x,s)$ with respect to $x$}
In this subsection, we derive the asymptotic expansion of $\partial_x \log  D(x,s)$ as $s\to+\infty$ and $x\to+\infty$. The main result is stated  in the proposition below. 
\begin{pro} \label{Pro:LDx}As $s\to+\infty$ and $x\to+\infty$, we have
\begin{equation}\label{eq:dxD}
\partial_x \log D (x,s)=-xs+\frac{\pi^2}{12}x^3+O\left(\frac{1}{\sqrt{s}\lambda_0^2}\right),
\end{equation} 
where $\lambda_0$ is given in  \eqref{eq:lambda0Est} and the error term is uniform for  
 $ \frac{x}{\sqrt{s} }\in [M_2, \frac{2 }{\pi }-M_1s^{-\frac{2}{3}}]$ with  a sufficiently large constant $M_1>0$ and any constant $0<M_2<\frac{2 }{\pi }$.
\end{pro}
\begin{proof}
Tracing back the series of invertible transformations $Y\to T \to A\to B\to  C\to R$ given in \eqref{eq:T}, \eqref{eq:A},  \eqref{eq:B},  \eqref{eq:C} and \eqref{eq:Rlarges}, we have 
\begin{align}\label{eq:Yasylarges}
Y(\sqrt{s}\lambda_0\lambda) =&(1-\sigma(\sqrt{s}\lambda_0;s))^{-\frac{1}{2}\sigma_3}R(\lambda)
P^{(\infty)}(\lambda) e^{(g(\lambda_0\lambda)-ix\sqrt{s}\lambda_0\lambda)\sigma_3}\\   \nonumber
&~~\times (1-\sigma(\sqrt{s}\lambda_0;s))^{\frac{1}{2}\sigma_3}
 \left(
                               \begin{array}{cc}
                                 1 &  - e^{2ix\sqrt{s}\lambda_0\lambda}\\                              
                           0   &1 \\ \end{array}
                             \right),\end{align}  
where $\lambda\in\widehat{\Omega}_2$ and $|\lambda\pm1|>\delta$, $g(\lambda)$ and $P^{(\infty)}(\lambda)$ are given in \eqref{eq:g} and \eqref{eq:Gp}, respectively. The region $\widehat{\Omega}_2$ is illustracted in Fig. \ref{fig:C}.
The function $e^{2ix\sqrt{s}\lambda_0\lambda}$ is exponentially small as $\lambda\to \infty$ in the region $\widehat{\Omega}_2$.
Substituting \eqref{eq:dgExpand}, \eqref{eq:dgExpandg1}, \eqref{eq:GpExpand}, \eqref{eq:Rasylarges} into \eqref{eq:Yasylarges}, we obtain
the asymptotic expansion as $\lambda\to\infty$
\begin{align}\label{eq:YasyLargelambda}
Y(\sqrt{s}\lambda_0\lambda) =&
(1-\sigma(\sqrt{s}\lambda_0;s))^{-\frac{1}{2}\sigma_3}  \left( I+\frac{1}{\lambda}\left(-\frac{1}{2}\sigma_2+\frac{g_1}{\lambda_0}\sigma_3+O\left(\frac{1}{s\lambda_0^3 }\right)\right)+O\left(\frac{1}{\lambda^2}\right)\right) \\ \nonumber 
& ~~\times(1-\sigma(\sqrt{s}\lambda_0;s))^{\frac{1}{2}\sigma_3},\end{align} 
where $\sigma_k$, $\lambda_0$, $g_1$ are given in \eqref{eq:PauliMatrices}, \eqref{eq:lambda0Est} and \eqref{eq:dgExpandg1}, respectively.
From the differential identity \eqref{eq:DY} and \eqref{eq:YasyLargelambda}, we obtain
\begin{equation}\label{eq:aAsy}
\partial_x \log D (x,s)=-2i\sqrt{s}\left(g_1+O\left(\frac{1}{s\lambda_0^2}\right)\right),
\end{equation} 
where the error term is uniform for  
 $ \frac{x}{\sqrt{s} }\in [M_2, \frac{2 }{\pi }-M_1s^{-\frac{2}{3}}]$ with  a sufficiently large constant $M_1>0$ and any constant $0<M_2<\frac{2 }{\pi }$. 
Substituting   \eqref{eq:dgExpandg1} into \eqref{eq:aAsy} and using \eqref{eq:lambda0Est}, we obtain \eqref{eq:dxD}.   
\end{proof}
\subsubsection{Asymptotics of the  logarithmic derivative of $D(x,s)$ with respect to $s$}
This subsection is devoted to proving the following result on the asymptotic expansion of $\partial_s \log  D(x,s)$ as $s\to+\infty$ and $x\to+\infty$. 
\begin{pro}\label{Pro:LDs} As $s\to+ \infty$ and $x\to+\infty$, we have
\begin{equation}\label{eq:dsD}
\partial_s \log D (x,s)=-\frac{x^2}{2}+O(1),
\end{equation} 
where the error term is uniform for $\frac{x}{\sqrt{s}}\in[M_2, M_3]$ with any constant $0<M_2<M_3<\frac{2}{\pi}$.  
\end{pro}
\begin{proof}

We derive the asymptotics of $\partial_s \log D (x,s)$  using the differential identity \eqref{eq:BC2}. 
From  \eqref{def:fh} and \eqref{eq:A},  we have 
\begin{equation}\label{eq:trB}
\left(\mathbf{ h}^{t}Y^{-1}(Y\mathbf{f})'\right)(\sqrt{s}\lambda)=\frac{\sigma(\sqrt{s}\lambda;s)}{2\pi i \sqrt{ s}}(A^{-1}A')_{12}(\lambda),
\end{equation} 
where $\lambda\in\Omega_2$ and $'$ denotes the derivative with respect to $\lambda$. From \eqref{eq:B} and \eqref{eq:C}, we can rewrite \eqref{eq:trB} as 
\begin{equation}\label{eq:trB1}
\left(\mathbf{ h}^{t}Y_{+}^{-1}(Y_{+}\mathbf{f})'\right)(\sqrt{s}\lambda_0\lambda)=\frac{\sigma(\sqrt{s}\lambda_0\lambda;s)}{2\pi i \sqrt{ s} \lambda_0(1-\sigma(\sqrt{s}\lambda_0;s))}e^{-2g(\lambda_0\lambda)}(C^{-1}C')_{12}(\lambda),
\end{equation} 
for $\lambda\in\widehat{\Omega}_2$. 
Using \eqref{eq:BC2}, \eqref{eq:B}, \eqref{eq:dsigma1}, \eqref{eq:sigma3} and \eqref{eq:trB1}, we obtain 
\begin{equation}\label{eq:intB}
\partial_s \log D(x,s)
=-\frac{1}{2\pi i}\int_{\mathbb{R}}(\widehat{C}^{-1}\widehat{C}')_{12}(\lambda)e^{-2\phi_+(\lambda_0\lambda)}\sigma(\sqrt{s}\lambda_0\lambda;s)d\lambda.
\end{equation}
Here, $\widehat{C}(\lambda)$  represents the analytic extension of $C(\lambda)$ from $\widehat{\Omega}_2$ to the upper complex plane, which implies
\begin{equation}\label{eq:hatB}\hat{C}(\lambda)
=C_+(\lambda) \left\{
                               \begin{array}{cc}
                              I,& \lambda\in(-1,1),\\
                            \left(
                               \begin{array}{cc}
                                 1 & e^{2\phi_{+}(\lambda_0\lambda)}\\ 
                           0   &1 \\ \end{array}
                             \right),& \lambda\in(-\infty,-1)\cup(1, +\infty).
                                 \end{array}
                             \right.
\end{equation}

We proceed to  evaluate the integral in \eqref{eq:intB}. For this purpose, we split the integral  into three parts
\begin{equation}\label{eq:intBs}
\partial_s \log D(x,s)=I_1+I_2+I_3,
\end{equation}
where
\begin{equation}\label{eq:intIj}
I_j
=-\frac{1}{2\pi i}\int_{D_j}(\widehat{C}^{-1}\widehat{C}')_{12}(\lambda)e^{-2\phi_+(\lambda_0\lambda)}\sigma(\sqrt{s}\lambda_0\lambda;s)d\lambda,
\end{equation}
$j=1,2,3$ with 
\begin{equation}\label{eq:Dj}D_1=(-1+\delta,1-\delta),  
D_2=(-\infty,-1-\delta)\cup(1+\delta,+\infty),  D_3=(-1-\delta,-1+\delta)\cup(1-\delta,1+\delta).\end{equation}

 We  evaluate $I_1$ first. Under the condition of this proposition, it follows from \eqref{eq:lambda0Est} that $\lambda_0$ is bounded away from zero.
 Therefore, from  \eqref{eq:phiEst1}, we have for $\lambda\in D_1$
\begin{equation}\label{eq:phiEst}\phi_+(\lambda_0\lambda)\geq c s\lambda_0^3(1-\lambda^2)^{\frac{3}{2}}\geq cs\lambda_0^3(1-\delta^2)^{\frac{3}{2}}\geq c_1s,
\end{equation}
for some constant $c_1>0$. From \eqref{eq:Rlarges}, it follows that $C(\lambda)$ can be approximated by the global parametrix $P^{(\infty)}(\lambda)$ given in \eqref{eq:Gp} for $\lambda\in D_1$. Thus, it follows from  \eqref{eq:Gp}, \eqref{eq:Rlarges}, \eqref{eq:Rasylarges}, \eqref{eq:dRasylarges}  and \eqref{eq:phiEst} that 
\begin{equation}\label{eq:I1est}
I_1
=O(e^{-c_2s}),
\end{equation}
as $s\to+\infty$ for some constant $c_2>0$.

Next, we estimate the integral $I_2$.  It can be seen from \eqref{eq:g} and \eqref{eq: phi} that $\phi_+(\lambda)\in i\mathbb{R}$ for $\lambda\in D_2$.
Therefore, using \eqref{eq:Gp}, \eqref{eq:Rlarges}, \eqref{eq:Rasylarges}, \eqref{eq:dRasylarges} and \eqref{eq:hatB}, we have
 \begin{equation}\label{eq:integrant}
(\hat{C}^{-1}\hat{C}')_{12}(\lambda)e^{-2\phi_+(\lambda_0\lambda)}=2\lambda_0\phi_+'(\lambda_0\lambda)-
\frac{i}{\lambda^2-1}\cos\left(2i\phi_+(\lambda_0\lambda)\right)+O\left(\frac{1}{s(1+\lambda^2)}\right),\end{equation}    
where $\lambda\in(-\infty,-1-\delta)\cup(1+\delta,+\infty) $. Therefore, we have  as $s\to+\infty$ 
\begin{equation}\label{eq:intI2}
I_2= I_{2,1}+I_{2,2}+O\left(\frac{1}{s}\right), 
\end{equation} 
where 
\begin{equation}\label{eq:intI21}
I_{2,1}=\frac{1}{2\pi }\int_{D_2}\cos\left(2i\phi_+(\lambda_0\lambda)\right)\sigma(\sqrt{s}\lambda_0\lambda;s)\frac{d\lambda}{\lambda^2-1},
\end{equation}
and 
\begin{equation}\label{eq:intI22}
I_{2,2}=-\frac{\lambda_0}{\pi i }\int_{D_2}\phi_+'(\lambda_0\lambda)\sigma(\sqrt{s}\lambda_0\lambda;s)d\lambda.
\end{equation}
From \eqref{eq:sigma},  we have  
\begin{equation}\label{eq:Estint21}
|I_{2,1}|\leq\frac{1}{2\pi }\int_{D_2}\frac{d\lambda}{\lambda^2-1}.
\end{equation}
Noting that $\phi_+'(\lambda)=\phi_+'(-\lambda)$, we have
\begin{equation}\label{eq:int221}
I_{2,2}=-\frac{2\lambda_0}{\pi i }\int_{1+\delta}^{+\infty}\phi_+'(\lambda_0\lambda)\sigma(\sqrt{s}\lambda_0\lambda;s)d\lambda=-\frac{2}{\pi i } \int_{\lambda_0(1+\delta)}^{+\infty}\phi_+'(\lambda)\sigma(\sqrt{s}\lambda;s)d\lambda.
\end{equation}
This, together with \eqref{eq:Estint21}, implies  
\begin{equation}\label{eq:int221}
I_{2}=-\frac{2}{\pi i } \int_{\lambda_0(1+\delta)}^{+\infty}\phi_+'(\lambda)\sigma(\sqrt{s}\lambda;s)d\lambda+O(1),
\end{equation}
 as $s\to+\infty$.
 
Now we estimate the integral $I_3$. From \eqref{eq:f},  \eqref{eq:Pl},  \eqref{eq: AiryP}, \eqref{eq:Rasylarges} and \eqref{eq:dRasylarges}, we have 
for $\lambda\in(-1-\delta,-1+\delta)$
\begin{equation}\label{eq:intAiry}
(\hat{C}^{-1}\hat{C}')_{12}(\lambda)e^{-2\phi_+(\lambda_0\lambda)}
=2\pi i f'(\lambda)K_{\Ai}(f(\lambda), f(\lambda))+O(1), 
\end{equation}
where $f(\lambda)$ is defined in \eqref{eq:f}. Here $K_{\Ai}$ is the Airy kernel defined in \eqref{eq:KAi} and
\begin{equation}\label{eq:Airykernel}
K_{\Ai}(z,z)=\Ai'(z)^2-z\Ai(z)^2, \quad z\in\mathbb{R}.
\end{equation}
 Using the asymptotic behaviors of $\Ai(z)$ and $\Ai'(z)$ \cite[Equations (9.7.5), (9.7.6), (9.7.9) and (9.7.10)]{O}, we have 
\begin{equation}\label{eq:AirykernelAsy1}
K_{\Ai}(z,z)=O\left(z^{1/2}e^{-\frac{4}{3}z^{3/2}}\right),  \quad z\to+\infty, 
\end{equation}
and 
\begin{equation}\label{eq:AirykernelAsy2}
K_{\Ai}(z,z)=\frac{1}{\pi}|z|^{\frac{1}{2}}+O(z^{-5/2}),  \quad z\to-\infty.
\end{equation}
Therefore,  from  \eqref{eq:intAiry}-\eqref{eq:AirykernelAsy2} and \eqref{eq:f}, we have  
\begin{equation}\label{eq:int5}
\frac{1}{2\pi i}\int_{-1-\delta}^{-1+\delta}(\widehat{C}^{-1}\widehat{C}')_{12}(\lambda)e^{-2\phi_+(\lambda_0\lambda)}\sigma(\sqrt{s}\lambda_0\lambda;s)d\lambda=\frac{1}{\pi i}\int^{-\lambda_0}_{-\lambda_0(1+\delta)} \phi_+'(\lambda)\sigma(\sqrt{s}\lambda;s)d\lambda
+O(1).
\end{equation}
Similarly, we have 
\begin{equation}\label{eq:int6}
\frac{1}{2\pi i}\int_{1-\delta}^{1+\delta}(\widehat{C}^{-1}\widehat{C}')_{12}(\lambda)e^{-2\phi_+(\lambda_0\lambda)}\sigma(\sqrt{s}\lambda_0\lambda;s)d\lambda=\frac{1}{\pi i}\int_{\lambda_0}^{\lambda_0(1+\delta)} \phi'_+(\lambda)\sigma(\sqrt{s}\lambda;s)d\lambda
+O(1).
\end{equation}
From \eqref{eq:intIj}, \eqref{eq:int5}, \eqref{eq:int6},  and the fact that $\phi_{+}'(\lambda)$ and $\sigma(\lambda;s)$ are even functions for $\lambda\in\mathbb{R}$, we have   as $s\to+\infty$
\begin{equation}\label{eq:I3est}
I_3= -\frac{2}{\pi i }\int_{\lambda_0}^{\lambda_0(1+\delta)} \phi'(\lambda)\sigma(\sqrt{s}\lambda;s)d\lambda
+O(1).
\end{equation}

From \eqref{eq:int221} and \eqref{eq:I3est}, we have as $s\to+\infty$
\begin{equation}\label{eq:int3est}
I_2+I_3= -\frac{2}{\pi i}\int_{\lambda_0}^{+\infty} \phi'(\lambda)\sigma(\sqrt{s}\lambda;s)d\lambda
+O(1).
\end{equation}
From \eqref{eq:dphiexp1} and  \eqref{eq:hatsigma3}, we have as $s\to+\infty$
\begin{align}\label{eq:Error}
\int_{\lambda_0}^{+\infty}\phi_+'(\lambda) \widehat{\sigma}(\lambda;s)d\lambda=O(1),
\end{align}
where $\widehat{\sigma}(\lambda;s)$ is defined in \eqref{eq:hatsigma}. 
Therefore, replacing  $\sigma(\sqrt{s}\lambda;s)$  by $\chi_{(-1,1)}(\lambda)$ in \eqref{eq:int3est}, we have 
\begin{equation}\label{eq:I23estExp}
I_2+I_3=-\frac{2}{\pi i}\int_{\lambda_0}^{1}\phi_+'(\lambda)d\lambda+O(1),
\end{equation}
as $s\to+\infty$.

We proceed to evaluate the integral in \eqref{eq:I23estExp}.  From \eqref{eq: phi}, \eqref{eq:derV} and \eqref{eq:dg1}, we obtain the expression of $\phi'_+(\lambda)$ in terms of the following Cauchy principal value integral:
\begin{equation}\label{eq:dphi1}
\phi_+'(\lambda)=\frac{ix \sqrt{s} \lambda }{  \sqrt{\lambda^2-\lambda_0^2}} \left(1-\frac{2\sqrt{s}}{\pi x} P.V. \int_{\lambda_0}^{+\infty} \sigma(\sqrt{s}\zeta;s)\sqrt{\zeta^2-\lambda_0^2}\frac{\zeta d\zeta}{\zeta^2-\lambda^2} \right),
\end{equation}
for $\lambda>\lambda_0$.
Substituting this  into \eqref{eq:I23estExp} and changing the order of integration, we have
\begin{align}\label{eq:dphi21}
\int_{\lambda_0}^{1}\phi_+'(\lambda)d\lambda=& ix\sqrt{s}\sqrt{1-\lambda_0^2}\nonumber\\
&+\frac{2is}{\pi } \int_{\lambda_0}^{+\infty}\zeta\sigma(\sqrt{s}\zeta;s)\sqrt{\zeta^2-\lambda_0^2}\left(P.V. \int_{\lambda_0}^{1} \frac{1 }{  \sqrt{\lambda^2-\lambda_0^2}} \frac{\lambda d\lambda}{\lambda^2-\zeta^2} \right)d\zeta.
\end{align}
Using \eqref{eq:intlarges}, we have  for $\zeta>\lambda_0$
 \begin{equation}\label{eq:dphi4}
P.V. \int_{\lambda_0}^{1} \frac{1 }{  \sqrt{\lambda^2-\lambda_0^2}}\frac{\lambda d\lambda}{\lambda^2-\zeta^2}=\frac{1}{2\sqrt{\zeta^2-\lambda_0^2}}
\log \left( \frac{|\sqrt{\zeta^2-\lambda_0^2}-\sqrt{1-\lambda_0^2}|}{\sqrt{\zeta^2-\lambda_0^2}+\sqrt{1-\lambda_0^2}}\right).
\end{equation}
Next, substituting \eqref{eq:dphi4} into \eqref{eq:dphi21} and replacing $\sigma(\sqrt{s}\zeta;s)$ by $\chi_{(-1,1)}(\zeta)$, we obtain   as $s\to+\infty$
\begin{equation}\label{eq:I23est0}
I_2+I_3=-\frac{2}{\pi i}\left(ix\sqrt{s}\sqrt{1-\lambda_0^2}-\frac{is}{\pi }(1-\lambda_0^2)\right)+O(1).\end{equation}
Using \eqref{eq:lambda0Est}, we have
\begin{equation}\label{eq:I23est}
I_2+I_3=-\frac{x^2}{2}+O(1).\end{equation}

Finally, substituting the  estimates  \eqref{eq:I1est} and  \eqref{eq:I23est} into \eqref{eq:intBs}, we obtain  \eqref{eq:dsD} and complete the proof of this proposition. 
%
\end{proof}
\subsubsection{Integration of the differential identities and proof of Theorem \ref{thm: DAsylargs}
}
To derive the asymptotic expansion for $\log D(x,s)$, we integrate the  differential identities  using the asymptotic approximations of $\partial_x \log D(x,s)$  and  $\partial_s \log D(x,s)$ derived in Propositions \ref{Pro:LDx} and \ref{Pro:LDs}.

We choose a  sufficiently large but fixed value $s_0>0$ and consider the regime $s'\geq s_0$ and $\frac{x'}{\sqrt{s'} }\in [M_2, \frac{2 }{\pi }-M_1s'^{-\frac{2}{3}}]$ with sufficiently large $M_1>0$ and $0<M_2<\frac{2 }{\pi }$. 
The starting point of our analysis is the following integral representation for $\log D(x',s')$ 
\begin{equation}\label{eq:DAsy} 
\log D(x',s')=\log D(M_2\sqrt{s_0},s_0) +\int_{s_0}^{s'} \frac{d}{dt} \left( \log D(M_2\sqrt{t},t)\right) dt +\int_{M_2\sqrt{s'}}^{x'}  \partial_{x}\log D(x,s') dx.
\end{equation}
Here, the first term on the right-hand side of the above equation  is independent of $x'$ and $s'$.  The next task is to 
derive the  asymptotics of the integrals in this equation.

We first derive the asymptotics of the first integral in \eqref{eq:DAsy}.  The integrand of this  integral  can be expressed in the form 
\begin{equation}\label{eq:Dt} 
 \frac{d}{dt} \left( \log D(M_2\sqrt{t},t)\right) =\frac{M_2}{2\sqrt{t}}\partial_{x} \log D(M_2\sqrt{t},t)+\partial_{s}\log D(M_2\sqrt{t},t).
\end{equation}
Substituting \eqref{eq:dxD} and \eqref{eq:dsD} into this equation, we have
\begin{equation}\label{eq:DtAsy} 
 \frac{d}{dt} \left( \log D(M_2\sqrt{t},t)\right) =-\left(M_2^2-\frac{\pi^2}{24}M_2^4\right)t+O(1),
\end{equation}
as $t\to+\infty$.
Integrating the above equation with respect to $t$ over the interval $[s_0,s']$, we obtain the asymptotic expansion of this integral  as $s'\to +\infty$
\begin{equation}\label{eq:DtInt} 
 \int_{s_0}^{s'} \frac{d}{dt} \left( \log D(M_2\sqrt{t},t)\right) dt= -\left(\frac{1}{2}M_2^2-\frac{\pi^2}{48}M_2^4\right)s'^2+O(s'),
\end{equation}
where the error term is independent of $x'$. 

Next, we derive the asymptotics of  the second integral in \eqref{eq:DAsy} using the asymptotic expansion of $\partial_x \log D(x,s)$  given in \eqref{eq:dxD}. Note that the error term in  \eqref{eq:dxD} is uniform for $x$ in the interval of integration.  
Therefore,  from  \eqref{eq:dxD} and \eqref{eq:lambda0Est}, we obtain the asymptotic expansion of this integral as $s'\to+\infty$
\begin{equation}\label{eq:DxInt} 
\int_{M_2\sqrt{s'}}^{x'}  \partial_{x}\log D(x,s') dx=  -\frac{1}{2} s'x'^2+\frac{\pi^2}{48}x'^4 +\left(\frac{1}{2}M_2^2-\frac{\pi^2}{48}M_2^4\right)s'^2+O(s'), 
\end{equation}
 where the error term is uniform for $\frac{x'}{\sqrt{s'} }\in [M_2, \frac{2 }{\pi }-M_1s'^{-\frac{2}{3}}]$.

Substituting \eqref{eq:DtInt}  and \eqref{eq:DxInt} into \eqref{eq:DAsy}, we obtain \eqref{thm:Dasylargs}. This completes the proof of Theorem \ref{thm: DAsylargs}.

\vskip 2cm
\subsection{Proof of Theorem \ref{thm:doubleScaling}}\label{sec:ProofTheorem3}

Tracing back the transformations  \eqref{eq:T}, \eqref{eq:A} and \eqref{eq:B} and using the differential identity \eqref{eq:DY}, 
we have 
\begin{equation}\label{eq:aAsy1}
\partial_{x} \log D(x,s) =-2i\sqrt{s}((R_1)_{11}+g_1),\end{equation}  
where  $g_1$ and $R_1$ are the coefficients of $1/\lambda$ in the large $\lambda$ expansions of $g(\lambda)$ and $R(\lambda)$, respectively. 
Therefore,  from \eqref{eq:gdoubleExpand}, \eqref{eq:Psi1} 
and \eqref{eq:R2}, we have
\begin{equation}\label{eq:dD1}
\partial_{x} \log D(x,s) = -\frac{2}{\pi}\int_{\mathbb{R}}\sigma(\lambda;s)\lambda^2d\lambda+(2\pi)^{1/3}s^{1/6}H(\tau(x,s))+O(s^{-2/3}),
\end{equation}
as $s\to+\infty$, where $\tau(x,s)$ is given in \eqref{eq:tau0}  and the error term is uniform for $L(s)\leq  \frac{\pi}{2} \frac{x}{\sqrt{s}}\leq M$ with any constant $M>1$,  and $L(s)$ is given in \eqref{eq:L}. 
Here, $H(t)$ is the Hamiltonian associated with the Hastings-McLeod solution of the second Painlev\'e equation.
Using the initial value \eqref{thm:AsyD}, \eqref{eq:dtau0} and the fact that $H(t)$ is exponentially small as $t\to+\infty$, we obtain the following expansion after integrating \eqref{eq:dD1} with respect to $x$ 
 \begin{equation}\label{eq:DAsy1}
\log D(x,s) =-\frac{2x}{\pi}\int_{\mathbb{R}}\sigma(\lambda;s)\lambda^2d\lambda+c_0(s)+\int_{+\infty}^{\tau(x,s)}H(\xi)d\xi+O(s^{-1/6}), \end{equation}
 where the error term is uniform for $L(s)\leq  \frac{\pi}{2} \frac{x}{\sqrt{s}}\leq M$ with any constant $M>1$.

Next, we extend \eqref{eq:DAsy1} to the regime $x=\frac{2L(s)}{\pi} s^{\frac{1}{2}}+\frac{y}{(2\pi)^{\frac{1}{3}}}s^{-\frac{1}{6}}$ for $y\in[M',0]$ with any constant $M'<0$ and $L(s)$ defined in \eqref{eq:L}. 
 From \eqref{eq:lambda0Est} and \eqref{eq:dgExpandg1}, we have as $s\to+\infty$
\begin{equation}\label{eq:g1Expand}
g_1=-\frac{2i}{3\pi}s+\frac{is}{\pi}(1-\frac{\pi}{2}\frac{x}{\sqrt{s}})^2
+O(1/s). 
\end{equation}  
Thus,  from \eqref{eq:R11}, \eqref{eq:aAsy1}
 and \eqref{eq:g1Expand},  we have 
 \begin{equation}\label{eq:dD3}
\partial_x \log D(x,s)=-\frac{4}{3\pi}s^{\frac{3}{2}}+(2\pi)^{1/3}s^{\frac{1}{6}}H(\hat{\tau}(x,s))+O(s^{-1/2}),
\end{equation}
where $\hat{\tau}(x,s)$ is  defined in \eqref{eq:tau} and the error term is uniform for $y\in [M', 0]$ with any constant $M'<0$. Integrating the above equation with respect to $x$ from $\frac{2L(s)}{\pi} s^{\frac{1}{2}}$ to $\frac{2L(s)}{\pi} s^{\frac{1}{2}}+\frac{y}{(2\pi)^{\frac{1}{3}}}s^{-\frac{1}{6}}$  and using the asymptotics  given in \eqref{eq:DAsy1}, we obtain as $s\to+\infty$
\begin{equation}\label{eq:DAsy2}
\log D(x,s)=-\frac{2x}{\pi}\int_{\mathbb{R}}\sigma(\lambda;s)\lambda^2d\lambda+c_0(s)+\int_{+\infty}^{y}H(\xi)d\xi+O(s^{-1/6}),
\end{equation} 
where the error term is uniform for $y\in [M', 0]$ with any constant $M'<0$.
From  \eqref{eq:LAsy},  \eqref{eq:DAsy1}, \eqref{eq:DAsy2} and the fact that $H'=-u^2$, 
we obtain \eqref{eq:DasyScaling}. This completes the proof of Theorem \ref{thm:doubleScaling}.

 \subsection{Proof of Theorem \ref{thm:PVasy}}\label{sec:ProofTheorem4}
 
From \eqref{eq:T}, \eqref{eq:R4} and \eqref{eq:RPv}, we have as $s\to+\infty$
\begin{equation}\label{eq:hatPsiAsy} 
Y(\sqrt{s}\lambda)=\left(I+O\left(\frac{1}{ s^{2} \lambda}\right)\right)\Phi^{(\PV)}(\lambda,\sqrt{s}x)e^{-i\sqrt{s}x\lambda\sigma_3},\end{equation}
where  the error term is uniform for $|\lambda|>1+\delta$ and $\sqrt{s}x\in[0,M]$ for any constant $M>0$. From the differential identity \eqref{eq:DY}, \eqref{eq:PsiVInfinity},  \eqref{eq:hatPsiAsy} and  Remark \ref{rem:PVPro}, we obtain as $s\to+\infty$
\begin{equation}\label{eq:dD}
\partial_x\log D(x,s)=\frac{v(\sqrt{s}x)}{x}+O(1/s^2),\end{equation}
where $v(\tau)$ is the solution of the  $\sigma$-form of the fifth Painlev\'e equation
\eqref{eq:sigmaPV} with the boundary conditions \eqref{eq:PVasy}.
It follows from \eqref{eq:PVasy} that the function $ \tau\mapsto\frac{v(\tau)}{\tau}$ is integrable on the interval $[0,\sqrt{s}x]$. 
Integrating \eqref{eq:dD} with respect to $x$ and using $D(0,s)=1$,   we obtain \eqref{eq:DPV}. The equation \eqref{eq:DPV1} follows from \eqref{eq:DPV} and the integral representation of the classical sine kernel determinant obtained in \cite{JMMS}.
This completes the proof of Theorem \ref{thm:PVasy}.

\section*{Acknowledgements} 
The work of Shuai-Xia Xu was supported in part by the National Natural Science Foundation of China under grant numbers  12431008, 12371257 and 11971492, and by  Guangdong Basic and Applied Basic Research Foundation (Grant No. 2022B1515020063). 

\section*{Data availability statement} 
No datasets were generated or analyzed during the current study.

\section*{Declarations} 
\paragraph{Conflict of interest} The author has no competing interests to declare that are relevant to the content of this
article.

\end{document}